\documentclass[12pt]{iopart}
\usepackage{iopams}
\usepackage{cite}
\usepackage{graphicx}
\usepackage{color}

\usepackage[colorlinks=true,linkcolor=black,citecolor=black,urlcolor=black]{hyperref}

\newtheorem{theorem}{\bf Theorem}[section]
\newtheorem{definition}{\bf Definition}[section]

\newtheorem{condition}{\bf Hypothesis}

\newtheorem{lemma}{\bf Lemma}[section]
\newtheorem{remark}{\bf Remark}[section]
\newtheorem{notation}{\bf Notation}
\newtheorem{example}{\bf Example}

\newtheorem{stochastic}{\bf Stochastic system}
\newtheorem{model}{\bf Model}

\newenvironment{proof}[1][Proof.]{
  \begin{trivlist}
    \item[\hskip \labelsep {\bfseries #1}]
  \ignorespaces
  }{%
    \hspace*{\fill} $\square$\par
    \end{trivlist}
}

\begin{document}

\title[Stochastic interacting wave functions]{System of stochastic interacting wave functions that model quantum measurements}

%

\author{Carlos M. Mora$^{1}$$^{*}$}

\address{$^*$ Corresponding author.}

\address{$^1$ Departamento de Ingenier\'{\i}a Matem\'{a}tica, 
Facultad de Ciencias F\'{\i}sicas y Matem\'{a}ticas, Universidad de Concepci\'{o}n,
Casilla 160 C, Concepci\'{o}n, Chile.}

\eads{\mailto{cmora@ing-mat.udec.cl}}

\begin{abstract}
We develop a system of non-linear stochastic evolution equations that
describes the continuous  measurements of quantum systems with mixed initial state.
We address quantum systems  with unbounded Hamiltonians and unbounded interaction operators.
Using  arguments of  the theory of quantum measurements
we derive a system of stochastic interacting wave functions (SIWF for short) 
that models the continuous monitoring of quantum systems.
We prove the existence and uniqueness of the solution to this system 
under conditions general enough for the applications. 
We obtain that the mixed state generated by the SIWF at any time does not depend on the initial state,
and satisfies the diffusive stochastic quantum master equation,
which is also known as Belavkin equation.
We present two physical examples.
In one, the SIWF becomes a system of non-linear stochastic partial differential equations.
In the other, we deal with a model of a circuit quantum electrodynamics. \\

\noindent{\it Keywords\/}: Quantum measurement process, non-linear stochastic evolution equation,
stochastic master equation, Belavkin equation, open quantum system, unbounded generator,
random operators
\end{abstract}

\ams{81P15, 81S22, 60H15, 60H25, 60H30, 81P16}

\submitto{\NL}



\section{Introduction}


We address the modeling of continuous measurements of quantum systems. 
The states of quantum systems are described by density operators that 
are non-negative operators with unit trace on a separable complex Hilbert space $\left(\mathfrak{h},\left\langle \cdot,\cdot\right\rangle \right) $,
where the  inner product 
$\left\langle \cdot,\cdot \right\rangle:   \mathfrak{h} \times  \mathfrak{h} \rightarrow \mathbb{C}$ 
is anti-linear in the first variable and linear in the second argument. 
In this paper, we study the continuous monitoring of quantum systems with infinite-dimensional state space $\mathfrak{h}$.
This occurs, for instance,  
in quantum systems involving continuous variables like position,
where $\mathfrak{h}$ is a tensor product space containing $L^2 \left( \mathbb{R}^d, \mathbb{C} \right) $ among its components
(see, e.g., \cite{Cohen-Tannoudji,Gough,Sakurai} and Section \ref{sec:Ex1}).
Another example arises from  the  Rabi and Jaynes-Cummings models,
where the coupling between a  two-level atom and  a  quantized electromagnetic field of a cavity
is described by linear operators in 
$\mathfrak{h} = \ell^2 \left( \mathbb{Z}_+ \right) \otimes\mathbb{C}^2$  
(see, e.g., \cite{HarocheRaimond2006,Niemczyk2010} and Section  \ref{sec:Ex2}).

\subsection{Framework}

The evolution of a closed quantum system with a pure initial state 
$ \left\vert   \psi \left( 0 \right)  \rangle\langle  \psi \left( 0 \right) \right\vert  $
is given by $ \left\vert   \psi \left( t \right)  \rangle\langle  \psi \left( t \right) \right\vert  $,
where $\psi \left( t \right)$ satisfies the Schr\"{o}dinger equation
\begin{equation}
 \label{eq:DeterministicSchrodinger}
\mathrm{i}  \frac{d}{ d \, t} \psi \left( t \right) 
 =
H\left( t \right) \psi \left( t \right) 
\qquad \quad \mathrm{for \; all \,} t \geq 0 
\end{equation}
(see, e.g., \cite{Cohen-Tannoudji,Gough,Sakurai}).
Here,  $ \psi \left( t \right)  \in \mathfrak{h} $,
the time-dependent Hamiltonian $H\left( t \right)$  is a symmetric operator in $\mathfrak{h}$,
and following Dirac's notation we write 
$ \left\vert   \phi  \rangle\langle  \varphi \right\vert  $
for the operator
\[
\left\vert   \phi  \rangle\langle  \varphi \right\vert \left( \xi \right)
=
\langle  \varphi , \xi  \rangle  \phi 
\qquad
\qquad \quad \mathrm{for \; all \,} \xi \in \mathfrak{h} 
\]
whenever $ \phi, \varphi \in \mathfrak{h}$.

Assume that $\dim \left( \mathfrak{h} \right) = + \infty $ and that $H\left( t \right)$ is an unbounded operator.
If $H\left( t \right) $ does not depended on $t$,
then the existence and uniqueness of the solution to  (\ref{eq:DeterministicSchrodinger}) 
is guaranteed by  the self-adjointness of the Hamiltonian (see, e.g., \cite{ReedSimonVol2}).
In the non-autonomous case,
there exist a unique global solution to (\ref{eq:DeterministicSchrodinger}) 
when the commutator $\left[ H\left( t \right)  , C \right] C^{-1}$
is regular for  certain positive self-adjoint operator $C$ 
(see, e.g., \cite{MASPERO2017721} and references therein).
The well-posedness of (\ref{eq:DeterministicSchrodinger}) has also been established in  the theory of Floquet operators (see, e.g., \cite{Liang20224850} and references therein).

Let the initial state  $\varrho_{0}$ be mixed
(that is, $\varrho_{0}
=
\sum_{n = 1}^{\infty} p_n \left\vert\psi^n_0  \rangle\langle\psi^n_0  \right\vert 
$
with $ p_n \geq 0$, $\sum_{n = 1}^{\infty} p_n = 1 $ and $\left\Vert \psi^n_0 \right\Vert = 1 $),
then
the evolution of the closed quantum system in the  Schr\"{o}dinger picture is described by the von Neumann equation
\begin{equation}
\label{eq:vonNeumann}
\frac{d}{dt} \varrho_{t} = - \mathrm{i} \left( H\left( t \right) \varrho_{t}  - \varrho_{t}  H\left( t \right)\right) ,
\end{equation}
where $\varrho_{t}$ is a density operator in $\mathfrak{h}$.
If there exist a unique unitary propagator for (\ref{eq:DeterministicSchrodinger}),
then
\[
\varrho_{t} 
=
\sum_{n = 1}^{\infty} p_n \left\vert\psi^n_t  \rangle\langle\psi^n_t  \right\vert  ,
\]
where $\psi^n_t$ is the solution of (\ref{eq:DeterministicSchrodinger}) with initial condition $\psi^n_0 $.

In practice, 
the  vast majority of quantum systems are affected by their environment (see, e.g., \cite{HarocheRaimond2006}),
and therefore they are open.
In the  Schr\"{o}dinger picture,
the evolution of many open quantum systems is described by 
the quantum master equation in  the Gorini-Kossakowski-Sudarshan-Lindblad (GKSL) form 
\begin{equation}
\label{eq:GKSL}
\varrho_{t} =\varrho_0 + \int_{0}^{t}
\left( G \left( s \right) \varrho_{s} +\rho_{s} G \left( s \right)^{\ast}
+\sum_{\ell=1}^{\infty } L_{\ell} \left( s \right) \varrho_{s} L_{\ell}\left( s \right)^{\ast}\right) ds ,
\end{equation}
where
$t \geq 0$, 
$\varrho_{t}$ is a density operator in $\mathfrak{h}$,
the operators $L_{\ell} \left( t \right) $ in $\mathfrak{h}$ represent 
the interaction between the environment and the  quantum system with Hamiltonian  $H\left( t \right)$, 
and
$G \left( t \right)$ is a linear operator in $\mathfrak{h}$ satisfying
\begin{equation}
\label{eq:Def_G}
G\left( t \right)=-  \mathrm{i} H\left( t \right)-\frac{1}{2} \sum_{\ell=1}^{\infty}L_{\ell}\left( t \right)^{\ast}L_{\ell} \left( t \right) 
\end{equation}
 on appropriate domain
(see, e.g., \cite{AlickiLendi2007,BreuerPetruccione}).
The existence and uniqueness of the solution of the linear operator equation (\ref{eq:GKSL}) is proved, for instance, in  \cite{ChebFagn2,Chebotarev2000,Fagnola,FagMora2019,MoraAP}.

The evolution of the open quantum system (\ref{eq:GKSL}) is also described by 
the non-linear stochastic Schr\"odinger equation
\begin{eqnarray}
\fl
\label{eq:NLSSE_i}
& \hat{\phi}_t
 =
\hat{\phi}_0 
+ \sum_{ \ell =1}^{ \infty} \int_0^t \left(
L_{\ell} \left( s \right)    \hat{\phi}_s 
- \Re \left(  \langle \hat{\phi}_s,  L_{\ell} \left( s \right)   \hat{\phi}_s \rangle \right)  \hat{\phi}_s
\right) dW_s^{\ell} 
\\
\nonumber
\fl
& \quad
+ \hspace{-2pt} \int_0^t \hspace{-2pt} \left(  \hspace{-2pt} G\left( s \right)  \hat{\phi}_s 
+ \sum_{ \ell =1}^{ \infty} \left( \Re \left( \langle \hat{\phi}_s,  L_{\ell} \left( s \right)  \hat{\phi}_s \rangle  \right) 
 L_{\ell} \left( s \right)  \hat{\phi}_s
-\frac{1}{2} \left( \Re \left(\langle \hat{\phi}_s, L_{\ell} \left( s \right)  \hat{\phi}_s \rangle  \right) \right)^2  \hat{\phi}_s  \right)
 \hspace{-2pt} \right) \hspace{-1pt} ds 
\end{eqnarray}
(see, e.g., \cite{BarchielliBelavkin1991,Belavkin1,Ghirardi1990, Kupsch1996}),
where $t \geq 0$, $\hat{\phi}_t \in \mathfrak{h}$,
the $\mathfrak{h}$-valued random variable $\hat{\phi}_0$ satisfies $\left\Vert \hat{\phi}_0 \right\Vert =1$, 
and
$W^1, W^2, \ldots$ are independent real Wiener processes (which are also called Brownian motions).
From now on,
$\Re \left( z \right)$ stands for the real value of the complex number $z$,
and so
$
\Re \left( \langle \hat{\phi}_s,  L_{\ell} \left( s \right)  \hat{\phi}_s \rangle \right)
=
\langle \hat{\phi}_s,  \left( L_{\ell} \left( s \right) + L_{\ell} \left( s \right)^{\ast} \right)  \hat{\phi}_s \rangle
$
provided that $ \hat{\phi}_s$ belongs to the domain of $L_{\ell} \left( s \right)^{\ast} $.
We can choose $\hat{\phi}_0 $ such that $ \rho_{0} = \mathbb{E} \left( \left\vert \hat{\phi}_0   \rangle\langle \hat{\phi}_0   \right\vert \right) $,
and hence
\begin{equation}
\label{eq:Rep_GKSL}
 \varrho_{t} = \mathbb{E} \left( \left\vert \hat{\phi}_t  \rangle\langle \hat{\phi}_t  \right\vert \right) 
\end{equation}
(see, e.g., \cite{Percival,MoraAP}).
Using the unravelling (\ref{eq:Rep_GKSL}) we compute  the mean value of the quantum observables
by means of the numerical solution of (\ref{eq:NLSSE_i}),
which overcomes the difficulties arising from 
the numerical solution of  (\ref{eq:GKSL})
(see, e.g,  \cite{BreuerPetruccione,MoraAAP2005,Percival}).
Under a quantum non-explosion  condition,
(\ref{eq:NLSSE_i})  has a unique weak probabilistic solution (see, e.g., \cite{FagMora2013,MoraReAAP}),
which means that there exist an adapted process $\hat{\phi}_t$ and Brownian motions $W^1_t, W^2_t, \ldots$ 
such that  (\ref{eq:NLSSE_i}) holds, 
and the joint probability distribution of any $\hat{\phi}_t, W^1_t, W^2_t, \ldots$ satisfying (\ref{eq:NLSSE_i}) 
is unique. 
It is worth pointing out that (\ref{eq:NLSSE_i}) is different from
the stochastic non-linear  Schr\"odinger equation that 
arises from adding random effects to the classical non-linear  Schr\"odinger equation
(see, e.g., \cite{BarrueDebusscheTusseau} and references therein).

The state $ \left\vert \hat{\phi}_t  \rangle\langle \hat{\phi}_t  \right\vert $,
where  $\hat{\phi}_t$ satisfies the non-linear stochastic evolution equation (\ref{eq:NLSSE_i})
or its versions,
describes 
continuous monitoring  of quantum systems with an initial random pure state  
(see, e.g.,  \cite{BarchielliBelavkin1991,Barchielli,BarchielliHolevo1995,BreuerPetruccione}).
In case the initial state $ \rho_{0} $ is mixed,
the evolution of quantum systems undergoing continuous measurements is governed by 
the stochastic master equations 
(see, e.g., \cite{Barchielli,Pellegrini2010AHP,WisemanMilburn2010} and references therein),
which are also known as Belavkin's equations
(see, e.g., \cite{Belavkin2013} and references therein).
In this paper,
we focus on measurement processes described by  the diffusive stochastic master equation
\begin{eqnarray}
\fl
\nonumber
\rho_{t}
 & = 
 \rho_{0}
 +
 \int_0^t \left(
 G \left( s \right) \rho_{s}
 +  \rho_{s}   G\left( s \right)^{\ast} 
 + \sum_{\ell=1}^{\infty} L_{ \ell} \left( s \right) \rho_{s}  L_{\ell}\left( s \right)^{\ast}
 \right) ds
 \\
 \fl
 \label{eq:SQME}
 & \quad
 + 
 \sum_{\ell=1}^{\infty} 
  \int_0^t \left( L_{ \ell} \left( s \right) \rho_{s}  +  \rho_{s}  L_{\ell}\left( s \right)^{\ast}
  - 2 \Re \left( \Tr \left( L_{ \ell} \left( s \right) \rho_{s}  \right) \right) \rho_{s}
  \right) d W^{\ell}_s 
 \qquad \mathrm{for}\; \mathrm{all}\;  t \geq 0 ,
\end{eqnarray}
where
$ \rho_{t}$ is a random density operator on $\mathfrak{h}$,
and
$W^1, W^2, \ldots$ are independent real Brownian motions.
In the Belavkin equation (\ref{eq:SQME}),
$ \rho_{t}$ represents the quantum state describing the interaction of the quantum system 
(\ref{eq:vonNeumann})  with the apparatus 
$ L_{1}, L_2, \ldots $,
and the  outcome of the continuous monitoring of the observable  $ L_{\ell} $ is given by 
\begin{equation}
\label{eq:QFiltering}
B_{t}^{\ell}
=
W^{\ell}_t + 2  \int_0^t \Re \left( \Tr \left( L_{ \ell} \left( s \right) \rho_{s}  \right) \right) ds .
\end{equation}
In Section \ref{PhysicalBackground} we present the derivation of  (\ref{eq:SQME})
given by \cite{Barchielli} in the context of the Schr\"odinger picture  (see also, e.g., \cite{BarchielliGregoratti2013,WisemanMilburn2010}).
Starting from the Heisenberg picture,
the stochastic master equation (\ref{eq:SQME}), and its versions,
is obtained in the quantum filtering theory (see, e.g., \cite{BarchielliGregoratti2013,Belavkin2013,Bouten2007}),
where 
the $W^{\ell}_t $'s are the innovation processes.

The Belavkin equation (\ref{eq:SQME}) has a unique solution
when  $\mathfrak{h} $ is finite-dimensional
(see, e.g., \cite{Barchielli,Pellegrini2008}).
If
$G \left( s \right)$ and $L_{ \ell} \left( s \right)$ are bounded operator,
then normalizing the solution of the  linear stochastic quantum master equation 
yields a solution to the non-linear stochastic evolution equation  (\ref{eq:SQME})
(see, e.g., \cite{BarchielliHolevo1995,BarchielliPaganoniZucca} and Section \ref{PhysicalBackground}).
To the best of our knowledge,
it has not been proven  the existence and uniqueness of the solution of (\ref{eq:SQME})
in the case where $G \left( s \right)$ or some $L_{ \ell} \left( s \right)$ are unbounded operators 
like in Sections  \ref{sec:Ex1} and  \ref{sec:Ex2}.
A formal calculation, which is rigorous in case $ \dim \left( \mathfrak{h} \right) < + \infty$,
shows that the mean statistical operator $ \mathbb{E} \rho_{t} $ satisfies the GKSL equation (\ref{eq:GKSL}).

The stochastic Schr\"odinger equation (\ref{eq:NLSSE_i}) is connected with the linear stochastic Schr\"odinger equation
\begin{equation}
\label{eq:Linear_SSE}
\fl
\phi_{t} \left(\phi_{0} \right)
= \phi_{0}  +\int_{0}^{t}G \left( s \right) \phi_{s}  \left( \phi_{0} \right)   \, ds 
+ \sum_{\ell=1}^{\infty }\int_{0}^{t} L_\ell \left( s \right) \phi_{s}  \left( \phi_{0} \right)  \, dB_{s}^{\ell} 
\qquad \forall t \geq 0 ,
\end{equation}
where $\phi_{t} \left(\phi_{0} \right)$ is an $\mathfrak{h}$-valued adapted stochastic process and 
$B^1, B^2, \ldots$ are independent real Brownian motions on the filtered complete probability space 
$\left( \Omega,\mathfrak{F},\left( \mathfrak{F}_{t}\right) _{t \in \mathbb{I}},\mathbb{Q}\right) $.
Under the probability measure $ \mathbb{P} = \left\Vert \phi_{T} \left( \hat{\phi}_0  \right) \right\Vert ^2 \cdot \mathbb{Q} $
we have that 
$\phi_{t} \left( \hat{\phi}_0  \right) / \left\Vert \phi_{t} \left( \hat{\phi}_0  \right) \right\Vert $
and 
\[
W^{\ell}_t  
= 
B_{t}^{\ell} 
- 
\int_0^t 
\frac{ 2 \, \Re \langle  \phi_{s} \left( \hat{\phi}_0  \right)  ,  L_{\ell}\left( s \right)  \phi_{s} \left( \hat{\phi}_0  \right)  \rangle 
}{ \left\Vert \phi_{s} \left( \hat{\phi}_0  \right) \right\Vert ^2 } 
ds
\]
satisfy (\ref{eq:NLSSE_i}) in the interval $\left[ 0, T \right]$
(see, e.g., \cite{FagMora2013,MoraReAAP} for details).
Conversely, we can derive (\ref{eq:Linear_SSE}) from (\ref{eq:NLSSE_i}) 
(see, e..g., Section 2.5.4 of \cite{BarchielliGregoratti2013} and the proof of Proposition 2 of  \cite{MoraReAAP}).
In \cite{Holevo} is established the weak topological solution of (\ref{eq:Linear_SSE}), and the regular strong topological solution of (\ref{eq:Linear_SSE}) is studied by, e.g., \cite{FagMora2013,MoraMC2004,MoraReIDAQP}.

\subsection{Contributions of the paper}

We address the continuous measurements of infinite-dimensional quantum systems 
(i.e, $ \dim \left( \mathfrak{h} \right) = + \infty $)
with the property that 
the Hamiltonian $H \left( t \right)$ and the apparatus $L_{ \ell} \left( t \right)$
can be unbounded operators.
In particular, 
we develop the following model describing 
the evolution of quantum systems undergoing continuous monitoring.

\begin{model}
	\label{Modelo3}
	Suppose that the initial density operator $\rho_{0}$ is equal to the operator 
	$
	\sum_{n = 1}^{N} p_n \left\vert  \hat{\phi}^n_0   \rangle\langle  \hat{\phi}^n_0   \right\vert
	$,
	where $N \in \mathbb{N}$ or $N = + \infty$,
	the $p_n's$ are non-negative real numbers satisfying 
	$\sum_{n = 1}^{N} p_n = 1$,
	and the  $\hat{\phi}^n_0$'s are  unit vectors in $\mathfrak{h}$.
	Then, 
	the state of the quantum system at time $t$ is 
	\begin{equation}
		\label{eq:Def_rhom}
		\rho_{t} = \sum_{n = 1}^{N} p_n \, \left\vert \check{\psi}^n_t  \rangle\langle \check{\psi}^n_t  \right\vert ,
	\end{equation}
	where the $\check{\psi}^n_t$'s are $\mathfrak{h}$-valued adapted process with continuous sample paths
	satisfying the system of non-linear stochastic evolution equations on $\mathfrak{h}$:
	\begin{eqnarray}
		\nonumber
		\check{\psi}^n_t 
		& =
		\hat{\phi}^n_0 
		+ \int_0^t\left( G\left( s \right) \check{\psi}^n_s
		+ \sum_{ \ell =1}^{ \infty} \left( \wp_{\ell} \left( s \right)  L_{\ell} \left( s \right)  \check{\psi}^n_s
		-\frac{1}{2} \wp_{\ell} \left( s  \right)^2  \check{\psi}^n_s \right) 
		\right)ds 
		\\
		& \quad
		+
		\sum_{ \ell =1}^{ \infty} \int_0^t  
		\left( L_{\ell} \left( s \right) \check{\psi}^n_s - \wp_{\ell} \left( s \right) \check{\psi}^n_s  \right)  dW_s^{\ell} ,
		\label{eq:SIWEm}
	\end{eqnarray}
	where
	$W^1, W^2, \ldots$ are independent real Brownian motions,
	and
	\begin{equation}
		\label{eq:Def_pl_m}
		\wp_{\ell} \left( s  \right)
		=
		\sum_{n=1}^{N} p_n \, \Re \left( \langle  \check{\psi}^n_s ,  L_{\ell} \left( s \right)  \check{\psi}^n_s \rangle \right) .
	\end{equation}
	
\end{model}

There is no loss of generality in assuming $ N = + \infty $ in Model \ref{Modelo3},
because we can set $p_{k} = 0$ for all $ k \geq N+1 $ in case $N \in \mathbb{N}$.
On the hand,
the initial mixed state $\rho_0$  is random in many physical situations.
For instance, $\rho_0$ can be the partial trace of a quantum state interacting with a fluctuating environment.
This leads us to consider that 
$ p_1, p_2, \ldots  $ and $\hat{\phi}^1_0, \hat{\phi}^2_0, \ldots $ are real random variables and $\mathfrak{h}$-valued random variables respectively. 
In order to treat the case where some $p_n$ are equal to $0$ with positive probability 
we define 
$  \psi^n_t = \sqrt{p_n}  	\check{\psi}^n_t   $,
and so Model \ref{Modelo3} with a random initial density operator
becomes Model \ref{Modelo1} given below.

\begin{stochastic}
\label{st:SIWE}
 For any  $ n \in \mathbb{N} $ we have
\begin{eqnarray}
\nonumber
\psi^n_t 
& =
\psi^n_0 
+ \int_0^t\left( G\left( s \right) \psi^n_s
+ \sum_{ \ell =1}^{ \infty} \left( \wp_{\ell} \left( s \right)  L_{\ell} \left( s \right)  \psi^n_s
-\frac{1}{2} \wp_{\ell} \left( s  \right)^2  \psi^n_s \right) 
\right)ds 
\\
& \quad
+
\sum_{ \ell =1}^{ \infty} \int_0^t  
\left( L_{\ell} \left( s \right) \psi^n_s - \wp_{\ell} \left( s \right) \psi^n_s  \right)  dW_s^{\ell} ,
 \label{eq:SIWE}
\end{eqnarray}
where $ \psi^1_t, \psi^2_t, \ldots $ are $\mathfrak{h}$-valued adapted process with continuous sample paths,
\begin{equation}
 \label{eq:Def_pl}
 \wp_{\ell} \left( s  \right)
=
\sum_{n=1}^{\infty} \Re \left( \langle  \psi^n_s ,  L_{\ell} \left( s \right)  \psi^n_s \rangle \right) 
\end{equation}
and
$W^1, W^2, \ldots$ are independent real Brownian motions on the filtered complete probability space 
$\left( \Omega,\mathfrak{F},\left( \mathfrak{F}_{t}\right) _{t \in \mathbb{I}},\mathbb{P}\right) $.
\end{stochastic}

\begin{model}
	\label{Modelo1}
	Suppose that the initial random density operator 
	is distributed according to the probability measure $\mu$
	defined on $\mathcal{B} \left(  \mathfrak{L}_{1} \left( \mathfrak{h}\right) \right)$,
	where $\mathcal{B} \left(  \mathfrak{L}_{1} \left( \mathfrak{h}\right) \right)$ is 
	the collection of all Borel sets of  
	the space of all trace class operators on $\mathfrak{h}$.
	Let $\left(  \psi^n_t  \right)_{n \in \mathbb{R}}$ satisfy the system of non-linear stochastic evolution equations on $\mathfrak{h}$ given by the stochastic system \ref{st:SIWE} with
	$  \psi^n_0 = \sqrt{p_n}  \hat{\phi}^n_0  $, where:
	\begin{itemize}
		\item $p_1, p_2, \ldots$ are $\mathfrak{F}_{0}$-measurable real-value random variables  such that 
		$p_n \geq 0$ and  $\sum_{n = 1}^{\infty} p_n = 1$.
		
		\item $ \hat{\phi}^1_0,  \hat{\phi}^2_0, \ldots $ are $\mathfrak{F}_{0}$-measurable $\mathfrak{h}$-value random variables 
		such that  $\left\Vert  \hat{\phi}^n_0 \right\Vert = 1$ for all $n \in \mathbb{N}$.
		
		\item 
		$
		\sum_{n = 1}^{\infty} p_n \left\vert  \hat{\phi}^n_0   \rangle\langle  \hat{\phi}^n_0   \right\vert 
		$
		is distributed according to $\mu$, i.e., 
		for any $A \in \mathcal{B} \left(  \mathfrak{L}_{1} \left( \mathfrak{h}\right) \right)$ we have
		$
		\mathbb{P} \left( \sum_{n = 1}^{\infty} p_n \left\vert  \hat{\phi}^n_0   \rangle\langle  \hat{\phi}^n_0   \right\vert \in A \right)
		=
		\mu \left( A  \right)
		$.
	\end{itemize}
	Then, 
	the state of the quantum system at time $t$ is 
	\begin{equation}
		\label{eq:Def_rho}
		\rho_{t} = \sum_{n = 1}^{\infty} \left\vert \psi^n_t  \rangle\langle\psi^n_t  \right\vert .
	\end{equation}
\end{model}

According to Section \ref{sec:DerivationModel} we have that 
$\psi^n_t \equiv 0$ for  all $n \geq N+1$ 
whenever $0 = p_{N+1} = p_{N+2} = \cdots$.
Using the substitution $ \check{\psi}^n_t = \psi^n_t / \sqrt{p_n} $
we translate Model \ref{Modelo1}  into Model \ref{Modelo3}
(with a random initial density operator)
in the case  that $ p_n > 0 $ for all $n \in \mathbb{N}$ or  $ p_n > 0 $ for all $n \leq N$.
If $\rho_0$ is a random pure state (i.e., $  \rho_{0} $  and $\left\vert  \hat{\phi}_0   \rangle\langle  \hat{\phi}_0   \right\vert $ have the same distribution),
then the quantum state (\ref{eq:Def_rho}) becomes $\left\vert \hat{\phi}_t  \rangle\langle \hat{\phi}_t  \right\vert $,
where $\hat{\phi}_t$ is given by (\ref{eq:NLSSE_i});
in other words,  $\psi^1_t = \hat{\phi}_t $ and $\psi^n_t = 0$ for any $n \geq 2$.

In Section \ref{sec:EyU_solutions} 
we study the stochastic system \ref{st:SIWE}.
Under assumptions used to guarantee the existence and uniqueness of the solution to (\ref{eq:NLSSE_i}),
we first prove the existence and uniqueness of the solution of the stochastic system \ref{st:SIWE}
whenever  $\psi^1_0, \psi^2_0, \ldots$ are given regular $\mathfrak{h}$-random variables. 
For this purpose, we find that the sequence 
$\left( \psi^n_t  \right)_{n \in \mathbb{N}}$ satisfies
the following stochastic evolution equation on  $L^2 \left( \mathbb{N} ;  \mathfrak{h} \right)$:
\begin{eqnarray}
		\fl
		\label{eq:NLSSE_e}
		& \hat{x}_t
		=
		\hat{x}_0+ \sum_{ \ell =1}^{ \infty} \int_0^t \left(
		\widetilde{ L_{\ell} \left( s \right) }   \hat{x}_s 
		- \Re \left(  \langle \hat{x}_s, \widetilde{ L_{\ell} \left( s \right) }  \hat{x}_s \rangle \right)  \hat{x}_s
		\right) dW_s^k
		\\
		\fl
		& 
		\nonumber
		\  +
		\int_0^t\left( 
		\widetilde{ G\left( s \right) } \hat{x}_s
		+
		\sum_{ \ell =1}^{ \infty} \left( \Re \left( \langle \hat{x}_s, \widetilde{ L_{\ell} \left( s \right) } \hat{x}_s \rangle  \right) 
		\widetilde{ L_{\ell} \left( s \right) } \hat{x}_s
		-\frac{1}{2} \left( \Re \left(\langle \hat{x}_s, \widetilde{ L_{\ell} \left( s \right) } \hat{x}_s \rangle  \right) \right)^2 \hat{x}_s \right) 
		\right)ds ,
\end{eqnarray}
where  
$t \geq 0$,
$L^2 \left( \mathbb{N} ;  \mathfrak{h} \right)$ denotes the Hilbert space of all square-summable sequences of elements of $ \mathfrak{h}$,
$\hat{x}_t$ takes values in  $L^2 \left( \mathbb{N} ;  \mathfrak{h} \right)$,
$\hat{x}_0$ is distributed according to the law of $\left( \psi^n_0  \right)_{n \in \mathbb{N}}$,
and to each operator 
$A: \mathcal{D} \left( A \right) \subset  \mathfrak{h} \rightarrow \mathfrak{h} $ 
we associate the operator  
$\widetilde{A} : \mathcal{D} \left( \widetilde{A} \right) \subset   L^2 \left( \mathbb{N} ;  \mathfrak{h} \right)  
\rightarrow L^2 \left( \mathbb{N} ;  \mathfrak{h} \right)$
given by
\[
\left(  \widetilde{A} f \right)  \left(  n \right)  =  A f  \left(  n \right)  
\qquad  \mathrm{for}\ \mathrm{all}\;  n \in  \mathbb{N} 
\]
(see Section \ref{sec:L2Nh} for details).
Since (\ref{eq:NLSSE_e}) is a non-linear stochastic Schr\"odinger equation of type (\ref{eq:NLSSE_i}),
but with values in $L^2 \left( \mathbb{N} ;  \mathfrak{h} \right)$,
we use the results established by \cite{FagMora2013,MoraReAAP}
to deduce that the stochastic system (\ref{st:SIWE}) 
has a unique weak (in the probabilistic sense) solution 
(see, Theorem \ref{th:EyU-SIWE} for details).

In Section \ref{sec:Measurements} we deduce that the density operator given by Model \ref{Modelo1}, and so by Model \ref{Modelo3}, describes the evolution of the state of continuous quantum measurement processes.
More precisely,  
from the derivation of the diffusive stochastic master equation  (\ref{eq:SQME})
given by Section 3.1 of \cite{Barchielli}
we obtain that the state of the quantum system at time $t$ is 
\begin{equation}
	\label{eq:Int_Rho_t}
	\rho_t
	=
	\sum_{n = 1}^{\infty} 
	\left\vert
	\frac{ \phi_{t} \left(  \sqrt{p_n}  \hat{\phi}_{0}^n  \right) }{ 
		\sqrt{  \sum_{k=1}^{\infty} \left\Vert \phi_{t} \left(  \sqrt{p_k}  \hat{\phi}_{0}^k  \right)   \right\Vert ^{2}}}
	\rangle \langle  
	\frac{ \phi_{t} \left(  \sqrt{p_n}  \hat{\phi}_{0}^n  \right) }{ 
		\sqrt{  \sum_{k=1}^{\infty} \left\Vert \phi_{t} \left(  \sqrt{p_k}  \hat{\phi}_{0}^k  \right)   \right\Vert ^{2}}}
	\right\vert ,
\end{equation}
where  $\sqrt{p_n}  \hat{\phi}_{0}^n $ are regular enough,
and $ \phi_{t} \left(  \cdot  \right) $ satisfies 
the linear stochastic Schr\"odinger equation (\ref{eq:Linear_SSE})
with $B^{\ell}_t$ being the outcome of the observable  $ L_{\ell} $ as in (\ref{eq:QFiltering}).
We have that 
$\left( \phi_{t} \left(  \sqrt{p_n}  \hat{\phi}_{0}^n   \right) \right)_{n \in \mathbb{N}} $
satisfies the stochastic evolution equation 
\begin{equation}
	\label{eq:LSSE_e}
	x_t
	=
	x_0+\int_0^t \widetilde{ G\left( s \right) } \, x_s \, ds 
	+
	\sum_{ \ell =1}^{ \infty} \int_0^t   \widetilde{ L_{\ell} \left( s \right) } \, x_s \, dB_s^{\ell} 
	\qquad 
	\mathrm{for} \ \mathrm{all} \; t \geq 0,
\end{equation}
where $x_{t} $ is a $L^2 \left( \mathbb{N} ;  \mathfrak{h} \right)$-valued adapted stochastic process  with continuous sample paths.
Since (\ref{eq:LSSE_e}) is a linear stochastic Schr\"odinger equation of type (\ref{eq:Linear_SSE}) with values in $L^2 \left( \mathbb{N} ;  \mathfrak{h} \right)$,
using the contributions of \cite{MoraReAAP} we deduce that
$ \left( \left( \phi_{t} \left(  \sqrt{p_n}  \hat{\phi}_{0}^n  \right) 
/
\left\Vert \left( \phi_{t} \left(  \sqrt{p_k}  \hat{\phi}_{0}^k \right)  \right)_{k \in \mathbb{N}}   \right\Vert
 \right)_{n \in \mathbb{N}} 
, 
\left( W^{\ell}_t \right)_{\ell \in \mathbb{N}} \right)$ 
is the solution of the stochastic system \ref{st:SIWE}
under an adequate underlying probability measure,
where 
\[
W^{\ell}_t
=
B^{\ell}_t
+ 2  \int_0^t \Re \left( \Tr \left( L_{ \ell} \left( s \right) \rho_{s}  \right) \right) ds.
\]
Hence,
Models \ref{Modelo3} and  \ref{Modelo1} are alternative models to 
the stochastic master equation (\ref{eq:SQME})
for describing the evolution of quantum system undergoing continuous perfect measurements.

Under general conditions,
in Section \ref{sec:Well-definedness} 
we prove that the distribution of the density operator $\rho_{t}$ given by (\ref{eq:Def_rho}),
and so by (\ref{eq:Def_rhom}),
depends on the initial quantum state only as a function of the distribution $\mu$.
Hence,
the distribution of $\rho_{t}$ does not depend on 
the particular representation of $\rho_{0}$ by a random ensemble of random pure states,
and it does not depend on  the underlying probability space either.
Therefore,
Models \ref{Modelo3} and \ref{Modelo1} are well-defined,
which prevents inconsistent physical results.
For example,
we have that the measurement outcomes $B^{\ell}_t$ do not depend  on the decomposition
$
\sum_{n = 1}^{\infty} p_n \left\vert  \hat{\phi}^n_0   \rangle\langle  \hat{\phi}^n_0   \right\vert 
$
of $\rho_0$.
In order to prove the main result of Section \ref{sec:Well-definedness}
we deduce that
\begin{equation}
	\label{eq:Eta_Rep}
	\sum_{n = 1}^{\infty} \left\vert  \phi_{t} \left( \sqrt{p_n}  \hat{\phi}_{0}^n \right)  \rangle\langle  \phi_{t} \left( \sqrt{p_n}  \hat{\phi}_{0}^n \right)   \right\vert 
	=
	\Phi_t \rho_0 \Phi_t^{\star}  ,
\end{equation}
where 
$ \Phi_t:L^2 \left(  \Omega , \mathfrak{F}_{0}, \mathbb{Q}\right) \rightarrow L^2 \left(  \Omega , \mathfrak{F}_{t}, \mathbb{Q}\right) $ 
is the extension of the stochastic evolution operator  associated to the linear stochastic Schr\"odinger equation (\ref{eq:Linear_SSE}) (see Lemma \ref{OperadorPhi_t} for details),
and $\Phi_t^{\star}$ is the adjoint operator of  $\Phi_t$.
In case $\mathfrak{h}$ is finite-dimensional,
the representation (\ref{eq:Eta_Rep}) has been obtained by using that  
$\Phi_t^{\star}$  is the fundamental matrix of the so-called dual equation to (\ref{eq:Linear_SSE}),
i.e., the stochastic differential equation dual to (\ref{eq:Linear_SSE}) in $\mathfrak{h}$,
(see, e.g., Sections 2.2.2 and 3.1.1 of \cite{Barchielli}).
We do not use the dual equation to prove (\ref{eq:Eta_Rep}) with $\mathfrak{h}$  infinite-dimensional
due to this approach involves additional conditions that, for instance, 
guarantee the existence and uniqueness of the regular solution to the dual equation; 
the weak topological solution of the dual equation was established by \cite{Holevo}.
From (\ref{eq:Eta_Rep}) we get that $\rho_t$ defined by (\ref{eq:Int_Rho_t}) is equal to
$
\Phi_t \rho_0 \Phi_t^{\star} / \Tr \left( \Phi_t \rho_0 \Phi_t^{\star} \right) 
$, 
and so this $\rho_t$ does not depend  on the decomposition 
$
\sum_{n = 1}^{\infty} p_n \left\vert  \hat{\phi}^n_0   \rangle\langle  \hat{\phi}^n_0   \right\vert 
$
of $\rho_0$.
Finally,
by comparing $\rho_{t}$ given by (\ref{eq:Def_rho})
with an adequate density operator defined on the canonical space 
$\mathfrak{h}^{\mathbb{N}} \times  C \left( \left[ 0,T\right] ,  \mathbb{R}^{\mathbb{N}} \right)$,
we deduce that the probability distribution of $\rho_{t}$
does not depend on  the decomposition of the initial density operator
$\rho_0$ even if $\rho_0$ live in different probability spaces.

In Section \ref{sec:ExistenceRS_QME}
we prove that $\rho_{t}$ defined by Model \ref{Modelo1},
and so by Model \ref{Modelo3},
satisfies the Belavkin equation (\ref{eq:SQME})
provided that the initial condition $\rho_{0}$ is regular.
Hence, (\ref{eq:SQME}) has at least one solution.
Partially inspired by \cite{MoraAP},
we apply  the It\^{o} formula  to a functional of the solution of (\ref{eq:NLSSE_e}),
and then using techniques from operator theory we obtain that the density operator (\ref{eq:Def_rho})
satisfies (\ref{eq:SQME}) provided that
the integrals are  understood as Bochner integrals  in $ \mathfrak{L}_{2} \left( \mathfrak{h}\right)$,
where $ \mathfrak{L}_{2} \left( \mathfrak{h}\right)$ denotes the Hilbert space  of all Hilbert-Schmidt operators on $\mathfrak{h}$.
The problem of the uniqueness of the solution to  (\ref{eq:SQME}) remains open.
In addition to be a physically meaningful solution to  (\ref{eq:SQME}),
Model \ref{Modelo1} offers promising prospects for 
the numerical simulation of quantum measurement processes
with infinite-dimensional state space.
Indeed,  
in \cite{MoraFernBiscay2018} is solved numerically a set of finite-dimensional stochastic master equations  
by applying three exponential schemes to a finite-dimensional version of (\ref{eq:SIWE}).

According to Sections \ref{sec:EyU_solutions} and \ref{sec:Measurements} 
we have that Models \ref{Modelo3} and \ref{Modelo1} are mathematically sound
under the conditions ensuring the existence and uniqueness of the solution of
the non-linear stochastic Schr\"odinger equation (\ref{eq:NLSSE_i}) in \cite{FagMora2013},
which are slightly different from the hypotheses used in \cite{MoraReAAP} for that purpose.
If we additionally assume that $G\left( t \right), L_{1} \left( t \right), L_{2} \left( t \right), \ldots$ are closable operators,
then Section \ref{sec:ExistenceRS_QME} guarantees that $\rho_{t}$ solves the Belavkin equation (\ref{eq:SQME}) with a regular initial condition.
Therefore,
a number of physical situations previously studied in the literature satisfy the hypotheses required in this paper.
This fact is illustrated in Section \ref{sec:examples},
where we present two examples.
Section \ref{sec:Ex1} is devoted to an example in coordinate representation (i.e.,  $\mathfrak{h} =  L^{2}\left( \mathbb{R} ,\mathbb{C}\right) $), studied in \cite{FagMora2013},
that describes 
the continuous measurement of position of a closed quantum system  that is confined within a one-dimensional box.
Section \ref{sec:Ex2}  treats the continuous monitoring of the radiation field of a closed quantum system with Rabi Hamiltonian,
(for which $ \mathfrak{h} = \ell^2 \left(\mathbb{Z}_+ \right)\otimes \mathbb{C}^{2}$)
by using arguments similar to those in Section 4.1 of \cite{FagMora2019} and Section 4 of \cite{MoraReAAP}.

The paper is organized as follows.
Section \ref{sec:MainResults} presents  the main results of the paper,
which are proved in Section \ref{sec:Proofs}.
Throughout this paper we will use the following notation.
We write $\mathfrak{L}\left( \mathfrak{h}\right)$ for the space of all bounded linear operators on  $\mathfrak{h}$,
and 
$\mathfrak{L}_{1}\left( \mathfrak{h}\right)$ is the Banach space of all all trace-class operators on $\mathfrak{h}$ 
endowed  with the trace norm.

\section{Main results}
\label{sec:MainResults}

\subsection{Existence and uniqueness of solutions of the stochastic system \ref{st:SIWE} }
\label{sec:EyU_solutions}

Definition \ref{def:regular-sol-SIWE},  given below, provides a concept of regular solution 
for the stochastic system \ref{st:SIWE}   
that 
arises from adapting the notion of regular solution for the non-linear stochastic Schr\"odinger equation (\ref{eq:NLSSE_i})
(see, e.g., \cite{FagMora2013,MoraReAAP}).
Drawing a parallel with partial differential equations,
we describe the smoothness of the solutions to  the system of stochastic interacting wave functions  (\ref{eq:SIWE})
by replacing the partial derivatives of a function by the operator $C$.

\begin{condition}
\label{hyp:L-G-C-def}
The linear operator $C$ is a self-adjoint non-negative operator in $\mathfrak{h}$ with the following properties:

\begin{itemize}
 \item For any $t \geq 0$, 
 $\mathcal{D}\left(C \right) \subset \mathcal{D}\left( G  \left( t \right) \right)
 \cap  \bigcap_{\ell \in \mathbb{N}}  \mathcal{D}\left( L_{\ell} \left( t \right) \right) $,
 where $\mathcal{D}\left( A \right) $ denotes the domain of the operator $A$.

 \item The functions 
 $G \left( \cdot \right) \circ \pi _{\mathcal{D}\left(C \right)}, L_{\ell}  \left( \cdot \right)  \circ \pi _{\mathcal{D}\left(C \right)}
 : 
 \left( \left[ 0 , \infty \right[ \times \mathfrak{h}, 
\mathcal{B}\left(  \left[ 0 , 
\infty \right[ \times \mathfrak{h} \right) \right) 
\rightarrow 
\left(\mathfrak{h}, \mathcal{B} \left( \mathfrak{h} \right)  \right)
$
are measurable, where $\ell \in \mathbb{N}$, 
$ \mathcal{B} \left( \mathfrak{Y} \right)$ stands for the collection of all Borel sets of the topological space $ \mathfrak{Y}$
and
$
 \pi _{\mathcal{D}\left(C \right)} \left( x \right)
 =
\cases{
 x   & if  $ x\in \mathcal{D}\left( C\right) $
\\
0   & if  $ x \notin \mathcal{D}\left( C\right)$
}
$.
\end{itemize}
 
\end{condition}

\begin{definition}
\label{def:regular-sol-SIWE}
Let $C$ satisfy Hypothesis \ref{hyp:L-G-C-def}.
Assume that $\mathbb{I}$ is 
either $\left[ 0,\infty \right[ $ or the interval $\left[ 0,T\right] $
with $T\in \mathbb{R}_{+}$. 
We say that 
$\left( \Omega ,\mathfrak{F},\left( \mathfrak{F}_{t}\right) _{t \in \mathbb{I}}, \mathbb{P},
\left( \psi^n_t  \right)_{t \in \mathbb{I}}^{n \in \mathbb{N}},\left( W_{t}^{\ell}\right) _{t \in \mathbb{I}}^{\ell\in\mathbb{N}}\right) $
is a solution of class $C$ of the stochastic system \ref{st:SIWE}  with initial 
distribution $\theta$  if and only if:
\begin{itemize}
\item $W^{1}, W^2, \ldots $ are real-valued independent Brownian motions on
the filtered complete probability space 
$\left( \Omega,\mathfrak{F},\left( 
\mathfrak{F}_{t}\right) _{t \in \mathbb{I}},\mathbb{P}\right) $.

\item For any $n \in \mathbb{N}$, 
$\left(  \psi^n_t  \right)_{t \in \mathbb{I}} $  is an 
$ \mathfrak{h}  $-valued adapted process with continuous sample paths 
such that 
$\mathbb{P}\left( \sum_{n  \in \mathbb{N}} \left\Vert \psi^n_t \right\Vert ^{2} = 1 \; \mathrm{for}\; \mathrm{all}\; t  \in \mathbb{I} \right) =1$.

\item For any $n \in \mathbb{N}$ and $t \in \mathbb{I}$,  
$\psi^n_t  \in \mathcal{D}\left( C\right) $ $\mathbb{P}$-a.s.
and 
$
\sup_{s\in \left[ 0,t\right] } \sum_{n  \in \mathbb{N}}  \mathbb{E}_{\mathbb{P}}  \left\Vert C \psi^n_s \right\Vert ^{2}  < \infty 
$.

\item $\theta$ is the probability distribution of $ \left( \psi^n_0  \right)_{n \in \mathbb{N}} $ 
on $\left(\mathfrak{h}^{\mathbb{N}} ,  \mathcal{B} \left( \mathfrak{h}  \right)^{\mathbb{N}} \right)$,
where $ \mathcal{B} \left( \mathfrak{h}  \right)^{\mathbb{N}} $ is the product $\sigma$-algebra on 
$\mathfrak{h}^{\mathbb{N}} $.

\item 
$\mathbb{P}$-almost surely for all $t\in \mathbb{I}$,
\begin{eqnarray}
\fl
\nonumber
\psi^n_t 
& =
\psi^n_0 
+ \int_0^t\left( G\left( s \right) \pi _{\mathcal{D}\left( C\right) }\left(  \psi^n_s   \right)
+ \sum_{ \ell =1}^{ \infty} \left( \wp_{\ell} \left( s \right)  L_{\ell} \left( s \right)  \pi _{\mathcal{D}\left( C\right) }\left(  \psi^n_s   \right)
-\frac{1}{2} \wp_{\ell} \left( s  \right)^2    \psi^n_s    \right) 
\right)ds 
\\
\fl
\label{eq:SIWE_f}
& \quad
+
\sum_{ \ell =1}^{ \infty} \int_0^t  
\left( L_{\ell} \left( s \right) \pi _{\mathcal{D}\left( C\right) }\left(  \psi^n_s   \right) 
- \wp_{\ell} \left( s \right)  \psi^n_s    \right)  dW_s^{\ell} 
\end{eqnarray}
with
$
\wp_{\ell} \left( s  \right)
=
\sum_{n  \in \mathbb{N} } \Re \left( \langle  \psi^n_s ,  L_{\ell} \left( s \right)  \pi _{\mathcal{D}\left( C\right) }\left(  \psi^n_s   \right)  \rangle \right) 
$.

\end{itemize}

We shall say, for short, that 
$\left( \mathbb{P}, \left( \psi^n_t  \right)_{t \in \mathbb{I}}^{n \in \mathbb{N}},\left( W_{t}^{\ell}\right) _{t \in \mathbb{I}}^{\ell\in\mathbb{N}} \right)$ 
is a $C$-solution of the stochastic system \ref{st:SIWE} .
\end{definition}

\begin{remark}
\label{Nota:Distancia}
For any separable metric space $\left( \mathbf{E}, d \right) $,
the product  $\sigma$-algebra  $ \mathcal{B} \left(  \mathbf{E} \right)^{\mathbb{N}} $ 
is equal to  the Borel  $\sigma$-algebra of  
 $ \mathbf{E}^{\mathbb{N}} $ equipped with  the metric
\[
\fl
d_{\mathbf{E}^{\mathbb{N}}} \left(  \left( x_n\right)_{n \in \mathbb{N}} ,  \left( y_n\right)_{n \in \mathbb{N}}   \right)
=
\sum_{n=1}^{\infty} 2^{-n} \frac{d \left( x_n , y_n \right)}{ 1 + d \left( x_n , y_n \right)}
\qquad \quad \forall  \left( x_n\right)_{n \in \mathbb{N}} ,  \left( y_n\right)_{n \in \mathbb{N}} \in \mathbf{E}^{\mathbb{N}} 
\]
(see, e.g., proof of Proposition 12.2.2 of \cite{Dudley}).
Therefore,
the distribution $\theta$ of the Definition \ref{def:regular-sol-SIWE}  
is a probability measure on $ \mathcal{B} \left( \left( \mathfrak{h}^{\mathbb{N}}, d_{\mathfrak{h}^{\mathbb{N}}}  \right)\right)$,
because $\mathfrak{h}$ is a separable Hilbert space.
It is worth pointing out that 
the distance $ d_{\mathbf{E}^{\mathbb{N}}} $ induces the product topology of $\mathbf{E}^{\mathbb{N}} $ (see, e.g., Proposition 2.4.4 of \cite{Dudley}), and $\left( \mathbf{E}^{\mathbb{N}} , d_{\mathbf{E}^{\mathbb{N}}}  \right)$ is a Polish space.
\end{remark}

\begin{remark}
Since
$
\Tr \left(  \sum_{n = 1}^{\infty} \left\vert\psi^n_t  \rangle\langle\psi^n_t  \right\vert  \right) 
=  \sum_{n=1}^{\infty} \left\Vert \psi^n_t \right\Vert ^{2} 
$,
Definition \ref{def:regular-sol-SIWE} requires that 
$
 \Tr \left(  \sum_{n = 1}^{\infty} \left\vert\psi^n_t  \rangle\langle\psi^n_t  \right\vert  \right)  = 1
$
$ \mathbb{P}$-a.s.
Thus,
$ \sum_{n = 1}^{\infty} \left\vert\psi^n_t  \rangle\langle\psi^n_t  \right\vert $
is a density operator  $\mathbb{P}$-almost surely. 
\end{remark}

\begin{remark}
For simplicity of notation,
we omit writing  $ \pi _{\mathcal{D}\left( C\right) }\left(  \cdot  \right)$ in (\ref{eq:SIWE_f}),
when no confusion can arise.
Thus, (\ref{eq:SIWE_f}) becomes Stochastic system \ref{st:SIWE}.
 \end{remark}

In this paper, 
we assume the conditions on the coefficients $G \left( t \right)$, $L_{\ell} \left( t \right)$ described by Hypothesis \ref{hyp:CF}
given below.
These have been used to guarantee the existence and uniqueness of the solution of
the non-linear stochastic Schr\"odinger equation (\ref{eq:NLSSE_i}) (see, e.g., \cite{FagMora2013,MoraReAAP})
and in the study of the GKSL equation (\ref{eq:GKSL}) (see, e.g., \cite{Chebotarev2000,Fagnola,MoraJFA,MoraAP}).

\begin{condition}\label{hyp:CF}
Suppose that $C$ satisfies Hypothesis \ref{hyp:L-G-C-def}.  
In addition assume that:

\begin{itemize}

\item[(H2.1)] 
For all $t \geq 0$ and $x \in \mathcal{D}\left( C\right)$,
$\left\Vert  G  \left( t \right) x  \right\Vert^{2} 
\leq K \left( t \right) \left( \left\Vert  x  \right\Vert^{2} + \left\Vert  C x  \right\Vert^{2} \right)$,
where $K \left( \cdot \right)$ is a non-decreasing non-negative function.

\item[(H2.2)] For any $x \in \mathcal{D}\left(C \right)$ and all $t \geq 0$,
$
2\Re\left\langle  x, G \left( t \right) x\right\rangle 
+\sum_{\ell=1}^{\infty }\left\Vert  L_{\ell} \left( t \right) x \right\Vert ^{2} =  0
$.

\item[(H2.3)] There exists a non-decreasing non-negative function $\alpha$ 
and a core $\mathfrak{D}_{1}$ of $C^{2}$ such that  
for all $t \geq 0$ and $x \in \mathfrak{D}_{1}$ we have
\[
2\Re\left\langle C^{2} x, G \left( t \right) x\right\rangle 
+\sum_{\ell=1}^{\infty }\left\Vert C L_{\ell} \left( t \right) x \right\Vert ^{2}
\leq \alpha \left( t \right)  \left( \left\Vert  x  \right\Vert^{2} + \left\Vert  C x  \right\Vert^{2} \right) .
\]
\end{itemize}

\end{condition}

\begin{remark}
\label{re:CondH2.2}
Suppose that $G$ satisfies (\ref{eq:Def_G}) on $ \mathcal{D}\left( C\right) $,
where $C$ is a self-adjoint non-negative operator in $\mathfrak{h}$
such that 
$\mathcal{D}\left(C \right) \subset \mathcal{D}\left( H  \left( t \right) \right)
 \cap  \bigcap_{\ell \in \mathbb{N}}  \mathcal{D}\left( L_{\ell} \left( t \right) \right) $.
Then,
direct calculation yields Condition H2.2 of Hypothesis \ref{hyp:CF}. 
\end{remark}

We now establish the existence and uniqueness in law of the $C$-solution of  the stochastic system \ref{st:SIWE}.

\begin{definition}
We say that a uniqueness in joint law holds for the  $C$-solutions of the stochastic system \ref{st:SIWE} 
if for any two  $C$-solutions 
$\left( \mathbb{P}, \left( \psi^n_t  \right)_{t\in \mathbb{I} }^{n \in \mathbb{N}},\left( W_{t}^{\ell}\right) _{t\in \mathbb{I}}^{\ell\in\mathbb{N}} \right)$ 
and 
$\left( \hat{\mathbb{P}}, \left(  \hat{\psi}^n_t  \right)_{t\in \mathbb{I} }^{n \in \mathbb{N}},\left(  \hat{W}_{t}^{\ell}\right) _{t\in \mathbb{I}}^{\ell\in\mathbb{N}} \right)$
of the stochastic system \ref{st:SIWE}   with initial law $\theta$,
we have that
 the distributions of the random variables 
$\left(\left( \psi^n_{t\left(k \right)}  \right)_{1\leq k \leq N }^{n \in \mathbb{N}},
\left( W_{t\left(k \right)}^{\ell}\right) _{1\leq k \leq N}^{\ell\in\mathbb{N}} \right)$
and 
$\left(\left( \hat{\psi}^n_{t\left(k \right)}  \right)_{1\leq k \leq N }^{n \in \mathbb{N}},
\left( \hat{W}_{t\left(k \right)}^{\ell}\right) _{1\leq k \leq N}^{\ell\in\mathbb{N}} \right)$
coincide  for any $N \in \mathbb{N}$ and any collection 
$ t_{k \left( 1 \right)} <  t_{k \left( 2 \right)} < \cdots < t_{k \left( N \right)}$ of points of $\mathbb{I}$.
\end{definition}

\begin{theorem}
\label{th:EyU-SIWE}
Assume Hypothesis \ref{hyp:CF}.
Let $\theta$ be  a probability measure on $\mathfrak{h}^{\mathbb{N}} $,
endowed with the product $\sigma$-algebra,
such that
\[
\theta \left( \left\{ \left( x_n \right)_{n \in \mathbb{N}} :
x_n  \in \mathcal{D} \left( C \right) \,  \mathrm{for}\; \mathrm{all}\;  n \in \mathbb{N}, \, \sum_{n = 1}^{\infty} \left\Vert x_n \right\Vert^2 =1
\right\}  \right) = 1
\]
and
$
\int_{\mathfrak{h}^{\mathbb{N}}}
 \sum_{n = 1}^{\infty}  \left\Vert  C  x_n \right\Vert ^{2} \theta\left( d \left( x_n \right)_{n \in \mathbb{N} }\right) < \infty
$. 
Suppose that $\mathbb{I}$ is either $\left[ 0,\infty \right[ $ or the interval $\left[ 0,T\right] $ with $T\in \mathbb{R}_{+}$. 
Then, 
the stochastic system \ref{st:SIWE}  has a  solution 
$\left( \mathbb{P}, \left( \psi^n_t  \right)_{t\in \mathbb{I} }^{n \in \mathbb{N}},\left( W_{t}^{\ell}\right) _{t\in \mathbb{I}}^{\ell\in\mathbb{N}} \right)$ 
of class $C$ with initial law $\theta$.
Moreover, the uniqueness in joint law holds for the  $C$-solutions of  the stochastic system \ref{st:SIWE}.
\end{theorem}

\begin{proof}
Deferred to Section \ref{sec:th:EyU-SIWE}.
\end{proof}

\begin{remark}
In the space $ \mathfrak{h}^{\mathbb{N}}  \times \mathbb{R} ^{\mathbb{N}} $,
we consider the distance
\begin{eqnarray*}
&
 d_{\mathfrak{h}^{\mathbb{N}}  \times \mathbb{R} ^{\mathbb{N}} } 
 \left( \left( \left(  x_n \right)_{n \in \mathbb{N} } ,  \left(  u_n \right)_{n \in \mathbb{N} } \right)
,
\left( \left(  y_n \right)_{n \in \mathbb{N} } ,  \left(  v_n \right)_{n \in \mathbb{N} } \right)
 \right)
 \\
 &  =
d_{ \mathfrak{h} ^{\mathbb{N}}} \left( \left(  x_n \right)_{n \in \mathbb{N} } ,  \left(  y_n \right)_{n \in \mathbb{N} }  \right)
 +
d_{\mathbb{R} ^{\mathbb{N}}} \left( \left(  u_n \right)_{n \in \mathbb{N} } , \left(  v_n \right)_{n \in \mathbb{N} }  \right)
\end{eqnarray*}
for any $ \left(  x_n \right)_{n \in \mathbb{N} }, \left(  y_n \right)_{n \in \mathbb{N} } \in \mathfrak{h} ^{\mathbb{N}} $ 
and 
$  \left(  u_n \right)_{n \in \mathbb{N} },  \left(  v_n \right)_{n \in \mathbb{N} } \in  \mathbb{R} ^{\mathbb{N}} $.
Let $C \left( \mathbb{I} ,  \mathfrak{h} ^{\mathbb{N}}  \times \mathbb{R} ^{\mathbb{N}}  \right)$ be the space of all continuous 
$  \mathfrak{h} ^{\mathbb{N}}  \times \mathbb{R} ^{\mathbb{N}}  $-valued functions, together the metric
\[
\fl
d_{ C \left( \mathbb{I} ,  \mathfrak{h} ^{\mathbb{N}}  \times \mathbb{R} ^{\mathbb{N}}  \right) } \left( f , g \right)
=
\cases{
 \sup_{t \in \mathbb{I} } d_{\mathfrak{h}^{\mathbb{N}}  \times \mathbb{R} ^{\mathbb{N}} }   \left( f  \left( t  \right) , g \left( t  \right) \right)
 & if  $\mathbb{I}  = \left[ 0,T\right] $
 \\
 \sum_{n=1}^{\infty} 2^{-n} \frac{
 \sup_{t \in \left[ 0, n \right] } d_{\mathfrak{h}^{\mathbb{N}}  \times \mathbb{R} ^{\mathbb{N}} }   \left( f  \left( t  \right) , g \left( t  \right) \right)
 }{ 1 +  \sup_{t \in \left[ 0, n \right] } d_{\mathfrak{h}^{\mathbb{N}}  \times \mathbb{R} ^{\mathbb{N}} }  \left( f  \left( t  \right) , g \left( t  \right) \right) }
 &  if  $\mathbb{I}  = \left[ 0,\infty \right[  $
},
\]
where $T\in \mathbb{R}_{+}$.
Then,
$\mathcal{B} \left( C \left( \mathbb{I} ,  \mathfrak{h}^{\mathbb{N}}  \times \mathbb{R} ^{\mathbb{N}}   \right) \right) $ 
is equal to the $\sigma$-algebra 
\[
\left\{ A \cap  C \left( \mathbb{I} , \mathfrak{h}^{\mathbb{N}}  \times \mathbb{R} ^{\mathbb{N}}   \right)
:
A \in \mathcal{B} \left(  \mathfrak{h}^{\mathbb{N}}  \times \mathbb{R} ^{\mathbb{N}}  \right)^{\mathbb{I}} 
\right\}
,
\]
where $ \mathcal{B} \left(  \mathfrak{h}^{\mathbb{N}}  \times \mathbb{R} ^{\mathbb{N}}  \right)^{\mathbb{I}}$
is the product $\sigma$-algebra on $  \left(  \mathfrak{h}^{\mathbb{N}}  \times \mathbb{R} ^{\mathbb{N}}  \right)^{\mathbb{I}} $
(see, e.g., proof of Proposition 12.2.2 of \cite{Dudley}).
According to Theorem \ref{th:EyU-SIWE},
$\left(\left( \psi^n_t  \right)^{n \in \mathbb{N}},\left( W_{t}^{\ell}\right)^{\ell\in\mathbb{N}} \right)_{t\in \mathbb{I} }$
and 
$\left( \left(  \hat{\psi}^n_t  \right)^{n \in \mathbb{N}},\left(  \hat{W}_{t}^{\ell}\right)^{\ell\in\mathbb{N}} \right)_{t\in \mathbb{I} }$ have the same finite-dimensional probability distributions, 
and so 
the laws of 
$\left(\left( \psi^n_t  \right)^{n \in \mathbb{N}},\left( W_{t}^{\ell}\right)^{\ell\in\mathbb{N}} \right)_{t\in \mathbb{I} }$
and 
$\left( \left(  \hat{\psi}^n_t  \right)^{n \in \mathbb{N}},\left(  \hat{W}_{t}^{\ell}\right)^{\ell\in\mathbb{N}} \right)_{t\in \mathbb{I} }$
 on 
$\mathcal{B} \left( C \left( \mathbb{I} ,  \mathfrak{h}^{\mathbb{N}}  \times \mathbb{R} ^{\mathbb{N}}   \right) \right)  $ coincide.
This follows from applying the monotone class theorem or Proposition 1.4 of \cite{DaPrato}.
\end{remark}

\subsection{Continuous measurements}
\label{sec:Measurements}

\subsubsection{Physical Background.}
\label{PhysicalBackground}

Let the initial state of the quantum system be
the pure state 
$  \left\vert \hat{\phi}_0  \rangle\langle \hat{\phi}_0   \right\vert $,
where $ \hat{\phi}_0 $ is a unit vector of  $ \mathfrak{h}$.
Then,
we consider the following model of continuous quantum measurements (see, e.g., Sections 2.4  and 2.5.4 of \cite{Barchielli}).

\begin{itemize}
 \item For every $ \ell \in \mathbb{N} $, 
 $B_{t}^{\ell}$ is the  output of the continuous measurement of the quantum observable 
 $L_\ell \left( t\right) + L_\ell \left( t \right)^{\ast}  $ at time $t$.
 
 \item Under the probability measure $\mathbb{Q}$ on the filtration 
 $\left( \Omega ,\mathfrak{F}, \left(\mathfrak{F}_{t}\right) _{t\geq 0} \right)$, 
  $B^1, B^2, \ldots$ are real-valued independent Wiener processes.
 Moreover, any $\mathbb{Q}$-zero set belongs to  $\mathfrak{F}_{0}$.

 \item The state of the quantum system at time $t$ conditioned on 
 $\left\{ B_{s}^{\ell} : s \in \left[ 0, t \right] \, \mathrm{and} \, \ell \in \mathbb{N} \right\}$
 is  $  \left\vert \frac{ \phi_{t} \left( \hat{\phi}_{0} \right)}{\left\Vert \phi_{t} \left( \hat{\phi}_{0} \right) \right\Vert}   \rangle
 \langle \frac{\phi_{t} \left( \hat{\phi}_{0} \right)}{\left\Vert \phi_{t} \left( \hat{\phi}_{0} \right) \right\Vert}   \right\vert $,
 where $\phi_{t} \left( \hat{\phi}_{0} \right)$  is the strong solution of the linear stochastic Schr\"odinger equation 
 (\ref{eq:Linear_SSE}) with initial condition $\hat{\phi}_{0} $.

\item If $A \in \mathfrak{F}_{t}$, then the probability that the event $A$ occurs in physics is equal to
\[
\left\Vert \phi_{t} \left( \hat{\phi}_{0} \right)   \right\Vert ^{2} \cdot \mathbb{Q} \left( A \right)
: =
\int_A \left\Vert \phi_{t} \left( \hat{\phi}_{0} \right)   \right\Vert ^{2} d \mathbb{Q} .
\]
 
\end{itemize}
This model describes physical phenomena like the heterodyne detection 
(see, e.g. \cite{Barchielli,BreuerPetruccione,WisemanMilburn2010}).

Now,
we  suppose that the initial state $\rho_{0} $ is mixed, that is, 
\begin{equation}
\label{eq:Decomp_Inicial_rho}
\rho_{0} = \sum_{n = 1}^{\infty} p_n \left\vert \hat{\phi}^n_0  \rangle\langle \hat{\phi}^n_0  \right\vert 
\end{equation}
with $\left\Vert \hat{\phi}^n_0 \right\Vert = 1$, $p_n \geq 0$  and  $\sum_{n = 1}^{\infty} p_n = 1$.
The physical interpretation of  (\ref{eq:Decomp_Inicial_rho}) is that 
the system starts in the state $\left\vert \hat{\phi}^n_0  \rangle\langle \hat{\phi}^n_0  \right\vert$ with probability $p_n$.
As in Section 3.1 of  \cite{Barchielli},
we assume that the  output of the continuous measurement $B^1, B^2, \ldots$ are real-valued independent Wiener processes
in the the reference probability space 
$\left( \Omega ,\mathfrak{F}, \left(\mathfrak{F}_{t}\right) _{t\geq 0}, \mathbb{Q} \right)$,
and that $\mathbb{Q}  \left( \mathcal{N}  = n \right) = p_n$,
where the $\mathfrak{F}_{0}$-random variable  $\mathcal{N}: \Omega \rightarrow \mathbb{N}$
is defined by $\mathcal{N} \left( \omega \right) = n$
if and only if the initial state is $\left\vert \phi^n_0  \rangle\langle\phi^n_0  \right\vert$.

Set $\mathcal{B}_t = \left(\left( B_{s}^{\ell} \right)^{  \ell \in \mathbb{N}   }\right)_{ s \in \left[ 0, t \right] } $.
Thus, 
$\mathcal{B}_t$ is a random variable 
taking values in  $C \left( \left[ 0, t \right], \mathbb{R}^{\mathbb{N}}  \right)$ equipped with its Borel $\sigma$-algebra.
Since $\mathcal{B}_t$  and $\mathcal{N}$ are independent under $\mathbb{Q}$,
$
\mathbb{Q} \left( \left\{ \mathcal{N}  \in B \right\} \cap  \left\{ \mathcal{B}_t  \in A \right\} \right)
=
\mathbb{Q}_{ \mathcal{N}} \left(  B \right) \mathbb{Q}_{ \mathcal{B}_t} \left(  A \right) 
$,
where $ B \subset  \mathbb{N} $, $A \in \mathfrak{B} \left( C \left( \left[ 0, t \right], \mathbb{R}^{\mathbb{N}}  \right)\right) $,
and $\mathbb{Q}_{ \mathcal{N}}$ and $ \mathbb{Q}_{ \mathcal{B}_t}$ are the probability distribution of 
$\mathcal{N}$ and $ \mathcal{B}_t$ under $\mathbb{Q}$, respectively.
From the four item of the first paragraph we have that 
\[
 \mathbb{P}_{ \mathcal{B}_t |  \mathcal{N}  } \left( A | n \right)
=
\int_{A  }  
\left\Vert \phi_{t} \left( \hat{\phi}_{0}^n \right)   \right\Vert ^{2} d \mathbb{Q}_{\mathcal{B}_t }
\]
is the physical  conditional probability distribution of $\mathcal{B}_t$ given $\mathcal{N}$.
Applying the Bayes rule (see, e.g., Theorem 1.31 of \cite{Schervish1995}) we deduce that
the physical  conditional probability distribution of $\mathcal{N}$ given $\mathcal{B}_t$ is
\[ 
\fl
\mathbb{P}_{ \mathcal{N} |  \mathcal{B}_t  } \left( B | B_{s}^{\ell}  : s \in \left[ 0, t \right] \, \mathrm{and} \, \ell \in \mathbb{N}  \right)
=
\frac{1}{\sum_{k = 1}^{\infty} p_k \left\Vert \phi_{t} \left( \hat{\phi}_{0}^k \right)   \right\Vert ^{2} }
\sum_{n \in B}
p_n \left\Vert \phi_{t} \left( \hat{\phi}_{0}^n \right)   \right\Vert ^{2}
.
\]
Hence,
conditioned on $\left\{ B_{s}^{\ell} : s \in \left[ 0, t \right] \; \mathrm{and} \; \ell \in \mathbb{N} \right\}$,
the probability that the system is in the state 
 $  \left\vert \frac{\phi_{t} \left( \hat{\phi}_{0}^n \right)}{\left\Vert \phi_{t} \left( \hat{\phi}_{0}^n \right) \right\Vert}   \rangle
 \langle \frac{\phi_{t} \left( \hat{\phi}_{0}^n \right)}{\left\Vert \phi_{t} \left( \hat{\phi}_{0}^n \right) \right\Vert}   \right\vert $
is equal to
\[
\fl
\mathbb{P}_{ \mathcal{N} |  \mathcal{B}_t  } \left( \left\{ n \right\} | B_{s}^{\ell}  : s \in \left[ 0, t \right] \, \mathrm{and} \; \ell \in \mathbb{N}  \right)
=
p_n \left\Vert \phi_{t} \left( \hat{\phi}_{0}^n \right)   \right\Vert ^{2} / \left(
\sum_{k = 1}^{\infty} p_k \left\Vert \phi_{t} \left( \hat{\phi}_{0}^k \right)   \right\Vert ^{2} 
\right)
.
\]
Then,
conditioned on $\left\{ B_{s}^{\ell} : s \in \left[ 0, t \right] \; \mathrm{and} \; \ell \in \mathbb{N} \right\}$
the quantum system at time $t$ is in the mixed state
\[
\rho_t
=
\sum_{n = 1}^{\infty}
p_n \frac{\left\Vert \phi_{t} \left( \hat{\phi}_{0}^n \right)   \right\Vert ^{2}
}{
 \sum_{k = 1}^{\infty} p_k \left\Vert \phi_{t} \left( \hat{\phi}_{0}^k \right)   \right\Vert ^{2} 
}  
\left\vert \frac{\phi_{t} \left( \hat{\phi}_{0}^n \right)}{\left\Vert \phi_{t} \left( \hat{\phi}_{0}^n \right) \right\Vert}   \rangle
 \langle \frac{\phi_{t} \left( \hat{\phi}_{0}^n \right)}{\left\Vert \phi_{t} \left( \hat{\phi}_{0}^n \right) \right\Vert}   \right\vert ,
\]
(see, e.g., Section 3.1 of \cite{Barchielli}).
Hence,
\begin{equation}
\label{eq:Ded_rho_t}
 \rho_t
=
\frac{ 1 }{
 \sum_{k = 1}^{\infty} p_k \left\Vert \phi_{t} \left( \hat{\phi}_{0}^k \right)   \right\Vert ^{2} 
} 
\sum_{n = 1}^{\infty}
p_n 
\left\vert \phi_{t} \left( \hat{\phi}_{0}^n \right) \rangle \langle \phi_{t} \left( \hat{\phi}_{0}^n \right)  \right\vert 
 .
\end{equation}

 In the literature (see, e.g., \cite{Barchielli}),
(\ref{eq:Ded_rho_t})  is rewritten as $ \rho_{t} =  \eta_{t} / tr \left( \eta_{t}  \right) $ with 
\begin{equation}
\label{def:eta}
 \eta_{t} = \sum_{n = 1}^{\infty} p_n \left\vert \phi_{t} \left( \hat{\phi}^n_0 \right)  \rangle\langle \phi_{t} \left( \hat{\phi}^n_0 \right)  \right\vert 
.
\end{equation}
Applying the It\^{o} formula yields  the linear stochastic quantum master equation 
\begin{eqnarray}
\nonumber
\eta_{t}
 & = 
 \rho_{0}
 +
 \int_0^t \left(
 G \left( s \right) \eta_{s}
 +  \eta_{s}   G\left( s \right)^{\ast} 
 + \sum_{\ell=1}^{\infty} L_{ \ell} \left( s \right) \eta_{s}  L_{\ell}\left( s \right)^{\ast}
 \right) ds
 \\
 \label{eq:LQME}
 & \quad + 
 \sum_{\ell=1}^{\infty} 
  \int_0^t \left( L_{ \ell} \left( s \right) \eta_{s}  +  \eta_{s}  L_{\ell}\left( s \right)^{\ast}
  \right) d B^{\ell}_s ,
\end{eqnarray}
whenever $dim \left( \mathfrak{h} \right) < + \infty$ 
or 
$dim \left( \mathfrak{h} \right) = + \infty$ with $G \left( s \right)$, $L_{ \ell} \left( s \right)$ bounded operators
(see, e.g., \cite{BarchielliHolevo1995,BarchielliPaganoniZucca});
in these cases (\ref{eq:LQME}) has a unique solution (see, e.g., \cite{BarchielliPaganoni}).
By combining It\^{o}'s formula with Girsanov's theorem,
it is obtained that $\eta_{t} / tr \left( \eta_{t}  \right)$ satisfies the stochastic quantum master equation (\ref{eq:SQME})
(see, e.g., \cite{BarchielliHolevo1995,BarchielliPaganoniZucca}).

\subsubsection{Derivation of Model \ref{Modelo1}.}
\label{sec:DerivationModel}

Suppose that $\rho_{0}$ is given by (\ref{eq:Decomp_Inicial_rho}),
but with $p_1, p_2, \ldots$ being $\mathfrak{F}_{0}$-measurable real-value random variables  such that 
$p_n \geq 0$ and  $\sum_{n = 1}^{\infty} p_n = 1$, 
and $ \hat{\phi}^1_0,  \hat{\phi}^2_0, \ldots $ being 
$\mathfrak{F}_{0}$-measurable $\mathfrak{h}$-value random variables 
such that  $\left\Vert  \hat{\phi}^n_0 \right\Vert = 1$ for all $n \in \mathbb{N}$.
Since  (\ref{eq:Linear_SSE})  has a unique strong regular solution,
conditioning on $\mathfrak{F}_{0}$ we obtain from (\ref{eq:Ded_rho_t}) that
\[
 \rho_t
=
\frac{ 1 }{
 \sum_{k = 1}^{\infty} p_k \left\Vert \phi_{t} \left( \hat{\phi}_{0}^k \right)   \right\Vert ^{2} 
} 
\sum_{n = 1}^{\infty}
p_n 
\left\vert \phi_{t} \left( \hat{\phi}_{0}^n \right) \rangle \langle \phi_{t} \left( \hat{\phi}_{0}^n \right)  \right\vert 
 .
\]
Using the linearity of (\ref{eq:Linear_SSE}) yields
\begin{equation}
\label{eq:Ded_rho_t_equiv}
 \rho_t
=
\frac{ 1 }{
 \sum_{k = 1}^{\infty}  \left\Vert \phi_{t} \left(  \sqrt{p_k} \hat{\phi}_{0}^k \right)   \right\Vert ^{2} 
} 
\sum_{n = 1}^{\infty} 
\left\vert \phi_{t} \left(  \sqrt{p_n} \hat{\phi}_{0}^n \right) \rangle \langle  \phi_{t} \left(  \sqrt{p_n} \hat{\phi}_{0}^n \right)  \right\vert 
 .
\end{equation}
Taking $  \psi^n_0 = \sqrt{p_n}  \hat{\phi}_{0}^n  $  we obtain
\begin{equation}
\label{eq:Rho_ded}
 \rho_t
=
\sum_{n = 1}^{\infty}
\left\vert \psi_{t}^n \rangle \langle \psi_{t}^n  \right\vert  ,
\end{equation}
where,  by abuse of notation,
\[
\psi^n_t =  \phi_{t} \left(  \sqrt{p_n}  \hat{\phi}_{0}^n  \right) / \sqrt{  \sum_{k=1}^{\infty} \left\Vert \phi_{t} \left(  \sqrt{p_k}  \hat{\phi}_{0}^k  \right)   \right\Vert ^{2}  }
.
\]
Next, we get that $\left( \psi^n_t  \right)_{t \in \left[ 0,T\right]}^{n \in \mathbb{N}}$
is the unique weak probabilistic solution of the stochastic system \ref{st:SIWE}
with initial datum 
$  \left( \sqrt{p_n}  \hat{\phi}_{0}^n  \right)_{n \in \mathbb{N}} $,
and so (\ref{eq:Rho_ded}) coincides with (\ref{eq:Def_rho}).

First, we examine the solutions to the linear stochastic Schr\"odinger equation (\ref{eq:Linear_SSE})
that appear in (\ref{eq:Ded_rho_t}) and (\ref{eq:Ded_rho_t_equiv}).

\begin{definition} 
\label{def:regular-sol-LSIWE}
Suppose that $C$ satisfy Hypothesis \ref{hyp:L-G-C-def},
and that  $\mathbb{I}$ is 
either $\left[ 0,\infty \right[ $ or the interval $\left[ 0,T\right] $
with $T\in \mathbb{R}_{+}$. 
Let $B^1, B^2, \ldots$  be real-valued independent Wiener processes on a filtered probability 
space $\left( \Omega ,\mathfrak{F}, \left(\mathfrak{F}_{t}\right) _{t \in \mathbb{I}},\mathbb{Q}\right) $,
where $\left( \mathfrak{F}_{t}\right) _{t \in \mathbb{I}}$ satisfies the so-called usual conditions.
Assume that $\phi_0$ is an $\mathfrak{h}$-valued $\mathfrak{F}_{0}$-measurable random variable
such that $\mathbb{E} \left( \left\Vert \phi_0 \right\Vert ^{2} +  \left\Vert C \phi_0 \right\Vert ^{2}  \right) < \infty$.
An $\mathfrak{h}$-valued adapted process with continuous sample paths  
$\left( \phi_t  \right)_{t \in \mathbb{I}} $
 is called strong 
$C$-solution of (\ref{eq:Linear_SSE}) on  $\mathbb{I}$ with initial datum $\phi_0 $ if and only if:

\begin{itemize}

\item For any  $t\in \mathbb{I}$, $\phi_t  \in \mathcal{D}\left( C\right) $ $\mathbb{Q}$-a.s.

\item For all $t\in \mathbb{I}$, 
$ 
\mathbb{E} \left\Vert \phi_t \right\Vert ^{2} \leq  \mathbb{E} \left\Vert \phi_0 \right\Vert ^{2}
$
and
$
\sup_{s\in \left[ 0,t\right] } \mathbb{E} \left\Vert C \phi_s \right\Vert ^{2}  < \infty 
$.

\item Almost surely for all $t\in \mathbb{I}$,
\begin{equation*}
\phi_t
 =
\phi_0 
+ \int_0^t  G\left( s \right) \pi _{\mathcal{D}\left( C\right) }\left(  \phi_s   \right) ds 
+
\sum_{ \ell =1}^{ \infty} \int_0^t   L_{\ell} \left( s \right) \pi _{\mathcal{D}\left( C\right) }\left(  \phi_s   \right)   dB_s^{\ell} .
\end{equation*}

\end{itemize}
\end{definition}

\begin{theorem}
\label{th:regular-sol-LSIWE}
Let Hypothesis \ref{hyp:CF} hold true.
Suppose that $\left( \psi^n_0  \right)_{n \in \mathbb{N}}$ is a sequence of 
$\mathfrak{h}$-valued $\mathfrak{F}_{0}$-measurable random variables such that
$ 
\sum_{n=1}^{\infty} \left( \mathbb{E} \left\Vert \psi^n_0 \right\Vert ^{2} 
+  \mathbb{E} \left\Vert C \psi^n_0 \right\Vert ^{2} \right) <  \infty
$.
Let $\mathbb{I}$ be either $\left[ 0,\infty \right[ $ or $\left[ 0,T\right] $  with $T\in \mathbb{R}_{+}$.
Then, for each $n \in \mathbb{N}$ the stochastic evolution equation (\ref{eq:Linear_SSE})
has a unique strong $C$-solution $\left( \phi_t \left( \psi^n_0 \right) \right)_{ t \in \mathbb{I}}$ on  $\mathbb{I}$ 
with initial datum $ \psi^n_0$.
Moreover,
$\left(  \sum_{n=1}^{\infty} \left\Vert \phi_{t} \left( \psi^n_0 \right)   \right\Vert ^{2} \right) _{t \in \mathbb{I}} $ 
is a martingale, 
and
\begin{equation}
\fl
\label{eq:LS_cota}
\mathbb{E} \left( 
\sum_{n=1}^{\infty} \left\Vert C \phi_{t} \left( \psi^n_0 \right)   \right\Vert ^{2} \right)
\leq 
\exp \left( t \alpha \left( t \right) \right) 
\left( \mathbb{E} \left(\sum_{n=1}^{\infty} \left\Vert C  \psi^n_0    \right\Vert ^{2} \right)
+ t \alpha\left( t \right)  \mathbb{E} \left(
\sum_{n=1}^{\infty} \left\Vert  \psi^n_0    \right\Vert ^{2} \right) \right)
\end{equation}
for all $t \in \mathbb{I}$.
\end{theorem}

\begin{proof}
Deferred to Section \ref{sec:Proof:LSIWE}.
\end{proof}

Using Theorem \ref{th:regular-sol-LSIWE} we deduce that 
the positive operators
$
\sum_{n = 1}^{\infty} 
\left\vert \phi_{t} \left(  \hat{\phi}_{0}^n \right) \rangle \langle  \phi_{t} \left( \hat{\phi}_{0}^n \right)  \right\vert 
$
and
$
\sum_{n = 1}^{\infty} 
\left\vert \phi_{t} \left(  \sqrt{p_n} \hat{\phi}_{0}^n \right) \rangle \langle  \phi_{t} \left(  \sqrt{p_n} \hat{\phi}_{0}^n \right)  \right\vert 
$,
which are present in (\ref{eq:Ded_rho_t}) and (\ref{eq:Ded_rho_t_equiv}),
are positive trace class operators
(see also Lemma \ref{le:L1} given in Section \ref{sec:L2Nh}).
From Theorem \ref{th:Model}, given below, we have that the law of the density operator (\ref{eq:Rho_ded})
is uniquely determined by the solution of the stochastic system \ref{st:SIWE} for each regular dataset $ p_n $ and $ \hat{\phi}^n_0 $,
and 
the probability measure corresponding to this solution determines uniquely the physical  probability distribution 
$ \left(  \sum_{n=1}^{\infty} \left\Vert \phi_{T} \left(  \sqrt{p_n} \hat{\phi}_{0}^n  \right)   \right\Vert ^{2} \right) \cdot \mathbb{Q} $.
Hence, the density operators (\ref{eq:Def_rho}) and (\ref{eq:Rho_ded}) are identically distributed. 
Therefore,
the density operator $\rho_t $ defined by (\ref{eq:Def_rho})
is actually  given by (\ref{eq:Rho_ded}),
and so  it describes the state  of the quantum system continuously monitored by the apparatus $L_\ell \left( t\right) + L_\ell \left( t \right)^{\ast}  $.

\begin{theorem}
\label{th:Model}
Assume Hypothesis \ref{hyp:CF}. Let $T\in \mathbb{R}_{+}$.
Suppose that $\left( \psi^n_0  \right)_{n \in \mathbb{N}}$ is a sequence of 
$\mathfrak{h}$-valued $\mathfrak{F}_{0}$-measurable random variables such that
$
\sum_{n=1}^{\infty} \left\Vert \psi^n_0 \right\Vert ^{2} = 1 
$,
$
\sum_{n=1}^{\infty} \left\Vert C \psi^n_0 \right\Vert ^{2} < \infty 
$
and
$
\sum_{n=1}^{\infty} \mathbb{E} \left\Vert C \psi^n_0 \right\Vert ^{2} < \infty 
$.
Set
\[
\psi^n_t = 
\cases{
 \phi_{t} \left( \psi^n_0 \right) / \sqrt{  \sum_{n=1}^{\infty} \left\Vert \phi_{t} \left( \psi^n_0 \right)   \right\Vert ^{2}  }
&
 if  $  \sum_{n=1}^{\infty} \left\Vert \phi_{t} \left( \psi^n_0 \right)   \right\Vert ^{2} \neq 0 $
\\
0
&
 if  $  \sum_{n=1}^{\infty} \left\Vert \phi_{t} \left( \psi^n_0 \right)   \right\Vert ^{2} = 0 $
} 
\]
and 
\[
W_{t}^{\ell} = B_{t}^{\ell}
- \int_{0}^{t}    \left( 2 \sum_{n=1}^{\infty}  \Re \left(  \langle  \psi^n_s , L_{\ell} \left( s \right) \psi^n_s  \rangle  \right)  \right) ds 
,
\]
where $ \phi_{t} \left( \psi^n_0 \right) $ is the strong $C$-solution of  (\ref{eq:Linear_SSE}) with initial datum $ \psi^n_0 $.
Then,
$
\left( \left(  \sum_{n=1}^{\infty} \left\Vert \phi_{T} \left( \psi^n_0 \right)   \right\Vert ^{2} \right) \cdot \mathbb{Q}, \left( \psi^n_t  \right)_{t \in \left[ 0,T\right]}^{n \in \mathbb{N}},\left( W_{t}^{\ell}\right) _{t \in\left[ 0,T\right]}^{\ell\in\mathbb{N}} \right)
$
is the unique (in the joint law sense) $C$-solution of the stochastic system \ref{st:SIWE}  with initial distribution  equal to the law of $\left( \psi^n_0  \right)_{n \in \mathbb{N}}$.
\end{theorem}

\begin{proof}
Deferred to Section \ref{sec:Proof_Model}.
\end{proof}

 Let the assumptions of Theorem \ref{th:Model} hold. 
 Suppose that  $ \psi^n_0 = 0 $ for all $n \geq N+1$, where $N \in \mathbb{N}$.
 Since (\ref{eq:Linear_SSE}) with initial datum $ \psi^n_0 $ has a unique strong $C$-solution,
$ \phi_t \left( \psi_0^n \right)  \equiv 0 $ for any  $n \geq N+1$.
Applying Theorem \ref{th:Model} we obtain that 
$ \psi^n_t \equiv 0$ for  all $n \geq N+1$,
where $\left(\left( \psi^n_t  \right)^{n \in \mathbb{N}},\left( W_{t}^{\ell}\right)^{\ell\in\mathbb{N}} \right)_{ t \in \left[ 0,T\right] }$
is the unique (in the joint law sense) $C$-solution of the stochastic system \ref{st:SIWE} 
with initial distribution equal to the law of $\left( \psi^n_0  \right)_{n \in \mathbb{N}}$.
Therefore, 
$ \psi^1_t, \ldots, \psi^N_t$  satisfy the following system:
\begin{stochastic}
\label{st:SIWE_N}
 For any  $ n = 1, 2, \ldots, N $ we have
\begin{eqnarray*}
\psi^n_t 
& =
\psi^n_0 
+ \int_0^t\left( G\left( s \right) \psi^n_s
+ \sum_{ \ell =1}^{ \infty} \left( \wp_{\ell} \left( s \right)  L_{\ell} \left( s \right)  \psi^n_s
-\frac{1}{2} \wp_{\ell} \left( s  \right)^2  \psi^n_s \right) 
\right)ds 
\\
& \quad
+
\sum_{ \ell =1}^{ \infty} \int_0^t  
\left( L_{\ell} \left( s \right) \psi^n_s - \wp_{\ell} \left( s \right) \psi^n_s  \right)  dW_s^{\ell} ,
\end{eqnarray*}
where $ \psi^1_t, \psi^2_t, \ldots,  \psi^N_t$ 
are $\mathfrak{h}$-valued adapted process with continuous sample paths,
\[
\wp_{\ell} \left( s  \right)
=
\sum_{n=1}^{N} \Re \left( \langle  \psi^n_s ,  L_{\ell} \left( s \right)  \psi^n_s \rangle \right) 
\]
and
$W^1, W^2, \ldots$ are independent real Brownian motions.
\end{stochastic}
\noindent  Using Theorem \ref{th:EyU-SIWE} we obtain that
the stochastic system \ref{st:SIWE_N} has uniqueness in joint law for its $C$-solutions,
which are defined by Definition \ref{def:regular-sol-SIWE}
with $ n  \in \mathbb{N}$ replaced by $n \leq N$.

\begin{remark}
Next, we will carry out formal computations.
Consider the solution $\left( \psi^n_t  \right)_{n \in \mathbb{N}}$ of the stochastic system \ref{st:SIWE} given by Theorem \ref{th:Model} with  $  \psi^n_0 = \sqrt{p_n}  \hat{\phi}_{0}^n  $.
For the sake of simplicity, we assume $p_n > 0$.
Set 
$  \hat{ \psi }^n_t  = \psi^n_t / \left\Vert \psi^n_t  \right\Vert $,
and
$ p_n \left( t\right) = \left\Vert \psi^n_t  \right\Vert^2
$.
From (\ref{eq:Rho_ded}) we obtain 
\begin{equation}
	\label{eq:NuevaSIWF4}
\rho_t
=
\sum_{n = 1}^{\infty} p_n \left( t\right) 
\left\vert \hat{ \psi }^n_t \rangle \langle \hat{ \psi }^n_t  \right\vert ,
\end{equation}
and (\ref{eq:Def_pl}) yields
\begin{equation}
	\label{eq:NuevaSIWF1}
	\wp_{\ell} \left( t \right) 
	=
	\sum_{k=1}^{\infty}  p_n \left( t\right) \Re \left( \langle  \hat{ \psi }^k_t ,  L_{\ell} \left( t \right)  \hat{ \psi }^k_t \rangle \right)
	.
\end{equation}
Combining It\^{o}'s formula with Hypothesis H2.2 we obtain
\[
\left\Vert \psi^n_t \right\Vert ^{2} 
= 
\left\Vert \psi^n_0 \right\Vert ^{2} 
+
\sum_{ \ell =1}^{ \infty} \int_0^t  2 \, \Re \left(
\langle \psi^n_s , L_{\ell} \left( s \right) \psi^n_s - \wp_{\ell} \left( s \right) \psi^n_s \rangle 
\right)  dW_s^{\ell} ,
\]
and so
\begin{equation}
	\label{eq:NuevaSIWF2}
p_n \left( t\right) 
=
p_n
+
\sum_{ \ell =1}^{ \infty} \int_0^t  
2 \, p_n \left( s\right)\left(  
  \Re \left(  \langle \hat{ \psi }^n_s  , L_{\ell} \left( s \right) \hat{ \psi }^n_s   \rangle \right) 
-  \wp_{\ell} \left( s \right) 
\right) dW_s^{\ell} .
\end{equation}
Using the definition of $  \hat{ \psi }^n_t $ gives
$  \hat{ \psi }^n_t = \phi_{t} \left( \psi^n_0 \right) /  \left\Vert \phi_{t} \left( \psi^n_0 \right)   \right\Vert $.
It is worth pointing out that
under the probability $ \left( \sum_{k=1}^{\infty} \left\Vert \phi_{T} \left( \psi^k_0 \right)   \right\Vert ^{2} \right) \cdot \mathbb{Q} $
the evolution of $  \phi_{t} \left( \psi^n_0 \right) /  \left\Vert \phi_{t} \left( \psi^n_0 \right)   \right\Vert $ is not describe by 
the non-linear stochastic Schr\"odinger equation (\ref{eq:NLSSE_i}).
This leads us to formally apply the It\^{o} formula to deduce 
\begin{eqnarray}
\fl \nonumber
 \hat{ \psi }^n_t 
 & = 
 \hat{ \psi }^n_0
 +
 \int_{0}^t \left( 
 G\left( s \right) \hat{ \psi }^n_s  + 
 \sum_{ \ell =1}^{ \infty} 
 \left( 2 \wp_{\ell} \left( s \right) - \Re \left(  \langle \hat{ \psi }^n_s  , L_{\ell} \left( s \right) \hat{ \psi }^n_s   \rangle \right) 
 \right) L_{\ell} \left( s \right) \hat{ \psi }^n_s   
 \right)  ds
 \\
 \fl \nonumber
 &  + 
  \int_{0}^t  \sum_{ \ell =1}^{ \infty} 
 \left( \frac{3}{2} \wp_{\ell} \left( s \right)^2 
 - 5 \wp_{\ell} \left( s \right)   \Re \left(  \langle \hat{ \psi }^n_s  , L_{\ell} \left( s \right) \hat{ \psi }^n_s   \rangle \right) 
 + 3 \left( \Re \left(  \langle \hat{ \psi }^n_s  , L_{\ell} \left( s \right) \hat{ \psi }^n_s   \rangle \right)  \right)^2
 \right) \hat{ \psi }^n_s     ds
 \\
 \fl
 & \quad + 
\sum_{ \ell =1}^{ \infty} \int_0^t  \left( L_{\ell} \left( s \right) \hat{ \psi }^n_s 
- \Re \left(  \langle \hat{ \psi }^n_s  , L_{\ell} \left( s \right) \hat{ \psi }^n_s   \rangle \right)  \psi^n_s  \right)  dW_s^{\ell} .
	\label{eq:NuevaSIWF3}
\end{eqnarray}
Thus,
according to (\ref{eq:NuevaSIWF4}) we have that the density operator $\rho_t$ 
is the ensemble of pure states $\hat{ \psi }^n_t$ occurring with probability $p_n \left( t\right) $,
where $\hat{ \psi }^1_t, \hat{ \psi }^2_t, \ldots $ and $ p_1 \left( t\right), p_2 \left( t\right), \ldots  $ satisfy the non-linear system of stochastic evolution equations given by 
(\ref{eq:NuevaSIWF1}), (\ref{eq:NuevaSIWF2}) and (\ref{eq:NuevaSIWF3}).
\end{remark}

\subsection{Well-definedness of Model \ref{Modelo1}}
\label{sec:Well-definedness}

We begin by recalling the notion of regular density operator introduced by \cite{ChebGarQue98} in case  $C$ is invertible
(see also \cite{MoraAP}).

\begin{definition}
\label{def:C-regular-operator}
Let $C$ be a non-negative self-adjoint operator in $\mathfrak{h}$. 
A non-negative  operator  $ \varrho:  \mathfrak{h} \rightarrow   \mathfrak{h}$ is  $C$-regular 
if 
$
 \varrho = \sum_{n\in\mathfrak{I}}\lambda_{n}\left\vert u_{n}\rangle\langle u_{n}\right\vert
$,
where 
 $\mathfrak{I}$ is a countable set, 
$ \lambda_{n} \geq 0 $, 
$  u_{n} \in \mathcal{D}\left( C\right) $,
$ \sum_{n\in \mathfrak{I}}\lambda_{n} < \infty$
and
$
 \sum_{n\in \mathfrak{I}}\lambda_{n} \left( \left\Vert u_{n}\right\Vert ^{2} +  \left\Vert Cu_{n}\right\Vert ^{2} \right)<\infty
$.
We write  $\mathfrak{L}_{1,C}^+ \left( \mathfrak{h}\right) $ for the set of all $C$-regular non-negative operators in $\mathfrak{h}$.
\end{definition}

\begin{remark}
From  \cite{ChebGarQue98} we have that  a non-negative operator 
$ \varrho \in \mathfrak{L}_{1}\left( \mathfrak{h}\right)$ is  $C$-regular 
if and only if 
$
\sup_{n \in \mathbb{N}} \Tr \left( \varrho \, C^2 \left( I + \frac{1}{n} C^2  \right)^{-1} \right) < + \infty 
$.
This gives
$\mathfrak{L}_{1,C}^+ \left( \mathfrak{h}\right) \in  \mathcal{B} \left( \mathfrak{L}_{1}\left( \mathfrak{h}\right) \right)$.
\end{remark}

Let $\mathbb{E} \Tr \left( C \varrho_0  C \right)  < + \infty$,
where $C$ fulfills the conditions of Hypotheses \ref{hyp:CF}.
According to the following theorem we have that there is at least 
one quantum state $\rho_{t}$ defined by (\ref{eq:Def_rho}),
which is uniquely determined by each regular dataset  $ \left( \sqrt{p_n}  \hat{\phi}^n_0  \right)_{n \in \mathbb{N}} $
according to Theorem \ref{th:EyU-SIWE}.

\begin{theorem}
\label{th:Def_Model1}
Assume Hypotheses \ref{hyp:CF}.
Let $\mathbb{I}$ be 
either $\left[ 0,\infty \right[ $ or the interval $\left[ 0,T\right] $
with $T\in \mathbb{R}_{+}$. 
Suppose that $\mu$ is a probability measure on $ \mathfrak{L}_{1} \left( \mathfrak{h}\right) $
such that
$
\mu \left( \mathfrak{L}_{1,C}^+ \left( \mathfrak{h}\right) \right) = 1
$
and
$
\int_{ \mathfrak{L}_{1} \left( \mathfrak{h}\right) }
 \Tr \left( C \varrho  \, C \right)  \mu \left( d \varrho \right) < \infty
$. 
Then,
there exists a  $C$-solution 
$\left( \mathbb{P}, \left( \psi^n_t  \right)_{t \in \mathbb{I} }^{n \in \mathbb{N}},\left( W_{t}^{\ell}\right) _{t \in \mathbb{I}}^{\ell\in\mathbb{N}} \right)$ of the stochastic system \ref{st:SIWE} 
such that $  \psi^n_0 = \sqrt{p_n}  \hat{\phi}^n_0  $  for all $n \in \mathbb{N}$,
where
$p_n \in \left[ 0, + \infty \right[$ and $\hat{\phi}^n_0 \in \mathfrak{h}$ are $\mathfrak{F}_{0}$-measurable random variables,
$\sum_{n = 1}^{\infty} p_n = 1$, $\left\Vert  \hat{\phi}^n_0 \right\Vert = 1$,
and
$
 \sum_{n = 1}^{\infty} p_n \left\vert  \hat{\phi}^n_0   \rangle\langle  \hat{\phi}^n_0   \right\vert 
$
is distributed according to $\mu$.
\end{theorem}

\begin{proof}
Deferred to Section \ref{sec:Proof:Def_Model1}.
\end{proof}

We now establish that  the density operator (\ref{eq:Def_rho})  does not depend on the decomposition of 
the regular density operator $\rho_0$.
Thus,
Model \ref{Modelo1} is well-defined whenever 
we only consider  decompositions of the initial density operator  $\rho_0$
by regular random variables $\hat{\phi}^n_0  $.

\begin{theorem}
 \label{th:Invarianza_ci}
Assume Hypothesis \ref{hyp:CF}.
Suppose that
 $\left( \mathbb{P}, \left( \psi^n_t  \right)_{t \geq 0}^{n \in \mathbb{N}},\left( W_{t}^{\ell}\right) _{t \geq 0 }^{\ell\in\mathbb{N}} \right)$
 and
 $\left( \widetilde{\mathbb{P}}, \left( \chi^n_t  \right)_{t  \geq 0 }^{n \in \mathbb{N}},\left(  \widetilde{ W }_{t}^{\ell}\right) _{t \geq 0}^{\ell\in\mathbb{N}} \right)$
are $C$-solutions of the stochastic system \ref{st:SIWE}  
such that the $\mathfrak{L}_{1}\left( \mathfrak{h}\right)$-valued random variables
$
\sum_{n = 1}^{\infty} \left\vert\psi^n_0  \rangle\langle\psi^n_0  \right\vert 
$
and
$
\sum_{n = 1}^{\infty} \left\vert\chi^n_0  \rangle\langle\chi^n_0  \right\vert 
$
are identically distributed.
Then, the stochastic processes 
$
\left( \sum_{n = 1}^{\infty} \left\vert  \psi^n_t   \rangle\langle \psi^n_t  \right\vert ,\left( W_{t}^{\ell}\right)^{\ell\in\mathbb{N}} \right) _{t\geq 0 }
$
and
$
\left( \sum_{n = 1}^{\infty} \left\vert  \chi^n_t   \rangle\langle \chi^n_t  \right\vert ,\left(  \widetilde{ W }_{t}^{\ell}\right)^{\ell\in\mathbb{N}} \right) _{\geq 0 }
$
have the same finite-dimensional probability distributions.
\end{theorem}

\begin{proof}
Deferred to Section \ref{sec:Proof:Invarianza}.
\end{proof}

\subsection{Existence of a regular solution to the stochastic quantum master equation}
\label{sec:ExistenceRS_QME}

We obtain that that the density operator $\rho_{t}$ given by Model \ref{Modelo1}
is a regular solution of the stochastic quantum master equation (\ref{eq:SQME})
whenever the initial density operator is regular.

\begin{condition}
\label{hyp:Closable}
For any $t \geq 0$, the operators $G\left( t \right), L_{1} \left( t \right), L_{2} \left( t \right), \ldots$
are closable.
\end{condition}

\begin{theorem}
\label{th:Existence_SQME}
Assume Hypotheses \ref{hyp:CF} and \ref{hyp:Closable}.
Suppose that $\mu$ is a probability measure on $ \mathfrak{L}_{1} \left( \mathfrak{h}\right) $
such that
$
\mu \left( \mathfrak{L}_{1,C}^+ \left( \mathfrak{h}\right) \right) = 1
$
and
$
\int_{ \mathfrak{L}_{1} \left( \mathfrak{h}\right) }
 \Tr \left( C \varrho  \, C \right)  \mu \left( d \varrho \right) < \infty
$. 
Then:
\begin{itemize}
 \item[a)] There exists a  $C$-solution $\left( \mathbb{P}, \left( \psi^n_t  \right)_{t \geq 0 }^{n \in \mathbb{N}},\left( W_{t}^{\ell}\right) _{t \geq 0}^{\ell\in\mathbb{N}} \right)$ of the stochastic system \ref{st:SIWE} 
such that 
$\sum_{n = 1}^{\infty} \left\vert\psi^n_0  \rangle\langle\psi^n_0  \right\vert $ is distributed according to $\mu$.
 
 \item[b)] For any $t \geq 0 $ we set
$
 \rho_{t} = \sum_{n = 1}^{\infty} \left\vert\psi^n_t  \rangle\langle\psi^n_t  \right\vert 
$,
that is, $\rho_{t}$ is defined by (\ref{eq:Def_rho}).
Then, $\rho_{t} \in  \mathfrak{L}_{1,C}^+ \left( \mathfrak{h}\right) $  $\mathbb{P}$-a.s.  for any $t \geq 0$,
and 
\begin{eqnarray}
\fl
 \label{eq:NSME-e}
   \rho_t 
 & = 
  \rho_0 
 +
 \int_0^t  
 \left( G\left( s \right) \rho_s + \rho_s   G\left( s \right)^{\ast} 
 + \sum_{ \ell =1}^{ \infty}  L_{\ell} \left( s \right) \rho_s   L_{\ell} \left( s \right) ^{\ast} \right)   
ds
\\
\fl
&  \quad +
\sum_{ \ell =1}^{ \infty}
\int_0^t 
 \left(   L_{\ell} \left( s \right) \rho_s + \rho_s  \,  L_{\ell}\left( s \right)^{\ast}
 - 2  \Re \left( tr \left(  L_{\ell} \left( s \right) \rho_s \right) \right) \rho_s \right) 
 dW_s^{\ell} 
 \qquad 
 \mathrm{for \ all \;} t \geq 0 ,
\nonumber
\end{eqnarray} 
where the integral with respect to $ds$ is interpreted as a Bochner integral in both $ \mathfrak{L}_{1} \left( \mathfrak{h}\right)$
and $ \mathfrak{L}_{2} \left( \mathfrak{h}\right)$,
and  the integrals with respect to $dW_s^{\ell} $ are understood as stochastic integrals in  
$ \mathfrak{L}_{2} \left( \mathfrak{h}\right)$.
Moreover,
\[
t \mapsto \sum_{ \ell =1}^{ \infty} \int_0^t  \left( \rho_s  \,  L_{\ell}\left( s \right)^{\ast}  +   L_{\ell} \left( s \right) \rho_s 
 - 2  \Re \left( tr \left(  L_{\ell} \left( s \right) \rho_s \right) \right) \rho_s \right) 
 dB_s^{\ell} 
\]
is a continuous square integrable $ \mathfrak{L}_{2} \left( \mathfrak{h}\right)$-martingale on any bounded interval.
  
\end{itemize}

\end{theorem}

\begin{proof}
Deferred to Section \ref{sec:Proof:Existence_SQME}.
\end{proof}

\begin{remark}
\label{re:Hip3}
Hypothesis \ref{hyp:Closable} is required for the existence of solutions to the GKSL quantum master equation
(\ref{eq:GKSL}) (see, e.g., \cite{MoraAP}).
Hypothesis \ref{hyp:Closable} is true if $G\left( t \right)^{\ast}$, $L_{1} \left( t \right)^{\ast}$, $L_{2} \left( t \right)^{\ast}, \ldots$ 
are densely defined (see, e.g., Section 3.5.5 of \cite{Kato}).
\end{remark}


\subsection{Examples}
\label{sec:examples}

\subsubsection{Quantum system in coordinate representation.}
\label{sec:Ex1}

The following model describes the continuous measurement of position of a quantum system 
with Schr\"{o}dinger Hamiltonian that is
confined within a one-dimensional box (see, e.g., \cite{Bassi2009,Gough}).
Here, 
the stochastic system \ref{st:SIWE} becomes a system of non-linear stochastic partial differential equations.

\begin{example}
\label{ex:RepCoord}
Take $\mathfrak{h} =  L^{2}\left( \mathbb{R} ,\mathbb{C}\right) $.
Consider the Hamiltonian $ H\left( t \right) $ given by 
\[
H\left( t \right)  f \left( x \right)  = - \alpha \frac{d^2}{dx^2}  f \left( x \right)  + V \left( x \right)  \,  f \left( x \right)  
\]
where $\alpha \in \mathbb{R}$ and $ V:  \mathbb{R}  \rightarrow \mathbb{R} $.
Choose $ L_{\ell} \left(t \right)   = 0 $ for any $\ell \geq 2$,
and set 
\[
L_{1} \left( t \right)  f \left( x \right)  =  \gamma \, x f \left( x \right) 
\]
where $  \gamma \in \mathbb{R} $.
\end{example}

In Example \ref{ex:RepCoord}, we suppose that the potential $V$ satisfies  
$ V  \in C^{2}\left( \mathbb{R},\mathbb{R}\right)$,
\begin{equation*}
\max \left\{  
\left| V \left( x  \right) \right|  , \left| V^{\prime \prime} \left( x  \right) \right|  
\right\}
\leq 
K \left( 1+ \left| x \right| ^{2} \right)
\qquad \qquad \forall x \in \mathbb{R} ,
\end{equation*}
and
$
\left|   V^{\prime} \left( x  \right) \right| \leq K  \left( 1+ \left| x \right| \right)
$
for all $x \in \mathbb{R}$, where $K > 0$.

The operator in $L^{2}\left( \mathbb{R} ,\mathbb{C}\right) $  given by 
\begin{equation}
\label{def:C_E1}
f \mapsto - \frac{d^2}{dx^2} f \left( x \right)  + x^{2} f \left( x \right)
\qquad  \mathrm{for}\; \mathrm{all}\;  f \in C^{\infty}_{c}\left( \mathbb{R} , \mathbb{C} \right) 
\end{equation}
is essentially self-adjoint 
(see, e.g., Th. X.28 of  \cite{ReedSimonVol2}),
where
$C^{\infty}_{c}\left( \mathbb{R} , \mathbb{C} \right)  $ stands for 
the set of all  functions from $\mathbb{R}$ to $ \mathbb{C}$
with compact support and derivatives of any order. 
Let $C$  be the closure of (\ref{def:C_E1}).
From Section 5.1 of \cite{FagMora2013} we have that 
$C$ satisfies Hypotheses \ref{hyp:L-G-C-def} and \ref{hyp:CF}.
Consider an  initial density operator $\widetilde{ \rho }_{0}$  such that
\begin{equation}
\label{eq:Descomp_inicial_Ex1}
\widetilde{ \rho }_{0} = \sum_{k = 1}^{\infty} p_k \left\vert \phi^k_0  \rangle\langle \phi^k_0  \right\vert ,
\end{equation}
where 
$ \left\Vert   \phi^k_0 \right\Vert _{ L^{2}\left( \mathbb{R} ,\mathbb{C}\right)} = 1 $
and 
$ \sum_{k = 0}^{\infty} p_k \mathbb{E} \left\Vert  \left(  - \frac{d^2}{dx^2}  + x^{2}  \right)  \phi^k_0 \left( x \right)
\right\Vert^2 _{ L^{2}\left( \mathbb{R} ,\mathbb{C}\right) }
<
+ \infty 
$.
Hence,
\begin{equation*}
\label{eq:Cond_Ex1}
\mathbb{E} \Tr \left( \left(  - \frac{d^2}{dx^2}  + x^{2}  \right) \widetilde{ \rho }_{0}   \left(  - \frac{d^2}{dx^2}  + x^{2}  \right)  \right)   < \infty .
\end{equation*}
Then,
using Theorems \ref{th:EyU-SIWE}, \ref{th:Def_Model1} and  \ref{th:Invarianza_ci} 
we deduce that Model \ref{Modelo1} is well-defined for Example \ref{ex:RepCoord}
under the assumptions of this section.
Indeed, 
there is a unique (in joint law) density operator $ \rho_{t} $ given by (\ref{eq:Def_rho}) 
such that $ \rho_{0} $  is distributed according to the law of $\widetilde{ \rho }_{0} $.
Moreover,
Hypothesis \ref{hyp:Closable} holds true,
because 
$C^{\infty}_{c}\left( \mathbb{R} , \mathbb{C} \right)$ belongs to the domain of 
$G\left( t \right)^{\ast}$ and $L_{1} \left( t \right)^{\ast}$.
Therefore, 
applying Theorem \ref{th:Existence_SQME}  we deduce that 
the operator $\rho_{t}$ given by (\ref{eq:Def_rho}) satisfies (\ref{eq:SQME}).

%

\subsubsection{Measurement of the radiation field.}

\label{sec:Ex2}

Next, we model the continuous monitoring of the radiation field of a quantum system with Rabi Hamiltonian.
In circuit quantum electrodynamics,
this example describes a quantum non-demolition measurement of the squeezed quadrature observable of a microwave resonator coupled to a superconducting qubit.

\begin{example}
\label{ex:RadF}
Set $ \mathfrak{h} = \ell^2 \left(\mathbb{Z}_+ \right)\otimes \mathbb{C}^{2}$.
Consider 
the Pauli matrices 
$
\sigma^x = 
\left(
\begin{array}{cc}
 0 & 1
 \\
 1 & 0
\end{array}
\right)
$
and
$
\sigma^z = 
\left(
\begin{array}{cc}
 1 &  0
 \\
0  & -1
\end{array}
\right)
$. 
The creation and annihilation operators $a^{\dagger}$, $a$ are the closed operators on $\ell^2 \left(\mathbb{Z}_+ \right)$
defined by
$
a^{\dagger}\, e_{n}
=
 \sqrt{ n+1} \, e_{n+1}  
 $
for any $ n \in \mathbb{Z}_{+} $,
and 
\[
a \, e_{n} 
=
\cases{
 \sqrt{n} \, e_{n-1}  & if  $ n \in  \mathbb{N} $
 \\
 0 &  if  $ n = 0 $ 
 },
\]
where $(e_n)_{n\ge 0}$ is the canonical orthonormal basis of $\ell^2 \left(\mathbb{Z}_+ \right)$.

Choose 
 $
 H \left( t \right)  = 
 \omega_1  \sigma^z /2+ \omega_2 \, a^{\dag} a + g   \left( a^{\dag} + a \right )  \otimes \sigma^x 
 $
 and
\[
 L_1  \left( t \right)
=   
\sqrt{\alpha}  \left( \exp{\left (\mathrm{i} \psi\right )} \, a^{\dag} 
+ \exp{\left (- \mathrm{i}  \psi \right )} a \right) ,
\]
where 
$\omega_1, \omega_2, \alpha > 0 $ and $g, \psi \geq 0$;
as usual, we write $A$, $B$ instead of $A \otimes I$, $I \otimes B$ 
whenever $A$, $B$ are linear operator in $ \ell^2 \left(\mathbb{Z}_+ \right)$ and $\mathbb{C}^{2}$ respectively.
For any $\ell \geq 2$, set $ L_{\ell} \left(t \right)   = 0 $.
The operator $G \left( t \right) $ is defined by (\ref{eq:Def_G}) on the domain of the number operator $N = a^{\dag} a$.
\end{example}

Since the system of Example \ref{ex:RadF} is autonomous 
we will write 
$ H $, $ L_{\ell} $ and $ G $ instead of $ H \left( t \right)$, $ L_{\ell} \left(t \right)$ and $ G \left( t \right) $.
According to Remark \ref{re:Hip3} we have that Hypothesis \ref{hyp:Closable} holds,
because $G^{\ast}$, $L_{1}^{\ast}$, $L_{2}^{\ast}, \ldots$ 
are densely defined.
The domains of the closable operators $ G$ and $ L_{\ell}$ 
contain $\mathcal{D}\left(  N \right)$,
and so the closed graph theorem establishes that 
$G$ and $L_{\ell}$ are bounded operators from $\left( N, \left\Vert \cdot \right\Vert_N \right)$ to $\mathfrak{h}$,
where
$
\left\Vert x \right\Vert^2_N =  \left\Vert x \right\Vert^2 + \left\Vert N \, x \right\Vert^2 
$.
Thus, the condition H2.1 of Hypothesis \ref{hyp:CF} holds,
and  
Lemma 2.2 of  \cite{FagMora2013} yields Hypothesis \ref{hyp:L-G-C-def} with $C=N$.
From Remark \ref{re:CondH2.2} we obtain the condition H2.2 of Hypothesis \ref{hyp:CF}.

Using the commutation relations $\left[a, a^{\dag}  \right] = I$, $\left[a,  N \right] =  a$ 
and $\left[N, a^{\dag}  \right] = a^{\dag}$ we deduce, after a long computation, that
\[
\fl
-\frac{1}{2} N^2 L_{1}^{\ast} L_{1} -\frac{1}{2}  L_{1}^{\ast} L_{1} N^2 +  L_{1}^{\ast} N^2 L_{1}  
=
 -\frac{1}{2}  L_{1}^{\ast} L_{1} + L_{1}^{\ast} \left( z \,  a^{\dag}  - \bar{z} \, a \right)
 + L_{1}^{\ast} \left( z \,  a^{\dag}  - \bar{z} \, a \right) N
\]
and
$
\left[ H, N^2 \right] 
=
g \left( 2 \, a \, N - a - a^{\dag} - 2 \, a^{\dag} N \right)
$
on the domain of $N^3$, 
where $ z = \sqrt{\alpha}  \exp{\left(\mathrm{i} \psi\right)} $.
Therefore, 
\begin{eqnarray}
\nonumber
& 2\Re\left\langle N^{2} x, G \left( t \right) x\right\rangle 
+\sum_{\ell=1}^{\infty }\left\Vert N L_{\ell} \left( t \right) x \right\Vert ^{2}
\\
\nonumber
&  =
2 \, g \, \mathrm{i} \langle \left(  a^{\dag} - a \right) x , N  \otimes \sigma^x  x \rangle
- g \, \mathrm{i} \langle x ,  \left(  a^{\dag} + a \right)   \otimes \sigma^x  x \rangle
\\
& \quad 
+
\langle x, -\frac{1}{2}  L_{1}^{\ast} L_{1} + L_{1}^{\ast}  \left(  z \, a^{\dag} - \bar{z} \, a \right) x \rangle
+
\langle  \left(  \bar{z}  \, a - z \, a^{\dag} \right) L_{1} x , N \, x  \rangle
\label{eq:Ex2_H2.3}
\end{eqnarray}
for all $x \in \mathcal{D}\left( N^3 \right) $.
The operators appearing in the right-hand side of (\ref{eq:Ex2_H2.3})
are bounded operators from $\left( N, \left\Vert \cdot \right\Vert_N \right)$ to $\mathfrak{h}$.
This gives the condition  H2.3 of Hypothesis \ref{hyp:CF}.

Assume that the initial density operator $\rho_{0} $ satisfies 
\begin{equation}
\label{eq:Descomp_inicial_Ex2}
\fl
\rho_{0} = \sum_{k = 1}^{\infty} p_k \left\vert \phi^k_0  \rangle\langle \phi^k_0  \right\vert 
\; \mathrm{with} \;
\left\Vert   \phi^k_0 \right\Vert _{  \ell^2 \left(\mathbb{Z}_+ \right)\otimes \mathbb{C}^{2} } = 1
\; \mathrm{and} \;
\sum_{k = 0}^{\infty} p_k \mathbb{E} \left\Vert  N  \phi^k_0 \left( x \right)
\right\Vert^2 _{  \ell^2 \left(\mathbb{Z}_+ \right)\otimes \mathbb{C}^{2} }
<
+ \infty .
\end{equation}
Since 
$
\mathbb{E} \, \Tr \left( N  \varrho_0   N  \right)     < \infty 
$,
applying Theorems \ref{th:EyU-SIWE}, \ref{th:Def_Model1} and  \ref{th:Invarianza_ci} we obtain
that  Model \ref{Modelo1} is well-defined for Example \ref{ex:RadF} under the condition (\ref{eq:Descomp_inicial_Ex2}).
According to Theorem \ref{th:Existence_SQME} we have that 
 $\rho_{t}$, defined by (\ref{eq:Def_rho}),  satisfies (\ref{eq:SQME}).

\section{Proofs}
\label{sec:Proofs}

\begin{notation}
 We will use the symbols  $K$ and $K \left( \cdot \right)$ to denote no-negative constants
and  non-decreasing  non-negative functions on $\left[ 0, \infty \right[$, respectively.
We write  $\mathbb{I}$ instead of  $\left[ 0,\infty \right[ $ or $\left[ 0,T\right] $  with $T\in \mathbb{R}_{+}$.
\end{notation}

\subsection{Proof of Theorem \ref{th:EyU-SIWE}}
\label{sec:th:EyU-SIWE}

For the sake of completeness,
we start by studying the main properties of the linear operators in $L^2 \left( \mathbb{N} ;  \mathfrak{h} \right)$.
Section \ref{sec:SSE_L2} is devoted to the linear and non-linear  stochastic Schr\"odinger equations on $L^2 \left( \mathbb{N} ;  \mathfrak{h} \right)$
that are strongly related to the stochastic system \ref{st:SIWE}.
Finally,
we prove  Theorem \ref{th:EyU-SIWE}.
It is worth pointing out that
we also use the results of Sections \ref{sec:L2Nh} and \ref{sec:SSE_L2} 
in the proof of Theorems  \ref{th:Invarianza_ci}, \ref{th:Model}  and \ref{th:Existence_SQME}.

\subsubsection{The space of all square-summable sequences of elements of $ \mathfrak{h}$.}
\label{sec:L2Nh}

Set 
\[
L^2 \left( \mathbb{N} ;  \mathfrak{h} \right)
= 
\left\{ 
f : \mathbb{N} \rightarrow  \mathfrak{h} : \sum_{n =1}^{\infty}  \left\Vert f  \left( n \right) \right\Vert^2    < + \infty
\right\} .
\]
Then, $ L^2 \left( \mathbb{N} ;  \mathfrak{h} \right) $ is the 
space of all square-summable sequences of elements of $ \mathfrak{h}$.
For any $f,g \in L^2 \left( \mathbb{N} ;  \mathfrak{h} \right) $ we define
$
\langle f , g  \rangle_{L^2 \left( \mathbb{N} ;  \mathfrak{h} \right)} 
=
\sum_{n =1}^{\infty} \langle   f \left( n \right)  ,  g \left( n \right) \rangle 
$,
and so 
$ \left( L^2 \left( \mathbb{N} ;  \mathfrak{h} \right), \left\langle \cdot,\cdot \right\rangle_{L^2 \left( \mathbb{N} ;  \mathfrak{h} \right)} \right) $
is a separable Hilbert space.
We write  $ \langle f , g  \rangle $ instead of  $ \langle f , g  \rangle_{L^2 \left( \mathbb{N} ;  \mathfrak{h} \right)} $
when no confusion can arise.
From the definition of $  L^2 \left( \mathbb{N} ;  \mathfrak{h} \right) $ we obtain the following lemma.

\begin{lemma}
\label{le:L1}
Suppose that  $ x,y \in L^2 \left( \mathbb{N} ;  \mathfrak{h} \right)  $.
Then,
the series $\sum_{n=1}^{\infty} \left\vert x  \left( n \right) \rangle\langle y  \left( n \right)\right\vert$
converges in $\mathfrak{L}_{1}\left( \mathfrak{h}\right)$,
and
$
tr \left( \sum_{n=1}^{\infty} \left\vert x  \left( n \right) \rangle\langle y  \left( n \right)\right\vert \right)
=
\sum_{n=1}^{\infty} \langle y  \left( n \right) ,  x  \left( n \right) \rangle
$.
\end{lemma}

\begin{proof}
Since 
$
\left\Vert \,  \left\vert x \left( n \right) \rangle\langle y  \left( n \right) \right\vert \, \right\Vert_{  \mathfrak{L}_{1}\left( \mathfrak{h}\right) }
=
\left\Vert  x  \left( n \right) \right\Vert \left\Vert  y  \left( n \right) \right\Vert
$,
\[
\sum_{n=1}^{\infty}
\left\Vert \,  \left\vert x \left( n \right) \rangle\langle y  \left( n \right) \right\vert \, \right\Vert_{  \mathfrak{L}_{1}\left( \mathfrak{h}\right) }
\leq
\frac{1}{2} \sum_{n=1}^{\infty} \left\Vert  x  \left( n \right) \right\Vert^2 
+
\frac{1}{2} \sum_{n=1}^{\infty} \left\Vert  y  \left( n \right) \right\Vert^2 
< \infty .
\]
Therefore,
$
\left( \sum_{n=1}^{N} \left\vert x  \left( n \right) \rangle\langle y  \left( n \right)\right\vert \right)_{N \in \mathbb{N}}
$
converges in the $\mathfrak{L}_{1}\left( \mathfrak{h}\right)$-norm to a trace-class operator $\varrho$,
which is denoted by $\sum_{n=1}^{\infty} \left\vert x  \left( n \right) \rangle\langle y  \left( n \right)\right\vert$.
This yields
\[
tr \left( \varrho \right)
=
\lim_{N \rightarrow + \infty}
tr \left( \sum_{n=1}^{N} \left\vert x  \left( n \right) \rangle\langle y  \left( n \right)\right\vert \right)
=
\sum_{n=1}^{\infty} \langle y  \left( n \right) ,  x  \left( n \right) \rangle .
\]
\end{proof}

Applying standard arguments yields the next two lemmas.

\begin{definition}
\label{def:Operador_en_L2}
Let $A: \mathcal{D} \left( A \right) \subset  \mathfrak{h} \rightarrow \mathfrak{h} $ be a linear operador in $ \mathfrak{h} $. 
The linear operador  $\widetilde{A} : \mathcal{D} \left( \widetilde{A} \right) \subset   L^2 \left( \mathbb{N} ;  \mathfrak{h} \right)  
\rightarrow L^2 \left( \mathbb{N} ;  \mathfrak{h} \right)$ is defined by
\[
 \left(  \widetilde{A} f \right)  \left(  n \right)  =  A f  \left(  n \right)  
 \qquad  \mathrm{for}\ \mathrm{all}\;  n \in  \mathbb{N} 
\]
for any $f$ belonging to 
\[
\fl
\mathcal{D} \left( \widetilde{A} \right)  
=
\left\{ 
f \in  L^2 \left( \mathbb{N} ;  \mathfrak{h} \right) :   f  \left( n \right) \in \mathcal{D} \left(A\right)  
\; \mathrm{for}\ \mathrm{all}\;  n \in \mathbb{N}, \mathrm{and} \;
\sum_{n =1}^{\infty}  \left\Vert A f  \left( n \right) \right\Vert^2    < + \infty
\right\} .
\]

\end{definition}

\begin{lemma}
\label{lem:ClosedOperator}
Let  $A: \mathcal{D} \left( A \right) \subset  \mathfrak{h} \rightarrow \mathfrak{h} $
be a closed operator.
Then
$\widetilde{A} $ is closed.
\end{lemma}

\begin{proof}
Suppose that  the sequence $f_k \in \mathcal{D} \left( \widetilde{A} \right)  $ satisfies
$
f_k  \rightarrow_{k \rightarrow + \infty } f 
$
and
$
 \widetilde{A} \, f_k  \rightarrow_{k \rightarrow + \infty } g
$.
Then,
$ f_k  \left( n \right) \rightarrow_{k \rightarrow + \infty } f  \left( n \right) $
and
$ A \, f_k   \left( n \right) \rightarrow_{k \rightarrow + \infty } g  \left( n \right) $
for all $ n \in \mathbb{N}$.
Since $A$ is closed, $  f  \left( n \right) \in  \mathcal{D} \left( A \right) $ and $  A \, f  \left( n \right)  = g  \left( n \right)$.
This implies $ f \in  \mathcal{D} \left( \widetilde{A} \right) $ and $  \widetilde{A} \, f = g $.
\end{proof}

\begin{lemma}
\label{lem:Core}
Let  $\mathfrak{D}$ be a core of a closed operator 
$A: \mathcal{D} \left( A \right) \subset  \mathfrak{h} \rightarrow \mathfrak{h} $.
Then,
\[
\widetilde{ \mathfrak{D} }
=
\left\{ 
f \in  \mathcal{D} \left( \widetilde{A} \right)   :   f  \left( n \right) \in  \mathfrak{D} 
\quad   \mathrm{for}\ \mathrm{all}\;  n \in \mathbb{N}
\right\} .
\]
is a core of $\widetilde{A} $.
\end{lemma}

\begin{proof}
Consider $f \in \mathcal{D} \left( \widetilde{A}  \right)$.
For every $n \in \mathbb{N}$ there exists a sequence $\left( x^n_k \right)_{k \in \mathbb{N}}$
of elements of $\mathfrak{D}$ such that
$ \left\Vert  x^n_k  - f  \left( n \right) \right\Vert +  \left\Vert A \, x^n_k  - A \, f  \left( n \right) \right\Vert 
\rightarrow_{k \rightarrow + \infty } 0$,
and
\[ 
\left\Vert  x^n_k  - f  \left( n \right) \right\Vert + \left\Vert A \, x^n_k  - A \, f  \left( n \right) \right\Vert \leq 1/n .
\]
We set $f_k \left( n \right) = x^n_k $ for all $k \in \mathbb{N}$.
Hence $f_k \in \widetilde{ \mathfrak{D} }$ for any $k \in \mathbb{N}$.
Using Lebesgue's dominated convergence theorem yields
\[
\int_{\mathbb{N} } \left( 
\left\Vert  f_k \left( n \right)  - f  \left( n \right) \right\Vert^2 + \left\Vert A \, f_k \left( n \right)  - A \, f  \left( n \right) \right\Vert^2
\right)  \mu \left( d \, n \right) 
\rightarrow_{k \rightarrow + \infty } 0 ,
\]
where
$
\mu = \sum_{n=1}^{+ \infty} \delta_{\left\{ n \right\}}
$.
This gives 
$
\left\Vert  f_k  - f  \right\Vert^2 + \left\Vert \widetilde{A}  \, f_k   - \widetilde{A}  \, f  \right\Vert^2
\rightarrow_{k \rightarrow + \infty } 0 
$,
and so $\widetilde{ \mathfrak{D} }$ is a core of $\widetilde{A} $.
\end{proof}

We now establish that the application $A \mapsto \widetilde{A}$ preserves the self-adjointness property.

\begin{lemma}
\label{lem:Adjunto}
 Suppose that  $A: \mathcal{D} \left( A \right) \subset  \mathfrak{h} \rightarrow \mathfrak{h} $  
 is a self-adjoint operator. 
 Then, $\widetilde{A}$ is self-adjoint. 
\end{lemma}

\begin{proof}
 Since  $ \mathcal{D} \left( A \right) $ is dense in  $\mathfrak{h}$,
 $\mathcal{D} \left( \widetilde{A} \right)   $ is dense in $ L^2 \left( \mathbb{N} ;  \mathfrak{h} \right)$.
 For any $f, g \in \mathcal{D} \left( \widetilde{A} \right) $,
 \[
 \langle f , \widetilde{A} \, g  \rangle 
=
\sum_{n =1}^{\infty} \langle   f \left( n \right)  ,  A \, g \left( n \right) \rangle
=
\sum_{n =1}^{\infty} \langle  A \, f \left( n \right)  ,   g \left( n \right) \rangle
=
 \langle \widetilde{A} \, f ,  g  \rangle .
 \]
 This gives 
 $  \widetilde{A} \subset   \left( \widetilde{A} \right)^{\ast}  $.

Suppose that $f \in \mathcal{D} \left( \left( \widetilde{A} \right)^{\ast}  \right)$.
Then, there exists an $g \in L^2 \left( \mathbb{N} ;  \mathfrak{h} \right) $ such that 
\[
\langle f ,   \widetilde{A} \, h  \rangle 
=
\langle  g ,    h \rangle 
\qquad  \mathrm{for}\; \mathrm{all}\;  h \in  \mathcal{D} \left(  \widetilde{A}  \right) .
\]
For any $k \in  \mathbb{N}$ and $x \in  \mathcal{D} \left( A  \right)$ we define
$
h_k \left( n \right) =
\cases{
 x & \quad if  $n = k$
 \\
 0 & \quad if  $n \neq k$
}
$.
Thus, 
\[
\langle f \left( k \right)  ,   A \, x  \rangle
=
\langle f ,   \widetilde{A} \, h_k  \rangle 
=
\langle  g ,    h_k \rangle 
=
\langle  g \left( k \right)  ,   x  \rangle  .
\]
Since $x$  is an arbitrary element of $ \mathcal{D} \left( A  \right)$,
$ f \left( k \right) \in  \mathcal{D} \left( A^{\ast}   \right) = \mathcal{D} \left( A \right)  $ 
and $g \left( k \right) = A^{\ast}  f \left( k \right) = A \, f \left( k \right) $.
Therefore,
$f \in \mathcal{D} \left( \widetilde{A}  \right)$,
and  $ \left( \widetilde{A} \right)^{\ast}  f = g =  \widetilde{A} f $.
This implies $ \left( \widetilde{A} \right)^{\ast}  \subset \widetilde{A}$.
\end{proof}

Finally, we study the closure of $ A \, \sum_{n=1}^{\infty} \left\vert x  \left( n \right) \rangle\langle y  \left( n \right)\right\vert \, B$
with $A, B^{\ast}$  relatively bonded with respect to $C$.

\begin{lemma}
\label{le:A_rho_B}
Let $C$ be a self-adjoint non-negative operator $C$ in $\mathfrak{h}$.
Suppose that $A$ is a linear operator in $\mathfrak{h}$
such that $\mathcal{D}\left(C \right) \subset \mathcal{D}\left( A \right)$
and 
\begin{equation}
 \label{eq:Cota_A}
 \left\Vert A \, u  \right\Vert^{2} \leq K   \left( \left\Vert  u  \right\Vert^{2} + \left\Vert  C u  \right\Vert^{2} \right)
 \qquad 
 \mathrm{for} \ \mathrm{all} \; u \in \mathcal{D}\left( C\right) .
\end{equation}
Assume that $B$ is a densely defined linear operator in  $\mathfrak{h}$ 
satisfying $\mathcal{D}\left(C \right) \subset \mathcal{D}\left( B^{\ast}  \right)$.
For any $ x,y \in \mathcal{D} \left( \widetilde{C} \right)  $ we define 
$
\varrho
=
\sum_{n=1}^{\infty} \left\vert x  \left( n \right) \rangle\langle y  \left( n \right)\right\vert
$.
Then,
the unique bounded extension of $A \rho B$ 
is equal to $\sum_{n=1}^{\infty} \left\vert A \, x  \left( n \right) \rangle\langle B^{\ast}  y  \left( n \right)\right\vert$,
which belongs to  $\mathfrak{L}_{1}\left( \mathfrak{h}\right)$.
\end{lemma}

\begin{proof}
The operator $B^{\ast}$ is a closed operator in $\mathfrak{h}$,
because $B$ is densely defined.
Since $\mathcal{D}\left(C \right) \subset \mathcal{D}\left( B^{\ast}  \right)$,
using the closed graph theorem yields 
\begin{equation}
 \label{eq:Cota_B}
 \left\Vert B^{\ast} \, u  \right\Vert^{2} \leq K   \left( \left\Vert  u  \right\Vert^{2} + \left\Vert  C u  \right\Vert^{2} \right) 
\qquad  \mathrm{for}\ \mathrm{all}\;  u \in \mathcal{D}\left( C\right) .
\end{equation}
Take
$
\varrho_N
=
\sum_{n=1}^{N} \left\vert x  \left( n \right) \rangle\langle y  \left( n \right)\right\vert
$.
Let  $z \in \mathcal{D}\left( C\right)  $.
Since $ x,y \in \mathcal{D} \left( \widetilde{C} \right)  $,
\begin{equation}
 \label{eq:Conv_A}
 \fl
 A \, \varrho_N B \, z
=
\sum_{n=1}^N \langle y \left( n \right) , B \, z \rangle \, A \,  x \left( n \right) 
\longrightarrow_{N \rightarrow + \infty}
\sum_{n=1}^{+ \infty } \langle B^{\ast} y \left( n \right) ,  z \rangle \, A \,  x \left( n \right) .
\end{equation}
Therefore,
$
\lim_{N \rightarrow + \infty } C \, \varrho_N B \, z 
= 
\sum_{n=1}^{+ \infty } \langle B^{\ast} y \left( n \right) ,  z \rangle \, C \,  x \left( n \right)
$.
Since 
$ \lim_{N \rightarrow + \infty } \varrho_N B \, z  = \varrho  B \, z $
and 
$C$ is closed,
$  \varrho  B \, z \in  \mathcal{D} \left( C \right)  $ and
$
C \varrho  B \, z = \sum_{n=1}^{+ \infty } \langle B^{\ast} y \left( n \right) ,  z \rangle \, C \,  x \left( n \right) 
$.
Applying (\ref{eq:Cota_A}) and (\ref{eq:Conv_A}) we deduce that
\[
A \, \varrho B \, z
=
\sum_{n=1}^{+ \infty } \langle B^{\ast} y \left( n \right) ,  z \rangle \, A \,  x \left( n \right) .
\]
Combining Lemma \ref{le:L1}, (\ref{eq:Cota_A}) and (\ref{eq:Cota_B}) we obtain that
$
\sum_{n=1}^{\infty} \left\vert A \, x  \left( n \right) \rangle\langle B^{\ast}  y  \left( n \right)\right\vert
\in 
\mathfrak{L}_{1}\left( \mathfrak{h}\right)
$.
Since $\mathcal{D} \left( C \right)$ is dense in $\mathfrak{h}$,
the unique bounded extension of $A \rho B$ 
is equal to $\sum_{n=1}^{\infty} \left\vert A \, x  \left( n \right) \rangle\langle B^{\ast}  y  \left( n \right)\right\vert$.
\end{proof}

\subsubsection{ Stochastic Schr\"odinger equations on $L^2 \left( \mathbb{N} ;  \mathfrak{h} \right)$.}
\label{sec:SSE_L2}

We first examine the  linear stochastic Schr\"odinger equation {\color{blue} (\ref{eq:LSSE_e})},
where 
the map $A \mapsto \widetilde{ A }$ is given by Definition \ref{def:Operador_en_L2},
$B^1, B^2, \ldots$ are real-valued independent Wiener processes on a filtered complete probability 
space $\left( \Omega ,\mathfrak{F}, \left(\mathfrak{F}_{t}\right) _{t\geq 0},\mathbb{Q}\right) $,
and
$x_{t} $ is a $L^2 \left( \mathbb{N} ;  \mathfrak{h} \right)$-valued adapted stochastic process  with continuous sample paths.

\begin{lemma}
\label{lem:hyp-CF-e}
Suppose that Hypothesis \ref{hyp:CF} is fulfilled.
Choose 
\[
\widetilde{ \mathfrak{D}_{1} }
=
\left\{ 
f \in  \mathcal{D} \left( \widetilde{C^2} \right)   :   
f  \left( n \right) \in  \mathfrak{D_1} \quad  \mathrm{for}\ \mathrm{all}\;  n \in \mathbb{N}
\right\}.
\]
Then, 
Hypothesis \ref{hyp:CF} is valid when we replace 
$\mathfrak{h}$, $C$, $G \left( \cdot \right)$, $L_{\ell} \left( \cdot \right)$ and $\mathfrak{D}_{1}$
by 
$L^2 \left( \mathbb{N} ;  \mathfrak{h} \right)$, $\widetilde{C}$, $\widetilde{ G \left( \cdot \right)}$,  $\widetilde{ L_{\ell} \left( \cdot \right)}$ and $ \widetilde{ \mathfrak{D}_{1} } $, respectively.
\end{lemma}

\begin{proof}
Since $C$ is self-adjoint,
applying Lemma \ref{lem:Adjunto} we obtain that $\widetilde{C}$ is a self-adjoint operator in 
$L^2 \left( \mathbb{N} ;  \mathfrak{h} \right)$.
From Condition H2.1 we obtain that for all $f \in \mathcal{D} \left( \widetilde{C} \right)  $,
\[
 \sum_{n =1}^{\infty}  \left\Vert   G  \left( t \right) f  \left( n \right) \right\Vert^2   
 \leq 
 K \left( t \right) \sum_{n =1}^{\infty} \left( \left\Vert   C \, f  \left( n \right) \right\Vert^2 
 +   \left\Vert  f  \left( n \right) \right\Vert^2  \right)  < + \infty .
\]
This yields  
$\mathcal{D}\left( \widetilde{ C } \right) \subset \mathcal{D}\left(  \widetilde{ G  \left( t \right) } \right)$
and
\[  
\left\Vert  \widetilde{ G  \left( t \right) } \, f  \right\Vert^{2} \leq  
K \left( t \right) \left( \left\Vert  \widetilde{ C } \, f  \right\Vert^{2} +  \left\Vert  f   \right\Vert^2 \right)
\qquad
\mathrm{for} \ \mathrm{all} \; f \in \mathcal{D} \left( \widetilde{C} \right)  .
\]

Combining Condition H2.2 with the Cauchy-Schwarz inequality gives
\[
\fl
\sum_{\ell=1}^{\infty }\left\Vert  L_{\ell} \left( t \right) x \right\Vert ^{2} 
=
-2\Re\left\langle  x, G \left( t \right) x \right\rangle 
\leq
2 \left\Vert  x   \right\Vert  \left\Vert  G \left( t \right) x  \right\Vert
\leq
 \left\Vert  x   \right\Vert^2 +  \left\Vert  G \left( t \right) x  \right\Vert^2 
\]
whenever  $x \in \mathcal{D}\left(C \right)$, and so
$
\sum_{\ell=1}^{\infty }\left\Vert  L_{\ell} \left( t \right) x \right\Vert ^{2} 
\leq
 K \left( t \right)   \left(   \left\Vert  C \, x  \right\Vert^{2} +  \left\Vert  x  \right\Vert^{2} \right) 
$.
Therefore,
$\mathcal{D}\left( \widetilde{ C } \right) \subset \mathcal{D}\left(  \widetilde{ L_{\ell}  \left( t \right) } \right)$,
and 
\[
\sum_{\ell=1}^{\infty }\left\Vert  \widetilde{L_{\ell} \left( t \right)} f \right\Vert ^{2} 
\leq
K \left( t \right) \left( \left\Vert  \widetilde{ C } \, f  \right\Vert^{2} +  \left\Vert  f   \right\Vert^2 \right) 
\qquad \mathrm{ for \ all \;} f \in \mathcal{D} \left( \widetilde{C} \right).
\]
Moreover,
according to Condition H2.2 we have that for any $f \in \mathcal{D} \left( \widetilde{C} \right)$,
\[
\sum_{n=1}^{\infty}
\left(
2\Re\left\langle  f  \left( n \right) ,  G \left( t \right)  f  \left( n \right) \right\rangle 
+\sum_{\ell=1}^{\infty }\left\Vert   L_{\ell} \left( t \right)  f  \left( n \right)  \right\Vert ^{2} 
\right)
=  0 .
\]
This gives
\[
 2\Re\left\langle  f,  \widetilde{ G \left( t \right) } f \right\rangle 
+\sum_{\ell=1}^{\infty }\left\Vert   \widetilde{ L_{\ell} \left( t \right) } f  \right\Vert ^{2} 
=
0
\qquad \mathrm{ for \ all \;}  f \in \mathcal{D} \left( \widetilde{C} \right) .
\]

The map $f \mapsto f \left( n \right)$ is a continuous function 
from $ L^2 \left( \mathbb{N} ;  \mathfrak{h} \right)$ to $ \mathfrak{h} $,
and so it is measurable.
Using Condition H1.2 we obtain that
$ \left( t, f \right) \mapsto G \left( t \right) \circ \pi _{C} \left(  f \left( n \right)  \right) $
is measurable from
$\left(  \left[ 0 , \infty \right[ \times L^2 \left( \mathbb{N} ;  \mathfrak{h} \right)  ,  \mathcal{B} \left(  \left[ 0 , \infty \right[ \times L^2 \left( \mathbb{N} ;  \mathfrak{h} \right)  \right) \right)  $
to
$ \left(  \mathfrak{h},  \mathcal{B} \left( \mathfrak{h} \right) \right) $.
Hence,
for any $g \in  L^2 \left( \mathbb{N} ;  \mathfrak{h} \right) $,
the function
\[
 \left( t, f \right) \mapsto
 \sum_{n=1}^{\infty}  \pi _{\widetilde{C}} \left(  f \right) \langle  g \left( n \right) , G \left( t \right) \circ \pi _{C} \left(  f \left( n \right)  \right) \rangle
 =
\langle  g  ,  G \left( t \right) \circ \pi _{ \widetilde{C} } \left(  f   \right) \rangle
\]
is  $ \mathcal{B} \left(  \left[ 0 , \infty \right[ \times L^2 \left( \mathbb{N} ;  \mathfrak{h} \right) \right) $-measurable.
This implies that
\[
\fl
\widetilde{G \left( \cdot \right)}  \circ \pi _{ \widetilde{C} } : 
\left(
 \left[ 0 , \infty \right[ \times L^2 \left( \mathbb{N} ;  \mathfrak{h} \right)  ,  \mathcal{B} \left(  \left[ 0 , \infty \right[ \times L^2 \left( \mathbb{N} ;  \mathfrak{h} \right)  \right) \right) 
\rightarrow 
 \left(  L^2 \left( \mathbb{N} ;  \mathfrak{h} \right) ,  \mathcal{B} \left(  L^2 \left( \mathbb{N} ;  \mathfrak{h} \right)\right) \right)
\]
is measurable.
Likewise, Condition H1.2 implies that 
\[
\fl
\widetilde{L_{\ell}  \left( \cdot \right)}  \circ \pi _{ \widetilde{C} } : 
\left(
 \left[ 0 , \infty \right[ \times L^2 \left( \mathbb{N} ;  \mathfrak{h} \right)  ,  \mathcal{B} \left(  \left[ 0 , \infty \right[ \times L^2 \left( \mathbb{N} ;  \mathfrak{h} \right)  \right) \right) 
\rightarrow 
  \left(  L^2 \left( \mathbb{N} ;  \mathfrak{h} \right) ,  \mathcal{B} \left(  L^2 \left( \mathbb{N} ;  \mathfrak{h} \right)\right) \right)
\]
is measurable.

Combining Lemma \ref{lem:Core} with Condition H1.4 we obtain that 
$ \widetilde{ \mathfrak{D}_{1} } $ is a core of $\widetilde{C^2} $.
Then
$ \widetilde{ \mathfrak{D}_{1} } $ is a core of $  \left(  \widetilde{C} \right)^2 $,
because  $\widetilde{C^2} =  \left(  \widetilde{C} \right)^2$.
Applying Condition H2.3 we obtain that for any $f \in  \widetilde{ \mathfrak{D}_{1} }$,
\begin{eqnarray*}
& 2\Re\left\langle \left(  \widetilde{C} \right)^2 f , \widetilde{ G \left( t \right) } f \right\rangle 
+\sum_{\ell=1}^{\infty }\left\Vert \widetilde{ C } \widetilde{ L_{\ell} \left( t \right) }  f \right\Vert ^{2}
\\
& =
\sum_{n = 1}^{+ \infty} \left(
2\Re \left\langle  C^{2}  f \left( n \right) , G \left( t \right)  f \left( n \right) \right\rangle 
+\sum_{\ell=1}^{\infty }\left\Vert C \, L_{\ell} \left( t \right)   f \left( n \right) \right\Vert ^{2}
\right)
\\
& \leq  \alpha \left( t \right) \sum_{n = 1}^{+ \infty} \left( 
\left\Vert C \, f \left( n \right) \right\Vert^{2} +  \left\Vert f \left( n \right) \right\Vert ^{2}
\right)
=
 \alpha \left( t \right) \left( \left\Vert \widetilde{ C } \, f   \right\Vert^{2} +  \left\Vert f  \right\Vert ^{2} \right) .
\end{eqnarray*}
\end{proof}

\begin{theorem}
\label{th:EyU-LSSE-e}
Let Hypothesis \ref{hyp:CF} hold true.
Assume that $x_0$ is a $\mathfrak{F}_{0}$-measurable random variable such that 
$x_0 \in \mathcal{D}\left( \widetilde{  C }\right) $  $\mathbb{Q}$-a.s.,
and  $\mathbb{E}_{\mathbb{Q}} \left(  \left\Vert  x_0 \right\Vert^{2} + \left\Vert \widetilde{ C } \, x_0 \right\Vert^{2}\right) <\infty $.
Then
(\ref{eq:LSSE_e}) has a unique strong $\widetilde{C}$-solution 
$\left( x_{t} \right) _{t \in \mathbb{I}}$ with initial datum $x_0$. 
Moreover, 
$\left(  \left\Vert   x_{t}   \right\Vert ^{2}\right) _{t \in \mathbb{I}} $ 
is a martingale, 
and
\[ 
\mathbb{E}_{\mathbb{Q}} \left( \left\Vert \widetilde{C} \, x_{t}  \right\Vert ^{2} \right)
\leq 
\exp \left( t \alpha \left( t \right) \right) \left( \mathbb{E}_{\mathbb{Q}} \left( \left\Vert \widetilde{C } \, x_0  \right\Vert ^{2} \right)
+ t \alpha\left( t \right) \mathbb{E}_{\mathbb{Q}} \left( \left\Vert x_0  \right\Vert ^{2} \right) \right)
\]
for all $t \in \mathbb{I}$.
\end{theorem}

\begin{proof}
By Lemma \ref{lem:hyp-CF-e},
applying  Theorems 2.4 and 2.7 of \cite{FagMora2013}
we obtain  the assertion of this theorem. 
 \end{proof}

\begin{remark}
Under the assumptions of Theorem \ref{th:EyU-LSSE-e},
$ x_{t} \left( n \right)$  is the unique strong $C$-solution of  (\ref{eq:Linear_SSE}) with initial datum $x_0 \left( n \right)$.
\end{remark}

Second, we address the non-linear stochastic Schr\"odinger equation (\ref{eq:NLSSE_e}).
Since Theorem \ref{th:EyU-LSSE-e} holds,
applying Theorem 2.12 of  \cite{FagMora2013} yields the following theorem.

\begin{definition}
\label{def:sol_L2}
Let $C$ satisfy Hypothesis \ref{hyp:L-G-C-def}.
We say that 
$\left( \Omega ,\mathfrak{F},\left( \mathfrak{F}_{t}\right) _{t \in \mathbb{I}},
\left( \hat{x}_{t}\right) _{t \in \mathbb{I}},\left( W_{t}^{\ell}\right) _{t \in \mathbb{I}}^{\ell\in\mathbb{N}}\right) $
is a solution of class $\widetilde{C}$ of (\ref{eq:NLSSE_e}) on $\mathbb{I}$ with initial 
distribution $\theta$  if and only if:
\begin{itemize}
\item $W^{1}, W^2, \ldots $ are real valued independent Brownian motions on
the filtered probability space 
$\left( \Omega,\mathfrak{F},\left( 
\mathfrak{F}_{t}\right) _{t \in \mathbb{I}},\mathbb{P}\right) $,
where $\left(\mathfrak{F}_{t}\right) _{t \in \mathbb{I}}$ satisfies the so-called usual conditions.

\item $\left( \hat{x}_{t}\right) _{t \in \mathbb{I}}$ is an 
$L^2 \left( \mathbb{N} ;  \mathfrak{h} \right)$-valued adapted process with continuous sample paths 
such that  $\hat{x}_{0}$  is distributed according to $\theta$
and  $ \left\Vert \hat{x}_{t}\right\Vert =1$ $\mathbb{P}$-$a.s.$

\item For any $t \in \mathbb{I}$, 
$\hat{x}_{t}\in \mathcal{D}\left( \widetilde{C} \right)$  $\mathbb{P}$-$a.s.$
and 
$\sup_{s\in\left[ 0,t\right] }\mathbb{E}_{\mathbb{P}} \left(  \left\Vert  \widetilde{C}  \hat{x}_{s}\right\Vert ^{2} \right)<\infty$.

\item $\mathbb{P}$-$a.s.$ for all $t \in \mathbb{I}$,
\begin{eqnarray}
\fl
\nonumber
 \hat{x}_t
& =
\hat{x}_0+\int_0^t\left( \widetilde{ G\left( s \right)  } \pi_{\mathcal{D}\left(  \widetilde{C}\right) }\left( \hat{x}_{s} \right) 
+  \widetilde{ g }  \left( s,  \hat{x}_{s}  \right) 
\right)ds
\\
\label{eq:SIWE_fn}
\fl
& \quad
+
\sum_{ \ell =1}^{ \infty} \int_0^t 
\left(
 \widetilde{ L_{\ell} \left( s \right) }  \pi_{\mathcal{D}\left(  \widetilde{C}\right) }\left( \hat{x}_{s} \right) 
 - \Re \left(  \langle  \hat{x}_{s} , \widetilde{ L_{\ell} \left( s \right) }  \pi_{\mathcal{D}\left(  \widetilde{C}\right) }\left( \hat{x}_{s} \right) \rangle  \right)  \hat{x}_{s}
 \right) dW_s^k ,
\end{eqnarray}
where
$
 \pi _{\mathcal{D}\left( \widetilde{ C }\right)} \left( x \right)
 =
 \cases{
 x \quad & if  $ x\in \mathcal{D}\left(  \widetilde{C} \right) $
\\
0 \quad & if  $ x \notin \mathcal{D}\left(  \widetilde{C }\right) $
}
$ and we define $\widetilde{  g } \left( s, z \right)$ to be
\[
\fl
\sum_{ \ell =1}^{ \infty} \left( \Re \left( \langle z, \widetilde{ L_{\ell} \left( s \right) }  \pi_{\mathcal{D}\left(  \widetilde{C}\right) }\left( z \right) \rangle  \right) \widetilde{ L_{\ell} \left( s \right) } \pi_{\mathcal{D}\left(  \widetilde{C}\right) }\left( z \right)
-\frac{1}{2} \Re^2 \left(\langle z, \widetilde{ L_{\ell} \left( s \right) } \pi_{\mathcal{D}\left(  \widetilde{C}\right) }\left( z \right) \rangle \right) z\right) .
\]

\end{itemize}

We shall say, for short, that 
$\left( \mathbb{P}, \left( \hat{x}_{t}\right)_{t\in\mathbb{I}}, \left( W_{t}^{\ell}\right) _{t \in \mathbb{I}}^{\ell\in\mathbb{N}} \right)$ 
is a $\widetilde{C}$-solution of (\ref{eq:NLSSE_e}).
\end{definition}

\begin{theorem}
\label{th:EyU-NLSSE-e} 
Suppose that Hypothesis \ref{hyp:CF}  is fulfilled.
Let $\theta$ be  a probability measure on $L^2 \left( \mathbb{N} ;  \mathfrak{h} \right)$ concentrated on 
$\mathcal{D} \left(\widetilde{C} \right)\cap\left\{ x \in L^2 \left( \mathbb{N} ;  \mathfrak{h} \right)  :\left\Vert x\right\Vert =1\right\} $ 
such that 
\[
\int_{L^2 \left( \mathbb{N} ;  \mathfrak{h} \right)}\left\Vert  \widetilde{C} x\right\Vert ^{2}\theta\left( dx\right) < \infty
.
\] 
Then, 
(\ref{eq:NLSSE_e}) has a unique solution of class $\widetilde{C}$ on $\mathbb{I}$ with initial law $\theta$.
Here, the  uniqueness is interpreted in the weak probabilistic sense
(i.e., the finite-dimensional distributions of the pair 
$\left(  \left( \hat{x}_{t}\right), \left( W_{t}^{\ell}\right)^{\ell\in\mathbb{N}} \right)$ 
are uniquely determined for $\theta$ and the coefficients of (\ref{eq:NLSSE_e})).
\end{theorem}

\subsubsection{Proof of Theorem \ref{th:EyU-SIWE}}

\begin{lemma}
\label{lem:Existence}
Adopt the assumptions of Theorem \ref{th:EyU-SIWE}.
Let $\left( \mathbb{P}, \left( \hat{x}_{t}\right)_{t\in\mathbb{I}},  \left( W_{t}^{\ell}\right) _{t \in \mathbb{I}}^{\ell\in\mathbb{N}} \right)$ 
be a  $\widetilde{C}$-solution  of (\ref{eq:NLSSE_e}) on $\mathbb{I}$ with initial law  $\widetilde{\theta}$,
where $\widetilde{\theta}$ is the restriction of 
$\theta $ to $\mathcal{B} \left( L^2 \left( \mathbb{N} ;  \mathfrak{h} \right) \right)$.
For any $n \in \mathbb{N}$ we set 
$
\psi^n_t = \hat{x}_{t} \left( n \right)
$.
Then,
$\left( \mathbb{P}, \left( \psi^n_t  \right)_{t \in \mathbb{I}}^{n \in \mathbb{N}},\left( W_{t}^{\ell}\right) _{t \in \mathbb{I}}^{\ell\in\mathbb{N}} \right)$ 
is a $C$-solution of the stochastic system (\ref{st:SIWE}) with initial law $\theta $.
\end{lemma}

\begin{remark}
Since
$
L^2 \left( \mathbb{N} ;  \mathfrak{h} \right) \in \mathcal{B} \left( \mathfrak{h}  \right)^{\mathbb{N}}
$
and
$
\mathcal{B} \left( L^2 \left( \mathbb{N} ;  \mathfrak{h} \right) \right)
=
 L^2 \left( \mathbb{N} ;  \mathfrak{h} \right)  \cap \mathcal{B} \left( \mathfrak{h}  \right)^{\mathbb{N}}
$,
$ \mathcal{B} \left( L^2 \left( \mathbb{N} ;  \mathfrak{h} \right) \right) $
is a sub-$\sigma$-algebra of $ \mathcal{B} \left( \mathfrak{h}  \right)^{\mathbb{N}} $.
Therefore,
the restriction of $\theta $ to $\mathcal{B} \left( L^2 \left( \mathbb{N} ;  \mathfrak{h} \right) \right)$
is well-defined.
From
$
\theta \left( L^2 \left( \mathbb{N} ;  \mathfrak{h} \right)  \right) =1
$
we obtain that $\widetilde{\theta}$ is a probability measure defined on  
$\mathcal{B} \left( L^2 \left( \mathbb{N} ;  \mathfrak{h} \right) \right)$.
\end{remark}

\begin{remark}
Applying Theorem \ref{th:EyU-NLSSE-e} we obtain that 
(\ref{eq:NLSSE_e}) has a  $\widetilde{C}$-solution on $\mathbb{I}$ with initial law  $\widetilde{\theta}$,
under the assumptions of Theorem \ref{th:EyU-SIWE}.
This solution is unique in the weak probabilistic sense.
\end{remark}

\begin{proof}
We have that $\psi^n_t  \in \mathcal{D}\left( C\right) $ $\mathbb{P}$-a.s.,
because 
$\hat{x}_{t}\in \mathcal{D}\left( \widetilde{C} \right)$ $\mathbb{P}$-$a.s.$
Hence,
\[
 \pi_{\mathcal{D}\left(  \widetilde{C}\right) }\left( \hat{x}_{t} \right)  \left( n \right)
 =
 \hat{x}_{t}  \left( n \right)
 =
   \pi_{\mathcal{D}\left(  C \right) }\left( \psi^n_t  \right)
 \qquad \qquad \mathbb{P} \mathrm{-} a.s.
\]
Using (\ref{eq:SIWE_fn}) yields (\ref{eq:SIWE_f}), because
\[
\fl
\Re \left(  \langle  \hat{x}_{t} , \widetilde{ L_{\ell} \left( t \right) }  \pi_{\mathcal{D}\left(  \widetilde{C}\right) }\left( \hat{x}_{t} \right) \rangle \right)
=
\Re \sum_{n=1}^{\infty}   \langle   \psi^n_t  ,  L_{\ell} \left( t \right)   \pi_{\mathcal{D}\left(  C \right) }\left(  \psi^n_t  \right) \rangle
=
\wp_{\ell} \left( t  \right)  \qquad  \mathbb{P}\mathrm{-}a.s.
\]

For all $ f \in  \mathfrak{h}^ \mathbb{N} $ we set  $ pr_n \left( f  \right) = f  \left( n  \right) $.
Then,
for any $ A \in \mathcal{B} \left( \mathfrak{h}  \right)$ we have
\[
\left( \psi^n_t \right)^{-1} \left( A  \right)
=
\left( pr_n \circ \hat{x}_t \right)^{-1} \left( A  \right)
=
\left( \hat{x}_t \right)^{-1} \left( \left( pr_n \right)^{-1} \left( A  \right)  \cap L^2 \left( \mathbb{N} ;  \mathfrak{h} \right)  \right) .
\]
As 
$ \left( pr_n \right)^{-1} \left( A  \right) \cap L^2 \left( \mathbb{N} ;  \mathfrak{h} \right) \in \mathcal{B} \left( L^2 \left( \mathbb{N} ;  \mathfrak{h} \right) \right)$,
$\left(  \psi^n_t  \right)_{t \in \mathbb{I}} $  is an  $ \mathfrak{h}  $-valued adapted process.

Using Definition \ref{def:sol_L2} we finish checking  that 
$\left( \mathbb{P}, \left( \psi^n_t  \right)_{t \in \mathbb{I}}^{n \in \mathbb{N}},\left( W_{t}^{\ell}\right) _{t \in \mathbb{I}}^{\ell\in\mathbb{N}} \right)$ 
is a $C$-solution of the stochastic system (\ref{st:SIWE}) with initial law $\theta $.
\end{proof}

Next, we prove that the unique (in the weak probabilistic sense) $C$-solution to  the stochastic system (\ref{st:SIWE})  is 
the solution given by Lemma \ref{lem:Existence}.

\begin{proof}[Proof of Theorem \ref{th:EyU-SIWE}]
By Lemma \ref{lem:Existence},
the stochastic system (\ref{st:SIWE}) has a  solution of class $C$ with initial law $\theta$.
On the other hand, suppose that 
$\left( \mathbb{P}, \left( \psi^n_t  \right)_{ t\in \mathbb{I} }^{n \in \mathbb{N}},\left( W_{t}^{\ell}\right) _{t\in \mathbb{I}}^{\ell\in\mathbb{N}} \right)$ 
is a solution of class $C$  to the stochastic system (\ref{st:SIWE}) with initial law $\theta$.
Take
\[
\widetilde{\Omega} = \left\{ \omega \in \Omega: 
\sum_{n=1}^{\infty}  \left\Vert \psi^n_t \left( \omega \right) \right\Vert ^{2} = 1 \quad  \mathrm{for}\ \mathrm{all}\; t  \in \mathbb{I} 
 \right\} ,
 \]
and for any $t\in \mathbb{I}$ we set
\begin{equation*}
 \hat{x}_{t}  \left( \omega \right) 
=
\cases{
  \left( \psi^n_t   \left( \omega \right) \right)_{n \in \mathbb{N}} \quad
&
  if   $ \omega \in \widetilde{\Omega} $
 \\
 0 \quad
 &
  if   $\omega \in \Omega \setminus \widetilde{\Omega} $
} .
\end{equation*}
Thus,
$\left( \hat{x}_{t} \right)_{\in \mathbb{I}}$ has sample paths in $L^2 \left( \mathbb{N} ;  \mathfrak{h} \right) $.
Since  $\mathbb{P}\left( \widetilde{\Omega}^c  \right) = 0$, 
$ \widetilde{\Omega}  \in  \mathfrak{F}_{0}$
and
$\left( \hat{x}_{t} \right)_{t \in \mathbb{I}}$ and $\left( \left( \psi^n_t \right)_{n \in \mathbb{N}} \right)_{t \in  \mathbb{I} }$
are indistinguishable.
Due to
$
\mathcal{B} \left( L^2 \left( \mathbb{N} ;  \mathfrak{h} \right) \right)
\subset 
\mathcal{B} \left( \mathfrak{h}  \right)^{\mathbb{N}}
$,
$\left( \hat{x}_{t} \right)_{\in \mathbb{I}}$ is an $L^2 \left( \mathbb{N} ;  \mathfrak{h} \right)$-valued adapted stochastic process.

Choose $ \omega \in \widetilde{\Omega}$.
We claim that 
$\left( \hat{x}_{t} \left( \omega  \right) \right) _{t \in \mathbb{I}}$ is  continuous in $L^2 \left( \mathbb{N} ;  \mathfrak{h} \right)$.
Consider $t \in \mathbb{I}$ and a sequence $\left( s_n \right)_{n \in \mathbb{N}}$ of elements of $\mathbb{I}$
satisfying $s_n \rightarrow_{n \rightarrow \infty} t$.
Since $ \left\Vert \hat{x}_{s_n}  \left( \omega \right)  \right\Vert = 1$,
there exists a subsequence $\left( \hat{x}_{s_{n \left( k \right) }}   \left( \omega \right) \right)_{k \in \mathbb{N}}$ of 
$\left( \hat{x}_{s_{n  }}   \left( \omega \right) \right)_{n \in \mathbb{N}}$ such that
$\hat{x}_{s_{n \left( k \right) }}  \left( \omega \right)$ converges weakly to some $x \in L^2 \left( \mathbb{N} ;  \mathfrak{h} \right) $
as $k \rightarrow \infty$
(see, e.g., \cite{Kato}).
Using the continuity of $s \mapsto \psi^n_s $ yields $x =  \hat{x}_{t}  \left( \omega \right) $.
Due to 
$
\limsup_{k \rightarrow \infty} \left\Vert \hat{x}_{s_{n \left( k \right) }}  \left( \omega \right)  \right\Vert
= 1 
= \left\Vert \hat{x}_{t}  \left( \omega \right)  \right\Vert 
$,
$\hat{x}_{s_{n \left( k \right) }}  \left( \omega \right)$ converges to  $\hat{x}_{t}  \left( \omega \right)$
in $L^2 \left( \mathbb{N} ;  \mathfrak{h} \right) $
(see, e.g., \cite{Kato}).
Therefore, any sequence $\left( \hat{x}_{s_{n  }}   \left( \omega \right) \right)_{n \in \mathbb{N}}$ 
has a subsequence converging to  $\hat{x}_{t}  \left( \omega \right)$.
Hence, $ t \mapsto  \hat{x}_{t}  \left( \omega \right)$ is continuous $L^2 \left( \mathbb{N} ;  \mathfrak{h} \right)$.

From
$
\mathbb{E}_{\mathbb{P}}  \left( \sum_{n=1}^{\infty}  \left\Vert C \psi^n_t \right\Vert ^{2} \right) < \infty 
$
we have
$\sum_{n=1}^{\infty}  \left\Vert C \hat{x}_{t}  \left( n \right) \right\Vert ^{2} < \infty $
$\mathbb{P}$-a.s.,
and so $\hat{x}_{t}  \in  \mathcal{D} \left( \widetilde{C} \right)$ $\mathbb{P}$-a.s.
Therefore, 
$
 \pi_{\mathcal{D}\left(  \widetilde{C}\right) }\left( \hat{x}_{t} \right)  \left( n \right)
 =
 \hat{x}_{t}  \left( n \right)
 =
   \pi_{\mathcal{D}\left(  C \right) }\left( \psi^n_t  \right)
$  
$\mathbb{P}$-a.s.
Applying (\ref{eq:SIWE_f}) gives (\ref{eq:SIWE_fn}).
Hence,
$\left( \mathbb{P}, \left( \hat{x}_{t}\right)_{t\in\mathbb{I}},\left( W_{t}\right)_{t\in\mathbb{I}}\right)$ 
is a $\widetilde{C}$-solution of (\ref{eq:NLSSE_e}) with initial law  $\widetilde{\theta}$.
Since $\left( \hat{x}_{t} \right)_{t \in \mathbb{I}}$ and $\left( \left( \psi^n_t \right)_{n \in \mathbb{N}} \right)_{t \in  \mathbb{I} }$
are indistinguishable, 
applying Theorem \ref{th:EyU-NLSSE-e} we deduce the uniqueness of the finite-dimensional distributions of
$\left(\left( \psi^n_{t\left(k \right)}  \right)_{1\leq k \leq N }^{n \in \mathbb{N}},
\left( W_{t\left(k \right)}^{\ell}\right) _{1\leq k \leq N}^{\ell\in\mathbb{N}} \right)$. 
\end{proof}


\subsection{Proof of Theorem \ref{th:regular-sol-LSIWE}}
\label{sec:Proof:LSIWE}

\begin{proof}
From Theorem 2.4 of \cite{FagMora2013} we get that (\ref{eq:Linear_SSE})
has a unique strong $C$-solution $\left( \phi^n_t \left( \psi^n_0 \right) \right)_{ t \in \mathbb{I}}$ on  $\mathbb{I}$ 
with initial datum $ \psi^n_0$, and
\[ 
\mathbb{E} \left( \left\Vert C \phi_{t} \left( \psi^n_0 \right)   \right\Vert ^{2} \right)
\leq 
\exp \left( t \alpha \left( t \right) \right) 
\left( \mathbb{E} \left( \left\Vert C  \psi^n_0    \right\Vert ^{2} \right)
+ t \alpha\left( t \right)  \mathbb{E} \left( \left\Vert  \psi^n_0    \right\Vert ^{2} \right) \right)
\]
for all $t \in \mathbb{I}$.
This yields (\ref{eq:LS_cota}).
According to Theorem 2.7 of \cite{FagMora2013} we have that
$\left( \left\Vert \phi_{t} \left( \psi^n_0 \right)   \right\Vert ^{2} \right) _{t \in \mathbb{I}} $ 
is a martingale.
Then,
$
\mathbb{E} \sum_{n=1}^{\infty} \left\Vert \phi_{t} \left( \psi^n_0 \right)   \right\Vert ^{2}
= \sum_{n=1}^{\infty} \mathbb{E} \left\Vert \psi^n_0 \right\Vert ^{2} 
 <  \infty
$,
and so applying the monotone convergence theorem for conditional expectations (see, e.g., \cite{Dudley})
we deduce that $\left(  \sum_{n=1}^{\infty} \left\Vert \phi_{t} \left( \psi^n_0 \right)   \right\Vert ^{2} \right) _{t \in \mathbb{I}} $ 
is a martingale.
\end{proof}

\subsection{Proof of Theorem \ref{th:Model}}
\label{sec:Proof_Model}

\begin{proof}
Since $\left( \psi^n_0  \right)_{n \in \mathbb{N}} \in \mathcal{D} \left( \widetilde{C} \right)  $,
applying Theorem \ref{th:EyU-LSSE-e} we obtain that 
(\ref{eq:LSSE_e}) has a unique strong $\widetilde{C}$-solution 
$\left( x_{t} \right) _{t \in \mathbb{I}}$ with initial datum $\left( \psi^n_0  \right)_{n \in \mathbb{N}}$. 
Then, 
$\left( x_{t} \left( n \right) \right) _{t \in \mathbb{I}}$ is a  strong $C$-solution of  (\ref{eq:Linear_SSE})
with initial datum  $\psi^n_0$,
and so $  x_{t} \left( n \right)  =  \phi_{t} \left( \psi^n_0 \right)  $
due to the  strong $C$-solution of (\ref{eq:Linear_SSE}) is unique (see, e.g., Theorem 2.4 of \cite{FagMora2013}).

Proceeding as in the proof of Proposition 1 of \cite{MoraReAAP}  we deduce that 
\[
\left( \Omega,\mathfrak{F},\left( \mathfrak{F}_{t}\right) _{t\in\left[ 0,T\right] }, \left\Vert x_{T}
\right\Vert ^{2}\cdot\mathbb{P },\left( \hat{x}_{t}\right) _{t\in\left[ 0,T\right] },\left( W_{t}^{\ell}\right) _{t\in\left[ 0,T\right] }^{\ell\in \mathbb{N}}\right) 
\]
is the $\widetilde{C}$-solution of (\ref{eq:NLSSE_e}) on $\left[ 0,T\right]$ with initial datum $\theta$,
where $\theta$ is the probability distribution of $\left( \psi^n_0  \right)_{n \in \mathbb{N}}$,
$
\hat{x}_{t}=\left\{ 
\begin{array}{ll}
x_{t} / \left\Vert x_{t} \right\Vert 
& \;\mathrm{if } \; x_{t}  \neq 0 \\ 
0, & \;\mathrm{if } \; x_{t}  =0
\end{array}
\right. 
$,
and
$
W_{t}^{\ell} = B_{t}^{\ell}
-\int_{0}^{t}\frac{1}{\left\Vert x_{s} \right\Vert ^{2}}d\left[ B^{\ell}, 
\left\| x  \right\| ^2 \right]_{s} 
$.
Applying Lemma \ref{lem:Existence} it follows that 
$
\left( \left\Vert x_{T}\right\Vert ^{2}\cdot\mathbb{Q},\left( \hat{x}_{t} \left( n \right)\right) _{t\in\left[ 0,T\right] }^{n \in \mathbb{N}},\left( W_{t}^{\ell}\right) _{t\in\left[ 0,T\right] }^{\ell\in \mathbb{N}}\right) $
is a $C$-solution of the stochastic system (\ref{st:SIWE})  on $\left[ 0,T\right]$ with initial law $\theta $,
which is unique in the joint law sense by Theorem \ref{th:EyU-SIWE}.

For a given orthonormal basis $\left( e_n \right)_{n \in \mathbb{N}}$ of $\mathfrak{h}$
we define 
\begin{equation}
 \label{eq:Def_Conjugado}
  \overline{ u } = \sum_{n=1}^{\infty}  \overline{ \langle e_n , u \rangle } e_n
 \qquad \qquad
 \qquad \mathrm{for \ all \;} u \in \mathfrak{h} .
\end{equation}
From (\ref{eq:LSSE_e}) we obtain 
\begin{equation*}
	\overline{  x_t } 
	=
	\overline{ x_0 } + \int_0^t \overline{  \widetilde{ G\left( s \right) } \, x_s } \, ds 
	+
	\sum_{ \ell =1}^{ \infty} \int_0^t  \overline{  \widetilde{ L_{\ell} \left( s \right) } \, x_s } \, dB_s^{\ell} 
	\qquad 
	\mathrm{for} \ \mathrm{all} \; t \geq 0 .
\end{equation*}
Let 
$ 
F : L^2 \left( \mathbb{N} ;  \mathfrak{h} \right) \times  L^2 \left( \mathbb{N} ;  \mathfrak{h} \right) \rightarrow \mathbb{C} 
$
be given by
$
 F \left( x , y \right) = \langle \overline{x } ,  y  \rangle
 $.
Then, $F$ is a bilinear form on $L^2 \left( \mathbb{N} ;  \mathfrak{h} \right) $.
Applying the It\^{o} formula  to $F \left( \overline{x_t} , x_t \right)$ yields 
\[
\left\Vert x_{s} \right\Vert ^{2} 
= 
\left\Vert x_{0} \right\Vert ^{2} 
+
\sum_{ \ell =1}^{ \infty} \int_0^t 
2 \Re \left(  \langle  x_{s} , \widetilde{ L_{\ell} \left( s \right) }  x_{s}  \rangle  \right)   dB_s^k .
\]
Hence,
$
W_{t}^{\ell} = B_{t}^{\ell}
- \int_{0}^{t}  2 \Re \left(  \langle  \hat{x}_{s} , \widetilde{ L_{\ell} \left( s \right) }  \hat{x}_{s}  \rangle  \right) ds 
$.
Using  $  x_{t} \left( n \right)  =  \phi_{t} \left( \psi^n_0 \right)  $ completes the proof.
\end{proof}

\subsection{Proof of Theorem \ref{th:Def_Model1}}
\label{sec:Proof:Def_Model1}

First, we obtain that
any  $C$-regular random density operator $\varrho_0$
is equal to a statistical mixture of pure states represented by $ \mathfrak{h} $-valued regular random variables.

\begin{lemma}
\label{lem:Rep_Estado_Inicial}
Suppose that $\varrho_0$ is an $\mathfrak{L}_{1}\left( \mathfrak{h}\right)$-valued random variable such that
$\varrho_0 \in \mathfrak{L}_{1,C}^+ \left( \mathfrak{h}\right) $,
where $C$ is a self-adjoint non-negative operator in $\mathfrak{h}$. 
Then, 
there exist a decreasing sequence $\left( p_n \right)_{n \in \mathbb{N}}$ of non-negative random variables 
and a sequence $\left( \hat{\phi}^n_0 \right)_{n \in \mathbb{N}}$ of  $\mathfrak{h}$-valued random variables with unit norm  
such that
\begin{equation}
	\label{eq:Rep_Ei}
	\varrho_0 = \sum_{n =1}^{\infty} p_n \left\vert  \hat{\phi}^n_0  \rangle\langle  \hat{\phi}^n_0  \right\vert 
\end{equation}
and
$ 
\sum_{n=1}^{\infty} p_n \left\Vert C   \hat{\phi}^n_0 \right\Vert^2 
=
\Tr \left( C \varrho_0  C \right) 
< + \infty
$.
\end{lemma}

\begin{proof}
Since $\varrho_0$ is a  $\mathfrak{L}_{1} \left( \mathfrak{h}\right)$-valued non-negative random variable,
from, for instance, Proposition 1.8 of \cite{DaPrato} we have that
there exists  a decreasing sequence $\left( p_n \right)_{n \in \mathbb{N}}$ of non-negative random variables and a sequence $\left( \phi^n_0 \right)_{n \in \mathbb{N}}$ of  $\mathfrak{h}$-valued random variables such that 
\[
\varrho_0 = \sum_{n =1}^{\infty} p_n \left\vert  \phi^n_0  \rangle\langle  \phi^n_0  \right\vert ,
\]
$
\left\langle \phi^n_0, \phi^m_0 \right\rangle = 0
$
whenever $n \neq m$,
$
\left\Vert \phi^n_0  \left( \omega \right)  \right\Vert = 1
$
in case 
$
p_n  \left( \omega \right) > 0
$,
and
$
\left\Vert \phi^n_0  \left( \omega \right)  \right\Vert = 0
$
if
$
p_n  \left( \omega \right) = 0
$.
As 
$ \varrho_0 \, \phi_0^n = p_n \left\Vert  \phi_0^n \right\Vert^2  \phi_0^n $
and the range of $ \varrho_0 $ is contained in $ \mathcal{D}\left( C \right)$ (see, e.g., \cite{MoraAP}),
$ \phi_0^n \in  \mathcal{D}\left( C \right) $.
Fix an element $e$ of $ \mathcal{D}\left( C \right)$ with unit norm.
Then, for any $n \in \mathbb{N}$ we set 
\[
\hat{\phi}^n_0 \left( \omega \right)
=
\left\{ 
\begin{array}{ll}
	\phi^n_0  \left( \omega \right)
	& \;\mathrm{if } \; p_n  \left( \omega \right) > 0 \\ 
	e & \;\mathrm{if } \;  p_n  \left( \omega \right)  =0
\end{array}
\right. .
\]
Then,
the $\mathfrak{h}$-valued random variable $\hat{\phi}^n_0 $ satisfies (\ref{eq:Rep_Ei}).

According to Lemma \ref{le:A_rho_B} we have that
the unique bounded extension of $C \varrho_0 C $ belongs to  $\mathfrak{L}_{1} \left( \mathfrak{h}\right)$.
Using (\ref{eq:Rep_Ei}) we deduce that
\begin{eqnarray*}
	\fl
	\Tr \left( C \varrho_0 C \right)
	& = 
	\sum_{k=1}^{\infty} \langle e_k , C \varrho_0 C  e_k \rangle
	=
	\sum_{k=1}^{\infty} \langle C e_k ,  \varrho_0 C  e_k \rangle
	\\
	\fl
	& =
	\sum_{k=1}^{\infty} \sum_{n=1}^{\infty} p_n \langle C e_k ,   \hat{\phi}^n_0  \rangle \langle  \hat{\phi}^n_0 ,  C  e_k \rangle
	=
	\sum_{n=1}^{\infty} p_n \sum_{k=1}^{\infty}  \langle  e_k , C  \hat{\phi}^n_0  \rangle \langle C  \hat{\phi}^n_0 ,    e_k \rangle ,
\end{eqnarray*}
where 
$\left( e_n \right)_{n \in \mathbb{N}}$ is  an orthonormal basis of $ \mathfrak{h} $ formed by elements of $ \mathcal{D}\left( C \right)$.
Applying  the Parseval identity yields 
$
\Tr \left( C \varrho_0 C \right) 
= 
\sum_{n=1}^{\infty} p_n \left\Vert C  \hat{\phi}^n_0 \right\Vert^2 
$.
\end{proof}

\begin{proof}[Proof of Theorem \ref{th:Def_Model1}]

For any $\omega \in \mathfrak{L}_{1} \left( \mathfrak{h}\right) $ we set 
\[
\varrho_0 \left( \omega \right)
=
\cases{
  \omega &  if  $ \omega \in \mathfrak{L}_{1,C}^+ \left( \mathfrak{h}\right) $
\\
 0  &  if  $ \omega \notin \mathfrak{L}_{1,C}^+ \left( \mathfrak{h}\right) $
} .
\]
Then, 
$
\varrho_0: 
\left( \mathfrak{L}_{1}\left( \mathfrak{h}\right) ,  \mathcal{B}  \left( \mathfrak{L}_{1}\left( \mathfrak{h}\right) \right) , \mu \right)
\rightarrow
\left( \mathfrak{L}_{1}\left( \mathfrak{h}\right) ,  \mathcal{B}  \left( \mathfrak{L}_{1}\left( \mathfrak{h}\right) \right) \right)
$
is a random variable with distribution $\mu$.
Applying Lemma \ref{lem:Rep_Estado_Inicial} we obtain that there exist  random variables
$ p_n : \left( \mathfrak{L}_{1}\left( \mathfrak{h}\right) ,  \mathcal{B}  \left( \mathfrak{L}_{1}\left( \mathfrak{h}\right) \right) , \mu \right)
\rightarrow
\left( \mathbb{R} ,  \mathcal{B}  \left( \mathbb{R}\right) \right)
$
and
$
 \hat{\phi}^n_0 : \left( \mathfrak{L}_{1}\left( \mathfrak{h}\right) ,  \mathcal{B}  \left( \mathfrak{L}_{1}\left( \mathfrak{h}\right) \right) , \mu \right)
	\rightarrow
	\left( \mathfrak{h} ,  \mathcal{B}  \left( \mathfrak{h}\right) \right)
$
that satisfy (\ref{eq:Rep_Ei}), together with 
$p_n \geq 0$, 
$\left\Vert \hat{\phi}^n_0 \right\Vert = 1$,
and 
$ 
\sum_{n=1}^{\infty} p_n \left\Vert C   \hat{\phi}^n_0 \right\Vert^2 
=
\Tr \left( C \varrho_0  C \right) 
< + \infty
$.
Hence,
\[
\mathbb{E}_{\mu}  \left( \sum_{n=1}^{\infty} p_n \left\Vert C  \hat{\phi}^n_0 \right\Vert^2 \right)
=
\mathbb{E}_{\mu}  \Tr \left( C \varrho_0  C \right) 
=
\int_{ \mathfrak{L}_{1} \left( \mathfrak{h}\right) }
\Tr \left( C \varrho  \, C \right)  \mu \left( d \varrho \right) < \infty .
\]

Using Theorem \ref{th:EyU-SIWE} we deduce that the stochastic system (\ref{st:SIWE}) has a unique (in joint law) solution 
$\left( \mathbb{P}, \left( \psi^n_t  \right)_{t\in \mathbb{I} }^{n \in \mathbb{N}},\left( W_{t}^{\ell}\right) _{t\in \mathbb{I}}^{\ell\in\mathbb{N}} \right)$ of class $C$ such that 
$\left( \psi_0^n \right)_{n \in \mathbb{N}}$ is distributed according to the law of 
$\left( \sqrt{p_n} \hat{\phi}^n_0 \right)_{n \in \mathbb{N}}$.
The theorem follows from the fact that 
$\sum_{n = 1}^{\infty} \left\vert\psi^n_0  \rangle\langle\psi^n_0  \right\vert $ and $ \varrho_0 $ have the same probability distribution on $\mathfrak{L}_{1}\left( \mathfrak{h}\right) $.
\end{proof}

\subsection{Proof of Theorem \ref{th:Invarianza_ci}}
\label{sec:Proof:Invarianza}

We first examine the stochastic evolution operator $\Phi_t$ associated to the linear stochastic Schr\"odinger equation (\ref{eq:Linear_SSE}). In the physical literature, $\Phi_t$ is called propagator.

\begin{lemma}
	\label{OperadorPhi_t}
	Let Hypothesis \ref{hyp:CF} hold.
	For any $t > 0$ we set 
	\[
	\Phi_t \xi = \phi_{t} \left( \xi \right) ,
	\]
	where  $ \xi $ is an $\mathfrak{h}$-valued $\mathfrak{F}_{0}$-measurable random variable 
	satisfying
	$
	\mathbb{E} \left( \left\Vert C \xi \right\Vert ^{2} \right) +  \mathbb{E} \left( \left\Vert \xi \right\Vert ^{2} \right) < \infty 
	$,
	and
	$\phi_{t} \left( \xi \right) $ is the strong $C$-solution of  (\ref{eq:Linear_SSE}) with initial datum $\xi$.
	Then,
	$\Phi_t$ can be extended uniquely to an isometric operator on 
	$
	L^2 \left(  \Omega , \mathfrak{F}_{0}, \mathbb{Q}\right) 
	$
	to
	$
	L^2 \left(  \Omega , \mathfrak{F}_{t}, \mathbb{Q}\right) 
	$,
	which is denoted by $\Phi_t$ for simplicity of notation.
\end{lemma}

\begin{proof}
	Consider $\widetilde{\xi} \in L^2 \left(  \Omega , \mathfrak{F}_{0}, \mathbb{Q}\right)$.
	Since 
	$ \left\Vert n \left( n I +C  \right)^{-1} \right\Vert \leq 1 $
	and
	$\lim_{n \rightarrow + \infty} n \left( n I +C  \right)^{-1} \widetilde{\xi}  = \widetilde{\xi} $
	(see, e.g., \cite{Pazy}),
	applying the dominated convergence theorem yields 
	\[
	\mathbb{E} \left( \left\Vert n \left( n I +C  \right)^{-1}  \widetilde{\xi} -   \widetilde{\xi} \right\Vert ^{2} \right)
	\rightarrow_{n \rightarrow + \infty} 0 .
	\]
	As $C n \left( n I +C  \right)^{-1} $ is a bounded linear operator,
	\[
	\mathbb{E} \left( \left\Vert C n \left( n I +C  \right)^{-1}  \widetilde{\xi}  \right\Vert ^{2} \right)
	\leq \left\Vert C n \left( n I +C  \right)^{-1}   \right\Vert \mathbb{E} \left( \left\Vert \widetilde{\xi}  \right\Vert ^{2} \right)
	< \infty .
	\]
	This gives
	\[
	\mathbb{E} \left( 
	\left\Vert  n \left( n I +C  \right)^{-1}  \widetilde{\xi} \right\Vert ^{2} \right) 
	+  
	\mathbb{E} \left( \left\Vert C \,  n \left( n I +C  \right)^{-1}  \widetilde{\xi} \right\Vert ^{2} \right) < \infty ,
	\]
	and so the domain of the linear operator
	$
	\Phi_t: 
	L^2 \left(  \Omega , \mathfrak{F}_{0}, \mathbb{Q}\right) 
	\rightarrow
	L^2 \left(  \Omega , \mathfrak{F}_{t}, \mathbb{Q}\right)
	$
	is dense in $ L^2 \left(  \Omega , \mathfrak{F}_{0}, \mathbb{Q}\right)  $.
	Since 
	$  
	\mathbb{E} \left( \left\Vert \Phi_t \xi \right\Vert ^{2} \right) 
	=
	\mathbb{E} \left( \left\Vert \xi \right\Vert ^{2} \right) 
	$
	for any $ \xi \in \mathcal{D}\left( \Phi_t \right) $
	(see, e.g., \cite{FagMora2013}),
	$\Phi_t$ can be extended uniquely to an isometric operator on 
	$
	L^2 \left(  \Omega , \mathfrak{F}_{0}, \mathbb{Q}\right) 
	$
	to
	$
	L^2 \left(  \Omega , \mathfrak{F}_{t}, \mathbb{Q}\right) 
	$. 
\end{proof}

We now obtain that the operator $\eta_t$ given by (\ref{def:eta})
is equal to $\Phi_t \rho_0 \Phi_t^{\star}$ under general conditions.
Then, $\eta_t$ does not depend on the decomposition 
$ \sum_{n = 1}^{\infty} p_n \left\vert  \phi^n_0   \rangle\langle \phi^n_0  \right\vert $
of the initial density operator $\rho_0$.

\begin{lemma}
\label{le:Invarianza_lineal} 
Let Hypothesis \ref{hyp:CF} hold.
Consider the sequence of random variables
$
\psi^n_0 :  \left( \Omega ,\mathfrak{F}_0 \right) \rightarrow \left( \mathfrak{h} , \mathcal{B} \left( \mathfrak{h}  \right) \right)
$ 
that satisfy
$
\sum_{n=1}^{\infty} \left\Vert \psi^n_0 \right\Vert ^{2} 
=
1
$ a.s.,
and
$
\sum_{n=1}^{\infty} \mathbb{E} \left( \left\Vert C \psi^n_0 \right\Vert ^{2} \right) < \infty 
$.
Then, for any $t > 0$ we have
\begin{equation}
	\label{eq:Representacion_Phi}
	\sum_{n = 1}^{\infty} \left\vert  \phi_{t} \left( \psi^n_0 \right)  \rangle\langle  \phi_{t} \left( \psi^n_0 \right)   \right\vert
	=
	\Phi_t 
	\left(\sum_{n = 1}^{\infty} \left\vert  \psi^n_0  \rangle\langle   \psi^n_0   \right\vert \right) 
	\Phi^{\star}_t ,
\end{equation}
where $\Phi_t$ is given by Lemma \ref{OperadorPhi_t},
and 
$\phi_{t} \left( \psi^n_0 \right) $ denotes 
the strong $C$-solution of  (\ref{eq:Linear_SSE}) with initial datum $ \psi^n_0$.
\end{lemma}

\begin{proof}
For any $\zeta  \in L^2 \left(  \Omega , \mathfrak{F}_{t}, \mathbb{Q}\right) $,
\begin{eqnarray*}
	\fl
	\mathbb{E} \langle \zeta, \sum_{n = 1}^{\infty} \left\vert  \phi_{t} \left( \psi^n_0 \right)  \rangle\langle  \phi_{t} \left( \psi^n_0 \right)   \right\vert  \zeta  \rangle
	& = 
	\sum_{n = 1}^{\infty} \mathbb{E} \left( \left\vert \langle \zeta,    \Phi_t  \psi^n_0   \rangle \right\vert^2  \right)
	\\
	\fl
	& = 
	\sum_{n = 1}^{\infty} \mathbb{E} \left( \left\vert \langle  \Phi^{\star}_t \zeta,     \psi^n_0   \rangle \right\vert^2  \right)  
	= 
	\mathbb{E} \left( \sum_{n = 1}^{\infty} \left\vert \langle  \Phi^{\star}_t \zeta,     \psi^n_0   \rangle \right\vert^2  \right) 	.
\end{eqnarray*}
Hence,
\begin{equation}
	\label{eq:5.7.2}
	\mathbb{E} \langle \zeta, \sum_{n = 1}^{\infty} \left\vert  \phi_{t} \left( \psi^n_0 \right)  \rangle\langle  \phi_{t} \left( \psi^n_0 \right)   \right\vert  \zeta  \rangle
	=
	\mathbb{E} \langle \Phi^{\star}_t \zeta , \sum_{n = 1}^{\infty} \left\vert  \psi^n_0  \rangle\langle   \psi^n_0   \right\vert  \Phi^{\star}_t \zeta  \rangle .
\end{equation}
Using  the Cauchy-Schwarz inequality gives
\[
\left\Vert
\sum_{n = 1}^{\infty} \left\vert  \psi^n_0  \rangle\langle   \psi^n_0   \right\vert  \Phi^{\star}_t \zeta
\right\Vert
=
\left\Vert
\sum_{n = 1}^{\infty} \langle   \psi^n_0 ,  \Phi^{\star}_t \zeta \rangle \psi^n_0  
\right\Vert
\leq
\left\Vert \Phi^{\star}_t \zeta \right\Vert \sum_{n = 1}^{\infty} \left\Vert \psi^n_0 \right\Vert^2  
= 
\left\Vert \Phi^{\star}_t \zeta \right\Vert ,
\]
and so
$
\sum_{n = 1}^{\infty} \left\vert  \psi^n_0  \rangle\langle   \psi^n_0   \right\vert  \Phi^{\star}_t \zeta
\in
L^2 \left(  \Omega , \mathfrak{F}_{t}, \mathbb{Q}\right) 
$.
From (\ref{eq:5.7.2}) we obtain
\[
\mathbb{E} \langle \zeta, \sum_{n = 1}^{\infty} \left\vert  \phi_{t} \left( \psi^n_0 \right)  \rangle\langle  \phi_{t} \left( \psi^n_0 \right)   \right\vert  \zeta  \rangle
=
\mathbb{E} \langle  \zeta , \Phi_t \sum_{n = 1}^{\infty} \left\vert  \psi^n_0  \rangle\langle   \psi^n_0   \right\vert  \Phi^{\star}_t \zeta  \rangle .
\]
Applying the polarization identity we deduce that
\[
\mathbb{E} \langle \zeta_1, \sum_{n = 1}^{\infty} \left\vert  \phi_{t} \left( \psi^n_0 \right)  \rangle\langle  \phi_{t} \left( \psi^n_0 \right)   \right\vert  \zeta_2  \rangle
=
\mathbb{E} \langle  \zeta_1 , \Phi_t \sum_{n = 1}^{\infty} \left\vert  \psi^n_0  \rangle\langle   \psi^n_0   \right\vert  \Phi^{\star}_t \zeta_2  \rangle 
\]
for all $\zeta_1, \zeta_2 \in L^2 \left(  \Omega , \mathfrak{F}_{t}, \mathbb{Q}\right)$.
This implies (\ref{eq:Representacion_Phi}).
\end{proof}

 Lemma \ref{le:Invarianza_lineal} leads to the following lemma.

\begin{lemma}
\label{le:Invarianza_nolineal} 
Let Hypothesis \ref{hyp:CF} hold.
Consider the sequence of random variables
$
\psi^n_0, \chi^n_0 :  \left( \Omega ,\mathfrak{F}_0 \right) \rightarrow \left( \mathfrak{h} , \mathcal{B} \left( \mathfrak{h}  \right) \right)
$ 
that satisfy
$
\sum_{n=1}^{\infty} \mathbb{E} \left( \left\Vert C \psi^n_0 \right\Vert ^{2} +  \left\Vert C \chi^n_0 \right\Vert ^{2}  \right) < \infty 
$,
$
\sum_{n=1}^{\infty} \left\Vert \psi^n_0 \right\Vert ^{2} 
=
1
$ a.s.,
$
\sum_{n=1}^{\infty} \left\Vert \chi^n_0 \right\Vert ^{2} 
= 1 
$ a.s.
and
$
\sum_{n = 1}^{\infty} \left\vert\psi^n_0  \rangle\langle\psi^n_0  \right\vert 
=
\sum_{n = 1}^{\infty} \left\vert\chi^n_0  \rangle\langle\chi^n_0  \right\vert 
$
a.s.
Let $\phi_{t} \left( \psi^n_0 \right) $ and $ \phi_{t} \left( \chi^n_0 \right)$ denote 
the strong $C$-solution of  (\ref{eq:Linear_SSE}) with initial datum $ \psi^n_0$ and $\chi^n_0$,
respectively.
Then,
\begin{equation}
 \label{eq:5.7.1}
 \sum_{n = 1}^{\infty} \left\vert  \hat{ \psi }^n_t   \rangle\langle \hat{ \psi }^n_t  \right\vert
=
\sum_{n = 1}^{\infty} \left\vert  \hat{ \chi }^n_t  \rangle\langle \hat{ \chi }^n_t  \right\vert 
\qquad a.s.,
\end{equation}
where
\[
\hat{\psi}^n_t = 
\cases{
 \phi_{t} \left( \psi^n_0 \right) / \sqrt{  \sum_{n=1}^{\infty} \left\Vert \phi_{t} \left( \psi^n_0 \right)   \right\Vert ^{2}  }
&
 if  $  \sum_{n=1}^{\infty} \left\Vert \phi_{t} \left( \psi^n_0 \right)   \right\Vert ^{2} \neq 0 $
\\
0
&
 if  $  \sum_{n=1}^{\infty} \left\Vert \phi_{t} \left( \psi^n_0 \right)   \right\Vert ^{2} = 0 $
} 
,
\]
and
\[
\hat{\chi}^n_t = 
\cases{
 \phi_{t} \left( \chi^n_0 \right) / \sqrt{  \sum_{n=1}^{\infty} \left\Vert \phi_{t} \left( \chi^n_0 \right)   \right\Vert ^{2}  }
&
if  $ \sum_{n=1}^{\infty} \left\Vert \phi_{t} \left( \chi^n_0 \right)   \right\Vert ^{2} \neq 0 $
\\
0
&
 if  $  \sum_{n=1}^{\infty} \left\Vert \phi_{t} \left( \chi^n_0 \right)   \right\Vert ^{2} = 0 $
} .
\]
\end{lemma}

\begin{proof}
 Applying Lemma \ref{le:Invarianza_lineal} we obtain 
\begin{eqnarray*}
 \sum_{n = 1}^{\infty} \left\Vert  \phi_{t} \left( \psi^n_0 \right)  \right\Vert^2
& =
\Tr
\left(
\sum_{n = 1}^{\infty} \left\vert  \phi_{t} \left( \psi^n_0 \right)  \rangle\langle  \phi_{t} \left( \psi^n_0 \right)   \right\vert
\right)
\\
&
=
\Tr
\left(
\Phi_t 
\left(\sum_{n = 1}^{\infty} \left\vert  \psi^n_0  \rangle\langle   \psi^n_0   \right\vert \right) 
\Phi^{\star}_t 
\right) 
=
\sum_{n = 1}^{\infty} \left\Vert  \phi_{t} \left( \chi^n_0 \right)  \right\Vert^2 
\end{eqnarray*}
 almost surely. 
Hence,
\begin{eqnarray*}
\fl
 &
 \sum_{n = 1}^{\infty} \left\vert \frac{ \phi_{t} \left( \psi^n_0 \right) }{ \sqrt{\sum_{n = 1}^{\infty} \left\Vert  \phi_{t} \left( \psi^n_0 \right)  \right\Vert^2} } \rangle
\langle \frac{ \phi_{t} \left( \psi^n_0 \right)  
 }{ \sqrt{\sum_{n = 1}^{\infty} \left\Vert  \phi_{t} \left( \psi^n_0 \right)  \right\Vert^2} } 
\right\vert 
	=
	\frac{
		\Phi_t 
		\left(\sum_{n = 1}^{\infty} \left\vert  \psi^n_0  \rangle\langle   \psi^n_0   \right\vert \right) 
		\Phi^{\star}_t
	}{
	\sum_{n = 1}^{\infty} \left\Vert  \phi_{t} \left( \psi^n_0 \right)  \right\Vert^2 }
\\
\fl \qquad
& =
\sum_{n = 1}^{\infty} \left\vert \frac{ \phi_{t} \left( \chi^n_0 \right) }{
\sqrt{\sum_{n = 1}^{\infty} \left\Vert  \chi_{t} \left( \psi^n_0 \right)  \right\Vert^2} } 
 \rangle\langle  \frac{ \phi_{t} \left( \chi^n_0 \right) }{
\sqrt{\sum_{n = 1}^{\infty} \left\Vert  \phi_{t} \left( \chi^n_0 \right)  \right\Vert^2} } 
  \right\vert 
  \qquad a.s.
\end{eqnarray*}
provided that 
$ \sum_{n = 1}^{\infty} \left\Vert  \phi_{t} \left( \psi^n_0 \right)  \right\Vert^2 \neq 0$ 
and 
$ \sum_{n = 1}^{\infty} \left\Vert  \phi_{t} \left( \chi^n_0 \right)  \right\Vert^2 \neq 0$.
This implies (\ref{eq:5.7.1}).
\end{proof}

We now prove that the law of the density operator $\rho_{t}$ defined by (\ref{eq:Def_rho})
does not depend on  the decomposition of the initial density operator.
A difficulty we face is that the decompositions 
$ \sum_{n = 1}^{\infty}  \left\vert   \psi^n_0   \rangle\langle   \psi^n_0  \right\vert$
and  $ \sum_{n = 1}^{\infty}  \left\vert   \chi^n_0   \rangle\langle   \chi^n_0  \right\vert$
can live in different probability spaces.

\begin{proof}[Proof of Theorem \ref{th:Invarianza_ci}.] 
First, for all $\varrho \in \mathfrak{L}_{1} \left( \mathfrak{h}\right)$  we set 
\[
\Gamma \left( \varrho \right) = 
\cases{
 \varrho & if  $ \varrho \in \mathfrak{L}_{1}^{+} \left( \mathfrak{h}\right) $
 \\
 0 & if  $ \varrho \notin \mathfrak{L}_{1}^{+} \left( \mathfrak{h}\right) $
}, 
\]
where  
$\mathfrak{L}_{1}^{+} \left( \mathfrak{h}\right)$ stands for the set of all  non-negative  trace-class operators on $\mathfrak{h}$.
Since $  \mathfrak{L}_{1}^{+} \left( \mathfrak{h}\right) \in  \mathcal{B} \left( \mathfrak{L}_{1} \left( \mathfrak{h}\right) \right) $,
$
\Gamma :  \left( \mathfrak{L}_{1} \left( \mathfrak{h}\right) , \mathcal{B} \left( \mathfrak{L}_{1} \left( \mathfrak{h}\right) \right) \right)
 \rightarrow 
 \left( \mathfrak{L}_{1} \left( \mathfrak{h}\right) , \mathcal{B} \left( \mathfrak{L}_{1} \left( \mathfrak{h}\right) \right) \right) 
$
is a non-negative measurable function,
and so there exists a sequence $\left( \zeta_n \right)_{n \in \mathbb{N}}$ of  $\mathfrak{h}$-valued random variables such that
\begin{equation}
\label{eq:rep_Gamma}
\Gamma = \sum_{n \in \mathbb{N}} \left\vert \zeta_n  \rangle\langle  \zeta_n  \right\vert
\end{equation}
and 
$ \langle \zeta_n, \zeta_m \rangle = 0 $ whenever $n\neq m$ (see, e.g., \cite{DaPrato}).
We set $\zeta = \left( \zeta_n \right)_{n \in \mathbb{N}}$. 
Since $ \zeta_n:  \left( \mathfrak{L}_{1} \left( \mathfrak{h}\right) , \mathcal{B} \left( \mathfrak{L}_{1} \left( \mathfrak{h}\right) \right) \right)
 \rightarrow 
\left(  \mathfrak{h}   , \mathcal{B} \left(  \mathfrak{h}\right)  \right)  $ is measurable for all $n \in \mathbb{N}$,
 $\zeta : \left( \mathfrak{L}_{1} \left( \mathfrak{h}\right) , \mathcal{B} \left( \mathfrak{L}_{1} \left( \mathfrak{h}\right) \right) \right)
 \rightarrow 
\left(  \mathfrak{h}^{\mathbb{N}}   , \mathcal{B} \left(  \mathfrak{h}^{\mathbb{N}} \right)  \right)  
 $
is measurable, because
 $ \mathcal{B} \left( \left( \mathfrak{h}^{\mathbb{N}}, d_{\mathfrak{h}^{\mathbb{N}}}  \right)\right)
 =
 \mathcal{B} \left(  \mathfrak{h} \right)^{\mathbb{N}}
 $
 (see, e.g., Remark \ref{Nota:Distancia}).

Second, 
using classical arguments we will specify a stochastic basis over 
$\Omega =  \mathfrak{h}^{\mathbb{N}} \times  C \left( \left[ 0,T\right] ,  \mathbb{R}^{\mathbb{N}} \right)$.
Consider the Wiener measure $\mathbb{Q}_{B}$ on the canonical space 
$\left( C \left( \left[ 0,T\right] ,  \mathbb{R}^{\mathbb{N}} \right) ,  \mathcal{B}  \left( C \left( \left[ 0,T\right] ,  \mathbb{R}^{\mathbb{N}} \right) \right) \right)$,
that is,  with respect to $\mathbb{Q}_{B}$ 
the coordinate maps $  \left( w_k \right)_{k \in \mathbb{N}} \mapsto w_n $ make up a sequence of independent real Brownian motions.
We choose 
 $\mathfrak{F} = \mathcal{B}  \left(  \mathfrak{h}^{\mathbb{N}} \right) \otimes
\mathcal{B}  \left( C \left( \left[ 0,T\right] ,  \mathbb{R}^{\mathbb{N}} \right) \right)$
and 
$ \mathbb{Q} = \mathbb{P}_{ \left( \psi^n_0  \right)_{n \in \mathbb{N}} } \otimes \mathbb{Q}_{B} $,
where 
$\mathbb{P}_{ \left( \psi^n_0  \right)_{n \in \mathbb{N}} }$ is the distribution of $\left( \psi^n_0  \right)_{n \in \mathbb{N}} $
in $ \mathfrak{h}^{\mathbb{N}}$.
Let $\mathfrak{F}^a_t$ be the sub-$\sigma$-algebra of 
$\mathcal{B}  \left( C \left( \left[ 0,T\right] ,  \mathbb{R}^{\mathbb{N}} \right) \right)$
generated by the $ \mathbb{R}^{\mathbb{N}} $-valued maps $f \mapsto f  \left( s\right)$ with $s \in  \left[ 0, t \right]$.
Is worth pointing out that $\mathfrak{F}^a_T = \mathcal{B}  \left( C \left( \left[ 0,T\right] ,  \mathbb{R}^{\mathbb{N}} \right) \right)$ (see, e.g., \cite{Ondrejat2004}).
We take
\[
\mathfrak{F}_t 
=
\cap_{\epsilon > 0} \sigma\left( 
\mathcal{B}  \left(  \mathfrak{h}^{\mathbb{N}} \right) \otimes \mathfrak{F}^a_{\left( t + \epsilon \right) \wedge T}  
\cup
\left\{ A \in \mathfrak{F}: \mathbb{Q} \left( A \right) = 0
\right\}
 \right).
\]
Thus, the filtration $\left( \mathfrak{F}_t \right)_{t \in \left[ 0, T \right]}$ satisfies the so-called usual conditions.
We define the functions $B^k : \Omega \rightarrow C \left( \left[ 0,T\right] ,  \mathbb{R} \right) $ 
by 
$
B^n \left( \left( \left( x_k  \right)_{k \in \mathbb{N}} ,  \left( w_k  \right)_{k \in \mathbb{N}} \right) \right)
=
w_n  
$.
Then $B^{1}, B^2, \ldots $ are real valued independent Brownian motions on
$\left( \Omega,\mathfrak{F},\left( \mathfrak{F}_{t}\right) _{t \in \mathbb{I}},\mathbb{Q}\right) $.

Third,
we define $\widetilde{\psi}_0  : \Omega \rightarrow \mathfrak{h}^{\mathbb{N}} $ by 
\[
\widetilde{\psi}_0 \left( \left( \left( x_n  \right)_{n \in \mathbb{N}} , \left( w_n  \right)_{n \in \mathbb{N}} \right) \right)
=
\left( x_n  \right)_{n \in \mathbb{N}} \, \mathbf{1}_{ \mathcal{D}\left( \widetilde{ C } \right)  \cap
\left\{ y \in \mathfrak{h}^{\mathbb{N}} : \sum_{n=1}^{\infty} \left\Vert y^n \right\Vert ^{2} = 1 \right\} }
\left( \left( x_n  \right)_{n \in \mathbb{N}} \right) ,
\]
where $\mathbf{1}$ denotes the indicator function.
Since 
\[ 
\mathbb{P}_{ \left( \psi^n_0  \right)_{n \in \mathbb{N}} }  \left(
 \mathcal{D}\left( \widetilde{ C } \right)  \cap
\left\{ y \in \mathfrak{h}^{\mathbb{N}} : \sum_{n=1}^{\infty} \left\Vert y^n \right\Vert ^{2} = 1 \right\}\right) = 1,
\]
$ \widetilde{\psi}_0 $ is distributed according to $\mathbb{P}_{ \left( \psi^n_0  \right)_{n \in \mathbb{N}} }$.
From Theorem \ref{th:regular-sol-LSIWE} we obtain that
there exists a unique $C$-solution $\phi_{t} \left( \widetilde{\psi}^n_0 \right) $  of  (\ref{eq:Linear_SSE})
with initial datum $ \widetilde{\psi}^n_0$  for any $n \in \mathbb{N}$,
where $\left( \widetilde{\psi}^n_0  \right)_{n \in \mathbb{N} }= \widetilde{\psi}_0$.

Fourth, we set 
\[
\Psi_n = \zeta_n \left( \sum_{n = 1}^{\infty}  \left\vert    \widetilde{\psi}^n_0   \rangle\langle   \widetilde{\psi}^n_0  \right\vert \right) .
\]
Thus,
$\Psi = \left(   \Psi_n \right)_{n \in \mathbb{N} }  
=  \zeta \left( \sum_{n = 1}^{\infty}  \left\vert    \widetilde{\psi}^n_0   \rangle\langle   \widetilde{\psi}^n_0  \right\vert \right)$
is an    $\mathfrak{h}^{\mathbb{N}}$-valued $\mathfrak{F}_{0}$-random variable,
which defined on $\Omega$.
As $ \sum_{n = 1}^{\infty}  \left\vert    \widetilde{\psi}^n_0   \rangle\langle   \widetilde{\psi}^n_0  \right\vert
 \in \mathfrak{L}_{1}^{+} \left( \mathfrak{h}\right)$,
 using (\ref{eq:rep_Gamma}) gives
\begin{equation}
\label{eq:Igualdad_psi}
\sum_{n = 1}^{\infty}  \left\vert    \widetilde{\psi}^n_0   \rangle\langle   \widetilde{\psi}^n_0  \right\vert 
= 
\sum_{n = 1}^{\infty} \left\vert   \Psi_n   \rangle\langle  \Psi_n  \right\vert .
\end{equation}

Using (\ref{eq:Igualdad_psi}) gives 
\[
\sum_{n = 1}^{\infty} \left\Vert \Psi_n   \right\Vert ^2 
= \Tr \left(  \sum_{n = 1}^{\infty}  \left\vert    \widetilde{\psi}^n_0   \rangle\langle   \widetilde{\psi}^n_0  \right\vert   \right) 
= \sum_{n = 1}^{\infty} \left\Vert \widetilde{\psi}^n_0   \right\Vert ^2
= 1
\quad  \mathbb{Q}\mathrm{-}a.s.
\]
Since 
$   \sum_{n = 1}^{\infty}  \left\vert    \widetilde{\psi}^n_0   \rangle\langle   \widetilde{\psi}^n_0  \right\vert
\in \mathfrak{L}_{1,C}^+ \left( \mathfrak{h}\right) $,
the range of $  \sum_{n = 1}^{\infty}  \left\vert    \widetilde{\psi}^n_0   \rangle\langle   \widetilde{\psi}^n_0  \right\vert  $ 
is contained in $ \mathcal{D}\left( C \right)$
(see, e.g., Theorems 3.1 and 3.2 of \cite{MoraAP}).
This implies $ \Psi_n \in  \mathcal{D}\left( C \right) $,
because 
$\left( \sum_{k = 1}^{\infty}  \left\vert   \psi^k_0   \rangle\langle  \psi^k_0  \right\vert \right) \Psi_n
= \left\Vert \Psi_n \right\Vert^2 \Psi_n $.
Applying the Parseval identity yields 
\begin{eqnarray*}
\fl
 \sum_{k=1}^{\infty} \langle e_k , C  \left( \sum_{n = 1}^{\infty} \left\vert   \Psi_n   \rangle\langle  \Psi_n  \right\vert \right)  C  e_k \rangle
& =
 \sum_{k=1}^{\infty} \langle C e_k ,   \left( \sum_{n = 1}^{\infty} \left\vert   \Psi_n   \rangle\langle  \Psi_n  \right\vert \right)  C e_k \rangle
 \\
 & =
\sum_{k=1}^{\infty} \sum_{n=1}^{\infty} \langle C e_k ,  \Psi_n  \rangle \langle \Psi_n ,  C  e_k \rangle
\\
& =
\sum_{n=1}^{\infty} \sum_{k=1}^{\infty}  \langle  e_k , C \Psi_n  \rangle \langle C \Psi_n ,    e_k \rangle
=
\sum_{n=1}^{\infty} \left\Vert C \Psi_n \right\Vert^2 ,
\end{eqnarray*}
where 
$\left( e_n \right)_{n \in \mathbb{N}}$ is  an orthonormal basis of $ \mathfrak{h} $ formed by elements of $ \mathcal{D}\left( C \right)$.
Likewise,
\[
 \sum_{k=1}^{\infty} \langle e_k , C  \left( \sum_{n = 1}^{\infty}  \left\vert    \widetilde{\psi}^n_0   \rangle\langle   \widetilde{\psi}^n_0  \right\vert \right)  C  e_k 
=
\sum_{n=1}^{\infty} \left\Vert C  \widetilde{\psi}^n_0 \right\Vert^2 .
\]
Hence,
$
\sum_{n=1}^{\infty} \mathbb{E} \left\Vert C \Psi_n \right\Vert^2 
= \sum_{n=1}^{\infty} \mathbb{E}  \left\Vert C \widetilde{\psi}^n_0  \right\Vert^2
< + \infty 
$.
According to Theorem \ref{th:regular-sol-LSIWE} we have that
there exists a unique $C$-solution $\phi_{t} \left(  \Psi_n \right) $  of  (\ref{eq:Linear_SSE}),
defined on $\Omega$,
with initial datum $ \Psi_n$  for any $n \in \mathbb{N}$.

Fifth, 
we compare the probability distributions of the density operators 
$ \sum_{n = 1}^{\infty} \left\vert  \psi^n_t   \rangle\langle \psi^n_t \right\vert$
and $\sum_{n = 1}^{\infty} \left\vert  \hat{ \Psi }^n_t  \rangle\langle \hat{ \Psi }^n_t  \right\vert ) $,
where
\begin{equation}
 \label{eq:Def_Psi}
 \hat{\Psi}^n_t = 
\cases{
 \phi_{t} \left( \Psi_n \right) / \sqrt{  \sum_{n=1}^{\infty} \left\Vert \phi_{t} \left( \Psi_n \right)   \right\Vert ^{2}  }
&
if  $  \sum_{n=1}^{\infty} \left\Vert \phi_{t} \left(  \Psi_n \right)   \right\Vert ^{2} \neq 0 $
\\
0
&
if  $  \sum_{n=1}^{\infty} \left\Vert \phi_{t} \left( \Psi_n \right)   \right\Vert ^{2} = 0 $
} .
\end{equation}
Combining (\ref{eq:Igualdad_psi}) with Lemma \ref{le:Invarianza_nolineal} we get
\begin{equation}
\label{eq:Igualdad_psi_2}
\sum_{n = 1}^{\infty} \left\vert  \hat{ \Psi }^n_t  \rangle\langle \hat{ \Psi }^n_t  \right\vert 
=
 \sum_{n = 1}^{\infty} \left\vert  \hat{ \psi }^n_t   \rangle\langle \hat{ \psi }^n_t  \right\vert
\qquad a.s.,
\end{equation}
where
\[
\hat{\psi}^n_t = 
\cases{
 \phi_{t} \left( \widetilde{\psi}^n_0 \right) / \sqrt{  \sum_{n=1}^{\infty} \left\Vert \phi_{t} \left( \widetilde{\psi}^n_0 \right)   \right\Vert ^{2}  }
&
 if  $  \sum_{n=1}^{\infty} \left\Vert \phi_{t} \left( \widetilde{\psi}^n_0 \right)   \right\Vert ^{2} \neq 0 $
\\
0
&
if  $ \sum_{n=1}^{\infty} \left\Vert \phi_{t} \left( \widetilde{\psi}^n_0 \right)   \right\Vert ^{2} = 0 $
} .
\]
From Theorem \ref{th:Model} we obtain that
$
\left( \left(  \sum_{n=1}^{\infty} \left\Vert \phi_{T} \left( \widetilde{\psi}^n_0 \right)   \right\Vert ^{2} \right) \cdot \mathbb{Q}, \left( \hat{\psi}^n_t  \right)_{t \in \left[ 0,T\right]}^{n \in \mathbb{N}},\left( \hat{W}_{t}^{\ell}\right) _{t \in\left[ 0,T\right]}^{\ell\in\mathbb{N}} \right)
$
is a $C$-solution of the stochastic system (\ref{st:SIWE}),
where
\begin{equation}
\label{eq:Def_W}
\hat{W}_{t}^{\ell} = B_{t}^{\ell}
- \int_{0}^{t}  2 \Re \left(  \sum_{n=1}^{\infty} \langle  \hat{\psi}^n_s , L_{\ell} \left( s \right) \hat{\psi}^n_s  \rangle  \right) ds .
\end{equation}
Applying Lemma \ref{le:A_rho_B} gives
\[
\hat{W}_{t}^{\ell} 
= 
B_{t}^{\ell}
- \int_{0}^{t}  2 \Re \left(  \Tr  \left( \sum_{n = 1}^{\infty}  \left( \left\vert   \hat{\psi}^n_s   \rangle\langle  \hat{\psi}^n_s  \right\vert \right) L_{\ell} \left( s \right) \right) \right)  ds ,
\]
and so (\ref{eq:Igualdad_psi_2}) leads to
\begin{equation}
 \label{eq:Def_Wh}
 \hat{W}_{t}^{\ell} 
=
B_{t}^{\ell}
- \int_{0}^{t}  2 \Re \left(  \Tr  \left( \sum_{n = 1}^{\infty} \left( \left\vert  \hat{\Psi}^n_s   \rangle\langle   \hat{\Psi}^n_s  \right\vert \right) L_{\ell} \left( s \right) \right) \right)  ds .
\end{equation}

Since $\left( \widetilde{\psi}^n_0  \right)_{n \in \mathbb{N} }$ and $\left( \psi^n_0  \right)_{n \in \mathbb{N} }$
have the same law,
applying Theorem \ref{th:EyU-SIWE} we obtain that the finite-dimensional distributions (under the probability $\mathbb{P}$) of 
$
\left(\left( \psi^n_t  \right)^{n \in \mathbb{N}},\left( W_{t}^{\ell}\right)^{\ell\in\mathbb{N}} \right)_{t \in\left[ 0,T\right]}
$
are equal to the finite-dimensional distributions (under the probability $\left( \sum_{n=1}^{\infty} \left\Vert \phi_{T} \left(\widetilde{\psi}^n_0 \right)   \right\Vert ^{2} \right) \cdot \mathbb{Q}$) of 
$
 \left( \left( \hat{\psi}^n_t  \right)^{n \in \mathbb{N}},\left( \hat{W}_{t}^{\ell}\right)^{\ell\in\mathbb{N}} \right)_{t \in\left[ 0,T\right]}
$.
Therefore,
$
\left( \sum_{n = 1}^{\infty} \left\vert  \psi^n_t   \rangle\langle \psi^n_t  \right\vert ,\left( W_{t}^{\ell}\right)^{\ell\in\mathbb{N}} \right)_{t \in\left[ 0,T\right]}
$
and
$\left(  \sum_{n = 1}^{\infty} \left\vert   \hat{\psi}^n_t   \rangle\langle  \hat{\psi}^n_t  \right\vert , \left(  \hat{ W }_{t}^{\ell}\right)^{\ell\in\mathbb{N}} \right)_{t \in\left[ 0,T\right]}$
have the same finite-dimensional distributions.

Combining (\ref{eq:Igualdad_psi}) with Lemma \ref{le:Invarianza_lineal} yields
\[
\sum_{n = 1}^{\infty} \left\vert \phi_{T} \left(\widetilde{\psi}^n_0 \right)  \rangle\langle \phi_{T} \left(\widetilde{\psi}^n_0 \right)  \right\vert
=
\sum_{n = 1}^{\infty} \left\vert  \phi_{T} \left(\Psi_n \right)  \rangle\langle \phi_{T} \left(\Psi_n \right)  \right\vert 
\qquad \mathbb{Q}\mathrm{-}a.s.,
\] 
and so
\[
\sum_{n=1}^{\infty} \left\Vert \phi_{T} \left(\widetilde{\psi}^n_0 \right)   \right\Vert ^{2} 
=
\sum_{n=1}^{\infty} \left\Vert  \phi_{T} \left(\Psi_n \right)   \right\Vert ^{2} 
\qquad 
\mathbb{Q}\mathrm{-}a.s.
\]
Therefore,
$
\left( \sum_{n=1}^{\infty} \left\Vert \phi_{T} \left(\widetilde{\psi}^n_0 \right)   \right\Vert ^{2} \right) \cdot \mathbb{Q}
=
\left( \sum_{n=1}^{\infty} \left\Vert \phi_{T} \left( \Psi_n \right)   \right\Vert ^{2} \right) \cdot \mathbb{Q}
$,
and so using (\ref{eq:Igualdad_psi_2}) we get that 
the finite-dimensional distributions (with respect to $\mathbb{P}$) of
$
\left( \sum_{n = 1}^{\infty} \left\vert  \psi^n_t   \rangle\langle \psi^n_t  \right\vert,\left( W_{t}^{\ell}\right)^{\ell\in\mathbb{N}} \right)_{t \in\left[ 0,T\right]}
$
are equal to the  finite-dimensional distributions (with respect to  $\left( \sum_{n=1}^{\infty} \left\Vert \phi_{T} \left( \Psi_n \right)   \right\Vert ^{2} \right) \cdot \mathbb{Q}$) of
$
\left(  \sum_{n = 1}^{\infty} \left\vert  \hat{ \Psi }^n_t  \rangle\langle \hat{ \Psi }^n_t  \right\vert ,\left(  \hat{ W }_{t}^{\ell}\right)^{\ell\in\mathbb{N}} \right)_{t \in\left[ 0,T\right]}
$.

Sixth, we verify that $ \left( \left(  \hat{\Psi}^n_t \right)_{n \in \mathbb{N} } \right)_{t \in \left[ 0,T\right]} $ satisfies  (\ref{eq:SIWE}).
From (\ref{eq:Def_Wh}) we get
\[
\hat{W}_{t}^{\ell} = B_{t}^{\ell}
- \int_{0}^{t}  2 \Re \left(  \sum_{n=1}^{\infty} \langle  \hat{\Psi}^n_s , L_{\ell} \left( s \right) \hat{\Psi}^n_s  \rangle  \right) ds ,
\]
and so
using Theorem \ref{th:Model} we get that
$
\left( \left(  \sum_{n=1}^{\infty} \left\Vert \phi_{T} \left( \Psi_n \right)   \right\Vert ^{2} \right) \cdot \mathbb{Q}, \left( \hat{\Psi}^n_t  \right)_{t \in \left[ 0,T\right]}^{n \in \mathbb{N}},\left( \hat{W}_{t}^{\ell}\right) _{t \in\left[ 0,T\right]}^{\ell\in\mathbb{N}} \right)
$
is a $C$-solution of (\ref{eq:SIWE}).
We have that $ \hat{\Psi}^n_0 = \Psi_n $
due to $\sum_{n \in \mathbb{N}} \left\Vert \Psi_n   \right\Vert ^2 = 1$ $\mathbb{Q}$-a.s.
Using the fact that $ \sum_{n=1}^{\infty} \left\Vert \phi_{T} \left( \Psi_n \right)   \right\Vert ^{2} $ is a $\mathbb{Q}$-martingale,
we deduce that the distribution of the initial condition $\left( \hat{\Psi}^n_0  \right)_{n \in \mathbb{N}} $
(with respect to $\left(  \sum_{n=1}^{\infty} \left\Vert \phi_{T} \left( \Psi_n \right)   \right\Vert ^{2} \right) \cdot \mathbb{Q}$)
is equal to $\mathbb{Q} \circ \left( \Psi_n  \right)_{n \in \mathbb{N}}^{-1}$.
From the definition of $\Psi$ we have that 
for any $A \in \mathcal{B} \left(  \mathfrak{h}^{\mathbb{N}} \right)$,
\[
\fl
\mathbb{Q} \circ \left( \Psi_n  \right)_{n \in \mathbb{N}}^{-1} \left( A \right)
=
\mathbb{Q}  \left(  
\sum_{n = 1}^{\infty}  \left\vert    \widetilde{\psi}^n_0   \rangle\langle   \widetilde{\psi}^n_0  \right\vert \in \zeta^{-1} \left( A \right)
\right)
=
\mathbb{P} 
\left(  \sum_{n = 1}^{\infty}  \left\vert   \psi^n_0   \rangle\langle   \psi^n_0  \right\vert \in \zeta^{-1} \left( A \right)
\right) .
\]

Seventh, 
in the previous  stages we replace 
$\left( \psi^n_0  \right)_{n \in \mathbb{N} }$ and $ \mathbb{P}$ by  
$\left( \chi^n_0  \right)_{n \in \mathbb{N} }$ and $ \widetilde{\mathbb{P}}$, respectively.
Thus, $\mathbb{Q}$ becomes 
$ \widetilde{\mathbb{Q}} = \widetilde{ \mathbb{P} }_{ \left( \chi^n_0  \right)_{n \in \mathbb{N}} } \otimes \mathbb{Q}_{B} $.
Moreover, we have that the finite-dimensional distributions (under  $\widetilde{\mathbb{P}}$) of
$
\left( \sum_{n = 1}^{\infty} \left\vert  \chi^n_t   \rangle\langle \chi^n_t  \right\vert  ,\left( \widetilde{W}_{t}^{\ell}\right)^{\ell\in\mathbb{N}} \right)_{t \in\left[ 0,T\right]}
$
are equal to the  finite-dimensional distributions (under  $\left( \sum_{n=1}^{\infty} \left\Vert \phi_{T} \left( \Psi_n \right)   \right\Vert ^{2} \right) \cdot  \widetilde{\mathbb{Q}} $) of
$
\left(  \sum_{n = 1}^{\infty} \left\vert  \hat{ \Psi }^n_t  \rangle\langle \hat{ \Psi }^n_t  \right\vert  ,
\left(  \hat{ W }_{t}^{\ell}\right)^{\ell\in\mathbb{N}} \right)_{t \in\left[ 0,T\right]}
$,
where, by abuse of notation, 
$  \hat{ \Psi }^n_t$ is defined by (\ref{eq:Def_Psi}) with $\phi_{t} \left( \Psi_n \right) $ 
$C$-solution of  (\ref{eq:Linear_SSE}) defined on $\left( \Omega ,\widetilde{\mathbb{Q}} \right)$.
We also have that 
$
\left( \left(  \sum_{n=1}^{\infty} \left\Vert \phi_{T} \left( \Psi_n \right)   \right\Vert ^{2} \right) \cdot \widetilde{\mathbb{Q}}, \left( \hat{\Psi}^n_t  \right)_{t \in \left[ 0,T\right]}^{n \in \mathbb{N}},\left( \hat{W}_{t}^{\ell}\right) _{t \in\left[ 0,T\right]}^{\ell\in\mathbb{N}} \right)
$
is a $C$-solution of the stochastic system (\ref{st:SIWE}) with the initial distribution given by 
$ A \mapsto 
\widetilde{\mathbb{P}}
\left(  \sum_{n = 1}^{\infty}  \left\vert   \chi^n_0   \rangle\langle   \chi^n_0  \right\vert \in \zeta^{-1} \left( A \right)
\right) 
$
for any $A \in \mathcal{B} \left(  \mathfrak{h}^{\mathbb{N}} \right)$.

Finally,
$
\mathbb{P} 
\left(  \sum_{n = 1}^{\infty}  \left\vert   \psi^n_0   \rangle\langle   \psi^n_0  \right\vert \in \zeta^{-1} \left( A \right)
\right) 
=
\widetilde{\mathbb{P}}
\left(  \sum_{n = 1}^{\infty}  \left\vert   \chi^n_0   \rangle\langle   \chi^n_0  \right\vert \in \zeta^{-1} \left( A \right)
\right) 
$
for all $A \in \mathcal{B} \left(  \mathfrak{h}^{\mathbb{N}} \right)$,
because  
$
\sum_{n = 1}^{\infty} \left\vert\psi^n_0  \rangle\langle\psi^n_0  \right\vert 
$
and
$
\sum_{n = 1}^{\infty} \left\vert\chi^n_0  \rangle\langle\chi^n_0  \right\vert 
$
have the same law in $\mathfrak{L}_{1}\left( \mathfrak{h}\right)$.
Hence,
$
\left( \left(  \sum_{n=1}^{\infty} \left\Vert \phi_{T} \left( \Psi_n \right)   \right\Vert ^{2} \right) \cdot \mathbb{Q}, \left( \hat{\Psi}^n_t  \right)_{t \in \left[ 0,T\right]}^{n \in \mathbb{N}},\left( \hat{W}_{t}^{\ell}\right) _{t \in\left[ 0,T\right]}^{\ell\in\mathbb{N}} \right)
$
and
$
\left( \left(  \sum_{n=1}^{\infty} \left\Vert \phi_{T} \left( \Psi_n \right)   \right\Vert ^{2} \right) \cdot \widetilde{\mathbb{Q}}, \left( \hat{\Psi}^n_t  \right)_{t \in \left[ 0,T\right]}^{n \in \mathbb{N}},\left( \hat{W}_{t}^{\ell}\right) _{t \in\left[ 0,T\right]}^{\ell\in\mathbb{N}} \right)
$
are $C$-solutions of the stochastic system (\ref{st:SIWE})  with the same initial distribution;
as we pointed out we use the same expression $  \hat{ \Psi }^n_t$  for two different stochastic processes.
Applying Theorem \ref{th:EyU-SIWE} we deduce that the  finite-dimensional distributions of 
$
\left(  \sum_{n = 1}^{\infty} \left\vert  \hat{ \Psi }^n_t  \rangle\langle \hat{ \Psi }^n_t  \right\vert , \left(  \hat{ W }_{t}^{\ell}\right) ^{\ell\in\mathbb{N}} \right)_{t \in\left[ 0,T\right]}
$
under $\left( \sum_{n=1}^{\infty} \left\Vert \phi_{T} \left( \Psi_n \right)   \right\Vert ^{2} \right) \cdot  \mathbb{Q} $
coincides with the finite-dimensional distributions of 
$
\left( \sum_{n = 1}^{\infty} \left\vert  \hat{ \Psi }^n_t  \rangle\langle \hat{ \Psi }^n_t  \right\vert ,\left(  \hat{ W }_{t}^{\ell}\right)^{\ell\in\mathbb{N}} \right)_{t \in\left[ 0,T\right]}
$
under  $\left( \sum_{n=1}^{\infty} \left\Vert \phi_{T} \left( \Psi_n \right)   \right\Vert ^{2} \right) \cdot  \widetilde{\mathbb{Q}} $.
Therefore,
from the fifth and seventh stages we obtain that 
$
\left( \sum_{n = 1}^{\infty} \left\vert  \psi^n_t   \rangle\langle \psi^n_t  \right\vert ,
\left( W_{t}^{\ell}\right)^{\ell\in\mathbb{N}} \right)_{t \in\left[ 0,T\right]}
$
and
$
\left( \sum_{n = 1}^{\infty} \left\vert  \chi^n_t   \rangle\langle \chi^n_t  \right\vert , \left( \widetilde{W}_{t}^{\ell}\right)^{\ell\in\mathbb{N}} \right)_{t \in\left[ 0,T\right]}
$
have the same finite-dimensional distributions.
Applying Theorem \ref{th:EyU-SIWE} yields the assertion of the theorem.
\end{proof}

\subsection{ Proof of Theorem \ref{th:Existence_SQME}}
\label{sec:Proof:Existence_SQME}

\begin{proof}
The assertion (a) follows from Theorem \ref{th:Def_Model1}.
In order to obtain the assertion (b) we define  $\theta $  to be the probability distribution of 
$ \left( \psi^n_0  \right)_{n \in \mathbb{N}}$ on $\mathfrak{h}^{\mathbb{N}} $.
As in Lemma \ref{lem:Existence}, 
$\widetilde{\theta}$ is the restriction of $\theta $ to $\mathcal{B} \left( L^2 \left( \mathbb{N} ;  \mathfrak{h} \right) \right)$.
According to  Theorem \ref{th:EyU-NLSSE-e},  
(\ref{eq:NLSSE_e})  has  a $\widetilde{C}$-solution $\left( \mathbb{P}, \left( \hat{x}_{t}\right)_{t \geq 0}, \left( B_{t}^{\ell}\right) _{t \geq 0 }^{\ell\in\mathbb{N}} \right)$ with initial distribution   $\widetilde{\theta}$.
Set
\[
\varrho_{t} = \sum_{n = 1}^{\infty} \left\vert \hat{x}_t \left( n \right)  \rangle\langle \hat{x}_t \left( n \right)  \right\vert 
\qquad \qquad \forall t \geq 0 .
\]
Since $ \hat{x}_t  \in  L^2 \left( \mathbb{N} ;  \mathfrak{h} \right) $,
using Lemma \ref{le:L1} yields  $\varrho_{t} \in \mathfrak{L}_{1}\left( \mathfrak{h}\right)$.
Since  $\hat{x}_{t}\in \mathcal{D}\left( \widetilde{C} \right)$  $\mathbb{P}$-$a.s.$ and $ \left\Vert \hat{x}_{t}\right\Vert =1$ $\mathbb{P}$-$a.s.$,
$\varrho_{t} \in  \mathfrak{L}_{1,C}^+ \left( \mathfrak{h}\right) $  $\mathbb{P}$-a.s.  for any $t \geq 0 $.

As in the proof of Theorem \ref{th:Model} we define 
\[
 \overline{ u } = \sum_{n=1}^{\infty}  \overline{ \langle e_n , u \rangle } e_n
 \qquad \qquad
 \qquad \mathrm{for \ all \;} u \in \mathfrak{h} ,
\]
where 
$\left( e_n \right)_{n \in \mathbb{N}}$ is a given orthonormal basis of $\mathfrak{h}$.
Then,
\begin{eqnarray*}
\fl
& \overline{ \hat{x}_t }
 =
\overline{ \hat{x}_0 }
+ \sum_{ \ell =1}^{ \infty} \int_0^t \overline{ \left(
\widetilde{ L_{\ell} \left( s \right) }   \hat{x}_s 
- \Re \left(  \langle \hat{x}_s, \widetilde{ L_{\ell} \left( s \right) }  \hat{x}_s \rangle \right)  \hat{x}_s
\right) } dB_s^k
\\
\fl
&  \, \, + \int_0^t\left( 
\overline{ \widetilde{ G\left( s \right) } \hat{x}_s } 
+
\sum_{ \ell =1}^{ \infty} \overline{ \left( \Re \left( \langle \hat{x}_s, \widetilde{ L_{\ell} \left( s \right) } \hat{x}_s \rangle  \right) 
\widetilde{ L_{\ell} \left( s \right) } \hat{x}_s
-\frac{1}{2} \Re^2 \left(\langle \hat{x}_s, \widetilde{ L_{\ell} \left( s \right) } \hat{x}_s \rangle  \right) \hat{x}_s \right) }
 \right)ds 
 .
\end{eqnarray*}
Take $u, v \in \mathfrak{h} $.
Let 
$ F : L^2 \left( \mathbb{N} ;  \mathfrak{h} \right) \times  L^2 \left( \mathbb{N} ;  \mathfrak{h} \right)
\rightarrow \mathbb{C} $
be given by
\begin{equation}
 \label{eq:Def_F}
 F \left( x , y \right) = \sum_{n=1}^{\infty} \langle \overline{x \left( n \right)} , u \rangle \langle v, y \left( n \right)  \rangle 
\qquad \qquad 
\quad \mathrm{for \ all \;}
x, y \in L^2 \left( \mathbb{N} ;  \mathfrak{h} \right) .
\end{equation}
Then, $F$ is a bilinear form on $L^2 \left( \mathbb{N} ;  \mathfrak{h} \right) $.
Applying the It\^{o} formula  to $F \left( \overline{\hat{x}_t} , \hat{x}_t \right)$ we obtain
\begin{eqnarray}
\label{eq:ItoE}
\fl
 & F \left( \overline{\hat{x}_t} , \hat{x}_t \right) 
  = 
 F \left( \overline{\hat{x}_0} , \hat{x}_0 \right)
 +
 \int_0^t  \xi \left( s \right) ds
\\
\fl
\nonumber
& \quad  +
\sum_{ \ell =1}^{ \infty}
\int_0^t  \left(
 \sum_{n=1}^{\infty} \langle \sigma_{\ell}  \left( s , \hat{x}_s \left( n \right) \right)  , u \rangle \langle v,  \hat{x}_s \left( n \right) \rangle 
 +
 \sum_{n=1}^{\infty} \langle \hat{x}_s \left( n \right) , u \rangle \langle v,  \sigma_{\ell}  \left( s , \hat{x}_s \left( n \right) \right)  \rangle 
\right) dB_s^{\ell}  ,
\end{eqnarray}
where 
$
\sigma_{\ell}  \left( s , z \right) 
=
L_{\ell} \left( s \right) z - \wp_{\ell} \left( s \right) z
$ 
with
$
\wp_{\ell} \left( s  \right)
=
\sum_{n=1}^{\infty} \Re \left( \langle \hat{x}_s \left( n \right) ,  L_{\ell} \left( s \right)  \hat{x}_s \left( n \right) \rangle \right) 
$,
\begin{eqnarray*}
 \xi \left( s \right)
 & =
\sum_{n=1}^{\infty} \langle 
G\left( s \right) \hat{x}_s \left( n \right) 
+  
g \left( s , \hat{x}_s \left( n \right)  \right)  , u \rangle \langle v, \hat{x}_s \left( n \right)  \rangle
\\
& \quad + 
\sum_{n=1}^{\infty} \langle \hat{x}_s \left( n \right),  u \rangle \langle v, 
G\left( s \right) \hat{x}_s \left( n \right) 
+  
g \left( s , \hat{x}_s \left( n \right)  \right) \rangle
\\
& \quad
+ 
\sum_{ \ell =1}^{ \infty}  \sum_{n=1}^{\infty} \langle \sigma_{\ell}  \left( s , \hat{x}_s \left( n \right) \right)   , u \rangle \langle v,  \sigma_{\ell}  \left( s , \hat{x}_s \left( n \right) \right)   \rangle 
\end{eqnarray*}
and
$
g \left( s , z \right) 
=
\sum_{ \ell =1}^{ \infty} \left( \wp_{\ell} \left( s \right)  L_{\ell} \left( s \right)  z
-\frac{1}{2} \wp_{\ell} \left( s  \right)^2  z \right) 
$.

From the definitions of $g$ and  $\sigma_{\ell}$ we have
\begin{eqnarray*}
\fl
& \langle g \left( s , \hat{x}_s \left( n \right)  \right)  , u \rangle \langle v, \hat{x}_s \left( n \right)  \rangle
+
\langle \hat{x}_s \left( n \right),  u \rangle \langle v, g \left( s , \hat{x}_s \left( n \right)  \right) \rangle
\\
\fl
& =
\sum_{ \ell =1}^{ \infty}   \wp_{\ell} \left( s \right)  \langle L_{\ell} \left( s \right)   \hat{x}_s \left( n \right) , u \rangle \langle v, \hat{x}_s \left( n \right)  \rangle
+
 \sum_{ \ell =1}^{ \infty}  \wp_{\ell} \left( s \right) \langle \hat{x}_s \left( n \right),  u \rangle \langle v,   L_{\ell} \left( s \right)   \hat{x}_s \left( n \right)    \rangle 
 \\
 \fl
& \quad
-  \sum_{ \ell =1}^{ \infty}  \wp_{\ell} \left( s  \right)^2 \langle \hat{x}_s \left( n \right),  u \rangle \langle v,     \hat{x}_s \left( n \right)    \rangle ,
\end{eqnarray*}
and 
\begin{eqnarray*}
\fl
&  \langle \sigma_{\ell}  \left( s , \hat{x}_s \left( n \right) \right)   , u \rangle \langle v,  \sigma_{\ell}  \left( s , \hat{x}_s \left( n \right) \right)   \rangle
%
\\
\fl
& =
\langle  L_{\ell} \left( s \right) \hat{x}_s \left( n \right)   , u \rangle \langle v,   L_{\ell} \left( s \right) \hat{x}_s \left( n \right)  \rangle
-  \wp_{\ell} \left( s \right) \langle  L_{\ell} \left( s \right) \hat{x}_s \left( n \right)   , u \rangle  \langle v,  \hat{x}_s \left( n \right)   \rangle 
\\
\fl
& \quad + 
 \wp_{\ell} \left( s \right)^2 \langle \hat{x}_s \left( n \right)   , u \rangle \langle v,  \hat{x}_s \left( n \right)   \rangle 
-  \wp_{\ell} \left( s \right)  \langle \hat{x}_s \left( n \right)   , u \rangle \langle v,   L_{\ell} \left( s \right) \hat{x}_s \left( n \right)  \rangle 
.
\end{eqnarray*}
Adding the two preceding expressions gives 
\begin{eqnarray*}
 \xi \left( s \right)
 & =
\sum_{n=1}^{\infty} \langle  G\left( s \right) \hat{x}_s \left( n \right) , u \rangle \langle v, \hat{x}_s \left( n \right)  \rangle
+ 
\sum_{n=1}^{\infty} \langle \hat{x}_s \left( n \right),  u \rangle \langle v,  G\left( s \right) \hat{x}_s \left( n \right) \rangle
\\
& \quad + 
\sum_{ \ell =1}^{ \infty} \sum_{n=1}^{\infty} 
\langle  L_{\ell} \left( s \right) \hat{x}_s \left( n \right)   , u \rangle \langle v,   L_{\ell} \left( s \right) \hat{x}_s \left( n \right)  \rangle .
\end{eqnarray*}

Combining Lemma  \ref{le:A_rho_B} with Hypothesis H2.1 we get that 
the unique bounded extension of $ G\left( s \right) \varrho_s   $ is equal to
$ \sum_{n=1}^{\infty} \left\vert  G\left( s \right) \hat{x}_s \left( n \right)  \rangle\langle \hat{x}_s \left( n \right) \right\vert$.
Hence,
\[
\langle v,   G\left( s \right) \varrho_s \, u   \rangle
=
\sum_{n=1}^{\infty} \langle \hat{x}_s \left( n \right),  u \rangle \langle v,  G\left( s \right) \hat{x}_s \left( n \right) \rangle 
\]
and 
\begin{equation}
 \label{eq:TrA_G_rho}
 \Tr \left( A \,  G\left( s \right) \varrho_s  \right)
=
 \sum_{n=1}^{\infty}  \langle \hat{x}_s \left( n \right) ,   A \, G\left( s \right) \hat{x}_s \left( n \right)  \rangle 
 \qquad
 \mathrm{for \ all \;} A \in  \mathfrak{L} \left( \mathfrak{h}\right) .
\end{equation}
Using (\ref{eq:TrA_G_rho}) and Hypothesis \ref{hyp:L-G-C-def} 
we obtain that  
$ \Tr \left( A \,  G\left( s \right) \varrho_s  \right) :  \left[ 0 , t \right]  \times \Omega \rightarrow \mathbb{C} $
is $\mathcal{B}\left(  \left[ 0 , t \right] \right) \otimes \mathfrak{F}_{t}$-measurable
for all $A \in  \mathfrak{L} \left( \mathfrak{h}\right) $.
Therefore,
$s \mapsto G\left( s \right) \varrho_s$ is a predictable  stochastic process with values in both spaces
$ \mathfrak{L}_{1}\left( \mathfrak{h}\right)$ and $ \mathfrak{L}_{2}\left( \mathfrak{h}\right)$.
Combining Hypothesis H2.1, $ \left\Vert   \hat{x}_s \right\Vert = 1$ and the Cauchy-Schwarz inequality yields
\begin{eqnarray*}
\fl
 \left\Vert  G\left( s \right) \varrho_s  \right\Vert_{  \mathfrak{L}_{1}\left( \mathfrak{h}\right) }
& \leq
\sum_{n=1}^{\infty} \left\Vert  \, \left\vert  G\left( s \right) \hat{x}_s \left( n \right)  \rangle\langle \hat{x}_s \left( n \right) \right\vert  \right\Vert_{  \mathfrak{L}_{1}\left( \mathfrak{h}\right) } 
 =
\sum_{n=1}^{\infty} \left\Vert   G\left( s \right) \hat{x}_s \left( n \right)   \right\Vert \left\Vert   \hat{x}_s \left( n \right)  \right\Vert
\\
\fl
& \leq
\left( \sum_{n=1}^{\infty} \left\Vert   G\left( s \right) \hat{x}_s \left( n \right)   \right\Vert^2 \right)^{1/2}
\leq
K \left( t \right)  \left(  1 +  \left\Vert \widetilde{ C } \hat{x}_s    \right\Vert ^2  \right).
\end{eqnarray*}

The operator $ G\left( s \right) $ is densely defined,
because  $\mathcal{D}\left(C \right)$ is a subset of  $\mathcal{D}\left( G  \left( s \right) \right)$.
According to Hypothesis \ref{hyp:Closable} we have that $ G\left( s \right) $ is closable.
Therefore,
$ G\left( s \right)^{\ast}  $ is a densely defined closed operator in $\mathfrak{h}$,
and $ G\left( s \right)^{ \ast \ast }$ is equal to the closure of  $ G\left( s \right)$.
Hence, $ \mathcal{D}\left(C \right) \subset \left( G\left( s \right)^{\ast}  \right)^{\ast} $.
Applying Lemma \ref{le:A_rho_B} we obtain that the unique bounded extension of $ \varrho_s   G\left( s \right)^{\ast}  $
is equal to  
$ \sum_{n=1}^{\infty} \left\vert \hat{x}_s \left( n \right)  \rangle\langle \left( G\left( s \right)^{\ast}  \right)^{\ast}  \hat{x}_s \left( n \right) \right\vert$,
and so, it becomes 
$ \sum_{n=1}^{\infty} \left\vert \hat{x}_s \left( n \right)  \rangle\langle G\left( s \right) \hat{x}_s \left( n \right) \right\vert$.
This gives
\[
\langle v, \varrho_s   G\left( s \right)^{\ast}  u   \rangle
= 
\sum_{n=1}^{\infty} \langle  G\left( s \right) \hat{x}_s \left( n \right) , u \rangle \langle v, \hat{x}_s \left( n \right)  \rangle .
\]
and
\[
\left\Vert  \varrho_s   G\left( s \right)^{\ast} \right\Vert_{  \mathfrak{L}_{1}\left( \mathfrak{h}\right) }
\leq
\sum_{n=1}^{\infty} \left\Vert   \hat{x}_s \left( n \right)  \right\Vert \left\Vert   G\left( s \right) \hat{x}_s \left( n \right)   \right\Vert 
\leq
K \left( t \right)  \left(  1 +  \left\Vert \widetilde{ C } \hat{x}_s \left( n \right)    \right\Vert ^2  \right).
\]
Since
$\Tr \left( A   \, \varrho_s   G\left( s \right)^{\ast}  \right)
=
 \sum_{n=1}^{\infty}  \langle G\left( s \right) \hat{x}_s \left( n \right) ,A \hat{x}_s \left( n \right)  \rangle
$
for all $ A \in \mathfrak{L} \left( \mathfrak{h}\right)$,
using Hypothesis \ref{hyp:L-G-C-def}  we obtain that
$ s \mapsto \varrho_s   G\left( s \right)^{\ast}  $ is 
 a predictable  stochastic process with values in both spaces
$ \mathfrak{L}_{1}\left( \mathfrak{h}\right)$ and $ \mathfrak{L}_{2}\left( \mathfrak{h}\right)$.

Analysis similar to that in the two previous paragraphs shows that 
the unique bounded extension of $ L_{\ell} \left( s \right) \varrho_s   L_{\ell} \left( s \right) ^{\ast}  $
is equal to  the linear operator
$ \sum_{n=1}^{\infty} \left\vert  L_{\ell} \left( s \right) \hat{x}_s \left( n \right)  \rangle\langle  L_{\ell} \left( s \right) \hat{x}_s \left( n \right) \right\vert$.
Hence,
\[
\langle v,  L_{\ell} \left( s \right) \varrho_s   L_{\ell} \left( s \right) ^{\ast}  u   \rangle
= 
\sum_{n=1}^{\infty} \langle  L_{\ell} \left( s \right)  \hat{x}_s \left( n \right) , u \rangle \langle v,   L_{\ell} \left( s \right)  \hat{x}_s \left( n \right)  \rangle 
\]
and
$ L_{\ell} \left( s \right) \varrho_s   L_{\ell} \left( s \right) ^{\ast}  $ 
is progressively measurable as 
a  stochastic process with values in both spaces
$ \mathfrak{L}_{1}\left( \mathfrak{h}\right)$ and $ \mathfrak{L}_{2}\left( \mathfrak{h}\right)$.
Then,
\[ 
\xi \left( s \right) 
 = 
 \langle v, \varrho_s   G\left( s \right)^{\ast}  u   \rangle + \langle v,   G\left( s \right) \varrho_s u   \rangle
 + \sum_{ \ell =1}^{ \infty} \langle v,  L_{\ell} \left( s \right) \varrho_s   L_{\ell} \left( s \right) ^{\ast}  u   \rangle 
.
\]

From Hypothesis H2.2 we conclude 
\[
\fl
 \sum_{ \ell =1}^{ \infty} \left\Vert  L_{\ell} \left( s \right) \hat{x}_s \left( n \right)  \right\Vert^2
 =
 -2  \Re\left\langle  \hat{x}_s \left( n \right) , G \left( s \right)  \hat{x}_s \left( n \right)\right\rangle
 \leq 
 2 \left\Vert \hat{x}_s \left( n \right) \right\Vert \left\Vert  G \left( s \right)  \hat{x}_s \left( n \right) \right\Vert .
\]
Since
\[
\sum_{ \ell =1}^{ \infty}  \sum_{n=1}^{\infty} 
\left\Vert \, 
\left\vert  L_{\ell} \left( s \right) \hat{x}_s \left( n \right)  \rangle\langle  L_{\ell} \left( s \right) \hat{x}_s \left( n \right) \right\vert
  \, \right\Vert_{  \mathfrak{L}_{1}\left( \mathfrak{h}\right) }
 =
 \sum_{ \ell =1}^{ \infty}  \sum_{n=1}^{\infty} 
\left\Vert  L_{\ell} \left( s \right) \hat{x}_s \left( n \right)  \right\Vert^2 ,
\]
combining Hypothesis H2.1 with  $ \left\Vert \hat{x}_s  \right\Vert = 1 $ yields
\begin{equation}
\label{eq:Des-L}
\sum_{ \ell =1}^{ \infty}  \sum_{n=1}^{\infty} 
\left\Vert \, 
\left\vert  L_{\ell} \left( s \right) \hat{x}_s \left( n \right)  \rangle\langle  L_{\ell} \left( s \right) \hat{x}_s \left( n \right) \right\vert
  \, \right\Vert_{  \mathfrak{L}_{1}\left( \mathfrak{h}\right) }
 \leq K \left( s \right) \left( 1 +  \left\Vert  \widetilde{C} \, \hat{x}_s   \right\Vert^2 \right).
\end{equation}
Hence,
$\sum_{ \ell =1}^{ N}  L_{\ell} \left( s \right) \varrho_s   L_{\ell} \left( s \right) ^{\ast} $
converges in $  \mathfrak{L}_{1}\left( \mathfrak{h}\right) $ to 
$\sum_{ \ell =1}^{ \infty}  L_{\ell} \left( s \right) \varrho_s   L_{\ell} \left( s \right) ^{\ast} $
as $N \rightarrow + \infty$.
Therefore,
\[
\left\Vert  \sum_{ \ell =1}^{ \infty}  L_{\ell} \left( s \right) \varrho_s   L_{\ell} \left( s \right) ^{\ast} 
\right\Vert_{  \mathfrak{L}_{1}\left( \mathfrak{h}\right) }
\leq K \left( s \right) \left( 1 +  \left\Vert  \widetilde{C} \, \hat{x}_s   \right\Vert^2 \right) ,
\]
\begin{equation*}
 \xi \left( s \right) 
 = 
 \langle v, \varrho_s   G\left( s \right)^{\ast}  u   \rangle + \langle v,   G\left( s \right) \varrho_s u   \rangle
 + \langle v, \left(   \sum_{ \ell =1}^{ \infty}  L_{\ell} \left( s \right) \varrho_s   L_{\ell} \left( s \right) ^{\ast} \right) u   \rangle ,
\end{equation*}
and $ s \mapsto \sum_{ \ell =1}^{ \infty} L_{\ell} \left( s \right) \varrho_s   L_{\ell} \left( s \right) ^{\ast} $
 is a predictable  stochastic process with values in both spaces
$ \mathfrak{L}_{1}\left( \mathfrak{h}\right)$ and $ \mathfrak{L}_{2}\left( \mathfrak{h}\right)$.
Since $\sup_{s\in\left[ 0, T\right] }\mathbb{E} \left( \left\Vert  \widetilde{C} \hat{x}_{s}\right\Vert ^{2} \right)<+ \infty$,
\[
\int_0^T \left\Vert  
 \varrho_s   G\left( s \right)^{\ast} + G\left( s \right) \varrho_s 
 +  \sum_{ \ell =1}^{ \infty} L_{\ell} \left( s \right) \varrho_s   L_{\ell} \left( s \right) ^{\ast}
\right\Vert_{  \mathfrak{L}_{1}\left( \mathfrak{h}\right) } ds
< + \infty 
\qquad \mathbb{P} \mathrm{-}a.s.
\]
Then, for all $ t \in \left[ 0, T \right]$ we have 
\begin{equation}
 \label{eq:ItoE_T1}
 \fl
 \int_{0}^t \xi \left( s \right) ds 
=
\langle v,  \int_{0}^t 
\left(
\varrho_s   G\left( s \right)^{\ast}  +  G\left( s \right) \varrho_s  +  \sum_{ \ell =1}^{ \infty}  L_{\ell} \left( s \right) \varrho_s   L_{\ell} \left( s \right) ^{\ast}
 \right) ds \, u   \rangle, 
\end{equation}
where
the integral is understood as a Bochner integral  in both $ \mathfrak{L}_{1} \left( \mathfrak{h}\right)$ and 
$ \mathfrak{L}_{2} \left( \mathfrak{h}\right)$.

Next, we examine  the stochastic integrals of (\ref{eq:ItoE}).
According to the definition of  $\sigma_{\ell}$ we have 
\begin{eqnarray*}
&
\sum_{n=1}^{\infty} \langle \sigma_{\ell}  \left( s , \hat{x}_s \left( n \right) \right)  , u \rangle \langle v,  \hat{x}_s \left( n \right) \rangle 
 +
 \sum_{n=1}^{\infty} \langle \hat{x}_s \left( n \right) , u \rangle \langle v,  \sigma_{\ell}  \left( s , \hat{x}_s \left( n \right) \right)  \rangle 
 \\
& =
\sum_{n=1}^{\infty} \langle L_{\ell} \left( s \right)  \hat{x}_s \left( n \right) , u \rangle \langle v,  \hat{x}_s \left( n \right) \rangle 
+
\sum_{n=1}^{\infty}  \langle \hat{x}_s \left( n \right) , u \rangle \langle v,  L_{\ell} \left( s \right)  \hat{x}_s \left( n \right)   \rangle 
\\
& \quad 
- 
2 \wp_{\ell} \left( s  \right)  \sum_{n=1}^{\infty}  \langle    \hat{x}_s \left( n \right) , u \rangle \langle v,  \hat{x}_s \left( n \right) \rangle .
\end{eqnarray*}
By employing H2.1 and H2.2 in a similar manner to (\ref{eq:Des-L})
and making use of Lemma  \ref{le:A_rho_B},
we obtain that the unique bounded extension of $ L_{\ell} \left( s \right) \varrho_s   $ is equal to
$ \sum_{n=1}^{\infty} \left\vert L_{\ell} \left( s \right) \hat{x}_s \left( n \right)  \rangle\langle \hat{x}_s \left( n \right) \right\vert$,
and consequently 
\[
\langle v,   L_{\ell} \left( s \right) \varrho_s \, u   \rangle
=
\sum_{n=1}^{\infty} \langle \hat{x}_s \left( n \right),  u \rangle \langle v,  L_{\ell} \left( s \right) \hat{x}_s \left( n \right) \rangle .
\]
Applying Lemma \ref{le:L1} we get
\[
tr \left(  L_{\ell} \left( s \right) \varrho_s \right)
=
\sum_{n=1}^{\infty} \langle \hat{x}_s \left( n \right), L_{\ell} \left( s \right) \hat{x}_s \left( n \right) \rangle ,
\]
and so
$ \wp_{\ell} \left( s  \right) = \Re \left( tr \left(  L_{\ell} \left( s \right) \varrho_s \right) \right)$.
Since $  L_{\ell} \left( s \right) $ is densely defined and  closable,
$ L_{\ell} \left( s \right)^{\ast}  $ is  densely defined closed operator
and $  L_{\ell} \left( s \right)^{\ast  \ast}$ is identical to the closure of  $  L_{\ell}\left( s \right)$.
Applying Lemma \ref{le:A_rho_B} with $A = I$ and $B = L_{\ell}\left( s \right)^{\ast} $,
where $I$ is the identity operator in $\mathfrak{h}$, 
we deduce that the unique bounded extension of $ \varrho_s  \,  L_{\ell}\left( s \right)^{\ast}  $
is equal to   
$ \sum_{n=1}^{\infty} \left\vert \hat{x}_s \left( n \right)  \rangle\langle  L_{\ell}\left( s \right) \hat{x}_s \left( n \right) \right\vert$.
Therefore,
\[
\langle v, \varrho_s  \,  L_{\ell}\left( s \right)^{\ast}  u   \rangle
= 
\sum_{n=1}^{\infty} \langle   L_{\ell} \left( s \right) \hat{x}_s \left( n \right) , u \rangle \langle v, \hat{x}_s \left( n \right)  \rangle .
\]
We have obtained that
\begin{eqnarray}
 \nonumber
&
\sum_{n=1}^{\infty} \langle \sigma_{\ell}  \left( s , \hat{x}_s \left( n \right) \right)  , u \rangle \langle v,  \hat{x}_s \left( n \right) \rangle 
 +
 \sum_{n=1}^{\infty} \langle \hat{x}_s \left( n \right) , u \rangle \langle v,  \sigma_{\ell}  \left( s , \hat{x}_s \left( n \right) \right)  \rangle 
 \\
 \label{eq:ItoE_T2}
& =
\langle v, \varrho_s  \,  L_{\ell}\left( s \right)^{\ast}  u   \rangle 
+
\langle v,   L_{\ell} \left( s \right) \varrho_s \, u   \rangle 
- 
2  \Re \left( tr \left(  L_{\ell} \left( s \right) \varrho_s \right) \right) \langle v,   \varrho_s \, u   \rangle  .
\end{eqnarray}
For all $  A \in \mathfrak{L} \left( \mathfrak{h}\right) $,
\begin{eqnarray*}
&
\Tr \left(
A \left(
L_{\ell} \left( s \right) \varrho_s + \varrho_s  \,  L_{\ell}\left( s \right)^{\ast}
- 2  \Re \left( tr \left(  L_{\ell} \left( s \right) \varrho_s \right) \right)    \varrho_s
\right) \right)
\\
&
=
 \sum_{n=1}^{\infty} \langle \hat{x}_s \left( n \right) ,  A \, L_{\ell} \left( s \right) \hat{x}_s \left( n \right)  \rangle
 +
 \sum_{n=1}^{\infty} \langle  L_{\ell}\left( s \right) \hat{x}_s \left( n \right) ,  A \, \hat{x}_s \left( n \right)  \rangle  
 \\
 & \quad
 - 
 2 \Re \left( 
 \sum_{n=1}^{\infty} \langle \hat{x}_s \left( n \right), L_{\ell} \left( s \right) \hat{x}_s \left( n \right) \rangle 
 \right)
  \sum_{n=1}^{\infty} \langle  \hat{x}_s \left( n \right) ,  A \, \hat{x}_s \left( n \right)  \rangle  ,
\end{eqnarray*}
and so 
using  Hypothesis \ref{hyp:L-G-C-def}  we deduce that 
$s \mapsto L_{\ell} \left( s \right) \varrho_s + \varrho_s  \,  L_{\ell}\left( s \right)^{\ast}
- 2  \Re \left( tr \left(  L_{\ell} \left( s \right) \varrho_s \right) \right)    \varrho_s $
is  a predictable  stochastic process with values in both spaces
$ \mathfrak{L}_{1}\left( \mathfrak{h}\right)$ and $ \mathfrak{L}_{2}\left( \mathfrak{h}\right)$.

Since 
$ \varrho_s  \,  L_{\ell}\left( s \right)^{\ast}  = \sum_{n=1}^{\infty} \left\vert \hat{x}_s \left( n \right)  \rangle\langle  L_{\ell}\left( s \right) \hat{x}_s \left( n \right) \right\vert$,
\begin{eqnarray*}
 \sum_{\ell =1}^{\infty}  \left\Vert  \varrho_s  \,  L_{\ell}\left( s \right)^{\ast}  \right\Vert^2_{  \mathfrak{L}_{1}\left( \mathfrak{h}\right) }
& \leq
\sum_{\ell =1}^{\infty}  \left(
\sum_{n=1}^{\infty}   \left\Vert  \,   \left\vert \hat{x}_s \left( n \right)  \rangle\langle  L_{\ell}\left( s \right) \hat{x}_s \left( n \right) \right\vert \,  \right\Vert_{  \mathfrak{L}_{1}\left( \mathfrak{h}\right) } \right)^2
\\
& =
\sum_{\ell =1}^{\infty}  \left(
\sum_{n=1}^{\infty}   \left\Vert  \hat{x}_s \left( n \right)  \right\Vert   \left\Vert  L_{\ell}\left( s \right) \hat{x}_s \left( n \right) \right\Vert
\right)^2 .
\end{eqnarray*}
Similarly,
the property
$ L_{\ell} \left( s \right) \varrho_s    = \sum_{n=1}^{\infty} \left\vert L_{\ell} \left( s \right) \hat{x}_s \left( n \right)  \rangle\langle \hat{x}_s \left( n \right) \right\vert$
leads to
\begin{eqnarray*}
 \sum_{\ell =1}^{\infty}  \left\Vert   L_{\ell} \left( s \right) \varrho_s  \right\Vert^2_{  \mathfrak{L}_{1}\left( \mathfrak{h}\right) }
& \leq
\sum_{\ell =1}^{\infty}  \left(
\sum_{n=1}^{\infty}   \left\Vert  \,   \left\vert   L_{\ell}\left( s \right) \hat{x}_s \left( n \right)  \rangle\langle  \hat{x}_s \left( n \right) \right\vert \,  \right\Vert_{  \mathfrak{L}_{1}\left( \mathfrak{h}\right) } \right)^2
\\
& =
\sum_{\ell =1}^{\infty}  \left(
\sum_{n=1}^{\infty}   \left\Vert    L_{\ell}\left( s \right) \hat{x}_s \left( n \right)  \right\Vert   \left\Vert  \hat{x}_s \left( n \right) \right\Vert
\right)^2 .
\end{eqnarray*}
Combining $tr \left(  \varrho_s  \right) = 1$ with 
 $ tr \left(  L_{\ell} \left( s \right) \varrho_s \right) =  \sum_{n=1}^{\infty} \langle \hat{x}_s \left( n \right) , L_{\ell} \left( s \right) \hat{x}_s \left( n \right)  \rangle$
 we get
\begin{eqnarray*}
 \sum_{\ell =1}^{\infty} 
 \left\Vert \Re \left( tr \left(  L_{\ell} \left( s \right) \varrho_s \right) \right) \varrho_s  \right\Vert^2_{  \mathfrak{L}_{1}\left( \mathfrak{h}\right) }
 & =
 \sum_{\ell =1}^{\infty} 
 \left\vert \Re \left( tr \left(  L_{\ell} \left( s \right) \varrho_s \right) \right)  \right\vert^2 
 \left\Vert \varrho_s  \right\Vert^2_{  \mathfrak{L}_{1}\left( \mathfrak{h}\right) }
 \\ & 
 = 
 \sum_{\ell =1}^{\infty} 
 \left\vert \Re \left( tr \left(  L_{\ell} \left( s \right) \varrho_s \right) \right)  \right\vert^2
 \\
 & \leq \sum_{\ell =1}^{\infty} \left( \sum_{n=1}^{\infty} 
 \left\vert \langle \hat{x}_s \left( n \right) , L_{\ell} \left( s \right) \hat{x}_s \left( n \right)  \rangle \right\vert \right)^2 
 \\
 & \leq \sum_{\ell =1}^{\infty} \left( \sum_{n=1}^{\infty} 
 \left\Vert \hat{x}_s \left( n \right)  \right\Vert  \left\Vert L_{\ell} \left( s \right) \hat{x}_s \left( n \right)   \right\Vert \right)^2 .
\end{eqnarray*}

We proceed to estimate 
$\sum_{\ell =1}^{\infty}  \left(
\sum_{n=1}^{\infty}   \left\Vert  \hat{x}_s \left( n \right)  \right\Vert   \left\Vert  L_{\ell}\left( s \right) \hat{x}_s \left( n \right) \right\Vert
\right)^2$.
Since  $ \left\Vert  \hat{x}_s \right\Vert = 1$,
\begin{eqnarray*}
\fl
\left(
\sum_{n=1}^{\infty}   \left\Vert  \hat{x}_s \left( n \right)  \right\Vert   \left\Vert  L_{\ell}\left( s \right) \hat{x}_s \left( n \right) \right\Vert
\right)^2
& \leq
\left( \sum_{n=1}^{\infty}   \left\Vert  \hat{x}_s \left( n \right)  \right\Vert^2 \right)
\left( \sum_{n=1}^{\infty}   \left\Vert  L_{\ell}\left( s \right) \hat{x}_s \left( n \right) \right\Vert^2  \right)
 \\ 
 \fl
 & 
=
 \sum_{n=1}^{\infty}   \left\Vert  L_{\ell}\left( s \right) \hat{x}_s \left( n \right) \right\Vert^2  .
\end{eqnarray*}
According to Hypothesis H2.2 we have that
\begin{eqnarray*}
 \sum_{n=1}^{\infty} \sum_{\ell =1}^{\infty}   \left\Vert  L_{\ell}\left( s \right) \hat{x}_s \left( n \right) \right\Vert^2 
 &  =
- \sum_{n=1}^{\infty}  2 \, \Re\left\langle   \hat{x}_s \left( n \right)  , G \left( s \right)  \hat{x}_s \left( n \right) \right\rangle 
 \\ & 
\leq 
\sum_{n=1}^{\infty}  2 \, \left\Vert  \hat{x}_s \left( n \right) \right\Vert \, \left\Vert   G \left( s \right) \hat{x}_s \left( n \right) \right\Vert .
\end{eqnarray*}
Using Hypothesis H2.1, together with  $ \left\Vert  \hat{x}_s \right\Vert = 1$, we get
\begin{eqnarray*}
\fl
\sum_{n=1}^{\infty} \sum_{\ell =1}^{\infty}   \left\Vert  L_{\ell}\left( s \right) \hat{x}_s \left( n \right) \right\Vert^2 
 & \leq K \left( s \right) 
 \sum_{n=1}^{\infty}  \left\Vert  \hat{x}_s \left( n \right)  \right\Vert  \sqrt{
 \left\Vert  \hat{x}_s \left( n \right)  \right\Vert^2 + \left\Vert  C \hat{x}_s \left( n \right)  \right\Vert^2 }
 \\
 & \leq K \left( s \right) 
 \sum_{n=1}^{\infty}  \left\Vert  \hat{x}_s \left( n \right)  \right\Vert  \left(
 \left\Vert  \hat{x}_s \left( n \right)  \right\Vert + \left\Vert  C \hat{x}_s \left( n \right)  \right\Vert \right)
 \\
 & =  K \left( s \right) 
 \sum_{n=1}^{\infty}    \left( \left\Vert  \hat{x}_s \left( n \right)  \right\Vert^2 
 + \left\Vert  \hat{x}_s \left( n \right)  \right\Vert \left\Vert  C \hat{x}_s \left( n \right)  \right\Vert \right)
 \\
 & \leq
  K \left( s \right)  \left( 1 + 
  \left(  \sum_{n=1}^{\infty}   \left\Vert  \hat{x}_s \left( n \right)  \right\Vert^2   \right)^{1/2}
  \left(  \sum_{n=1}^{\infty}  \left\Vert  C \hat{x}_s \left( n \right)  \right\Vert^2   \right)^{1/2}
   \right) 
   \\
 & =
  K \left( s \right)  \left( 1 + \left( \sum_{n=1}^{\infty}  \left\Vert  C \hat{x}_s \left( n \right)  \right\Vert^2   \right)^{1/2}  \right) .
\end{eqnarray*}
Therefore,
\begin{eqnarray*}
\fl
 \mathbb{E} \left( 
 \sum_{\ell =1}^{\infty} 
\left(\sum_{n=1}^{\infty}   \left\Vert  \hat{x}_s \left( n \right)  \right\Vert   \left\Vert  L_{\ell}\left( s \right) \hat{x}_s \left( n \right) \right\Vert
\right)^2 \right)
&
\leq 
K \left( s \right)  \left( 1 +  \mathbb{E}  \left\Vert  \widetilde{C} \hat{x}_s   \right\Vert  \right) 
\\
&
\leq 
K \left( s \right)  \left( 1 +  \sqrt{\mathbb{E}  \left(\left\Vert  \widetilde{C} \hat{x}_s   \right\Vert^2 \right) }\right) .
\end{eqnarray*}

Let $T> 0$.
Since 
\begin{eqnarray*}
& \sum_{ \ell =1}^{ \infty} 
\mathbb{E} \left(   \left\Vert  
 \varrho_s  \,  L_{\ell}\left( s \right)^{\ast}  +   L_{\ell} \left( s \right) \varrho_s 
 - 2  \Re \left( tr \left(  L_{\ell} \left( s \right) \varrho_s \right) \right) \varrho_s 
 \right\Vert^2_{  \mathfrak{L}_{1}\left( \mathfrak{h}\right) } \right)
 \\
& \leq
 3 \mathbb{E} \left(   \sum_{\ell =1}^{\infty}  \left\Vert  \varrho_s  \,  L_{\ell}\left( s \right)^{\ast}  \right\Vert^2_{  \mathfrak{L}_{1}\left( \mathfrak{h}\right) } \right)
 +  3 \mathbb{E} \left(  \sum_{\ell =1}^{\infty}  \left\Vert   L_{\ell} \left( s \right) \varrho_s  \right\Vert^2_{  \mathfrak{L}_{1}\left( \mathfrak{h}\right) } \right)
 \\
 & \quad
 + 6 \mathbb{E} \left(   \sum_{\ell =1}^{\infty} 
 \left\Vert \Re \left( tr \left(  L_{\ell} \left( s \right) \varrho_s \right) \right) \varrho_s  \right\Vert^2_{  \mathfrak{L}_{1}\left( \mathfrak{h}\right) } \right) ,
\end{eqnarray*}
\begin{eqnarray*}
& \sum_{ \ell =1}^{ \infty} 
\mathbb{E} \left(   \left\Vert  
 \varrho_s  \,  L_{\ell}\left( s \right)^{\ast}  +   L_{\ell} \left( s \right) \varrho_s 
 - 2  \Re \left( tr \left(  L_{\ell} \left( s \right) \varrho_s \right) \right) \varrho_s 
 \right\Vert^2_{  \mathfrak{L}_{1}\left( \mathfrak{h}\right) } \right)
 \\
 &   \leq 
K \left( s \right)  \left( 1 +  \sqrt{\mathbb{E}  \left(\left\Vert  \widetilde{C} \hat{x}_s   \right\Vert^2 \right) }\right)  .
\end{eqnarray*}
As $\sup_{s\in\left[ 0, T\right] }\mathbb{E} \left( \left\Vert  \widetilde{C} \hat{x}_{s}\right\Vert ^{2} \right)<+ \infty$
we have
\[
\sup_{s\in\left[ 0, T\right] } \sum_{ \ell =1}^{ \infty} 
\mathbb{E} \left(   \left\Vert  
 \varrho_s  \,  L_{\ell}\left( s \right)^{\ast}  +   L_{\ell} \left( s \right) \varrho_s 
 - 2  \Re \left( tr \left(  L_{\ell} \left( s \right) \varrho_s \right) \right) \varrho_s 
 \right\Vert^2_{  \mathfrak{L}_{1}\left( \mathfrak{h}\right) } \right)
< + \infty .
\]
Since 
$ \left\Vert  A \right\Vert^2_{  \mathfrak{L}_{2}\left( \mathfrak{h}\right) } 
\leq 
\left\Vert  A \right\Vert^2_{  \mathfrak{L}_{1}\left( \mathfrak{h}\right) } $
for any $A \in  \mathfrak{L}_{1}\left( \mathfrak{h}\right) $,
\[
\sup_{s\in\left[ 0, T\right] } \sum_{ \ell =1}^{ \infty} 
\mathbb{E} \left(   \left\Vert  
 \varrho_s  \,  L_{\ell}\left( s \right)^{\ast}  +   L_{\ell} \left( s \right) \varrho_s 
 - 2  \Re \left( tr \left(  L_{\ell} \left( s \right) \varrho_s \right) \right) \varrho_s 
 \right\Vert^2_{  \mathfrak{L}_{2}\left( \mathfrak{h}\right) } \right)
< + \infty .
\]
Hence,
$
\left( \sum_{ \ell =1}^{ \infty} \int_0^t  \left( \varrho_s  \,  L_{\ell}\left( s \right)^{\ast}  +   L_{\ell} \left( s \right) \varrho_s 
 - 2  \Re \left( tr \left(  L_{\ell} \left( s \right) \varrho_s \right) \right) \varrho_s \right) 
 dB_s^{\ell} \right)_{t \geq 0 }
$
and
\[
t \mapsto  \int_0^t  \left( \varrho_s  \,  L_{\ell}\left( s \right)^{\ast}  +   L_{\ell} \left( s \right) \varrho_s 
 - 2  \Re \left( tr \left(  L_{\ell} \left( s \right) \varrho_s \right) \right) \varrho_s \right) 
 dB_s^{\ell} 
\]
are continuous square integrable martingales on any bounded interval,
where the integrals are defined as stochastic integrals in  $ \mathfrak{L}_{2} \left( \mathfrak{h}\right)$.
Using  (\ref{eq:ItoE_T2}) gives
\begin{eqnarray}
\fl
\nonumber
 &
\sum_{ \ell =1}^{ \infty}
\int_0^t  \left(
 \sum_{n=1}^{\infty} \langle \sigma_{\ell}  \left( s , \hat{x}_s \left( n \right) \right)  , u \rangle \langle v,  \hat{x}_s \left( n \right) \rangle 
 +
 \sum_{n=1}^{\infty} \langle \hat{x}_s \left( n \right) , u \rangle \langle v,  \sigma_{\ell}  \left( s , \hat{x}_s \left( n \right) \right)  \rangle 
\right) dB_s^{\ell}  
\\
\fl
\label{eq:ItoE_T2R}
& =
 \langle v,  \left(
\sum_{ \ell =1}^{ \infty} \int_0^t  \left( \varrho_s  \,  L_{\ell}\left( s \right)^{\ast}  +   L_{\ell} \left( s \right) \varrho_s 
 - 2  \Re \left( tr \left(  L_{\ell} \left( s \right) \varrho_s \right) \right) \varrho_s \right) 
 dB_s^{\ell} 
 \right)  u \rangle .
\end{eqnarray}

Combining (\ref{eq:ItoE}), (\ref{eq:ItoE_T1}) and (\ref{eq:ItoE_T2R}) yields 
\begin{eqnarray}
 \nonumber
  \varrho_t  
 & = 
  \varrho_0 
 +
 \int_0^t   \left( \varrho_s   G\left( s \right)^{\ast} +  G\left( s \right) \varrho_s 
 + \sum_{ \ell =1}^{ \infty}  L_{\ell} \left( s \right) \varrho_s   L_{\ell} \left( s \right) ^{\ast} \right)  
ds
\\
\label{eq:pE_1}
&  \quad +
\sum_{ \ell =1}^{ \infty}
\int_0^t 
  \left( \varrho_s  \,  L_{\ell}\left( s \right)^{\ast}  +   L_{\ell} \left( s \right) \varrho_s 
 - 2  \Re \left( tr \left(  L_{\ell} \left( s \right) \varrho_s \right) \right) \varrho_s \right) 
 dB_s^{\ell}  .
\end{eqnarray}

According to Lemma \ref{lem:Existence} we have that
$\left( \mathbb{P}, \left( \hat{x}_t \left( n \right)  \right)_{t \in \mathbb{I}}^{n \in \mathbb{N}},\left( B_{t}^{\ell}\right) _{t \in \mathbb{I}}^{\ell\in\mathbb{N}} \right)$ 
is a $C$-solution of the stochastic system (\ref{st:SIWE}) with initial law $\theta $.
Applying Theorem \ref{th:EyU-SIWE} we obtain that 
for any  collection 
$ t_{k \left( 1 \right)} <  t_{k \left( 2 \right)} < \cdots < t_{k \left( N \right)}$,
where $N \in \mathbb{N}$, the distributions of
$\left( \left( \hat{x}_{t\left(k \right)} \left( n \right)  \right)_{1\leq k \leq N}^{n \in \mathbb{N}},\left( B_{t\left(k \right)}^{\ell}\right) _{ 1\leq k \leq N }^{\ell\in\mathbb{N}} \right)$ 
and 
$\left(\left( \psi^n_{t\left(k \right)}  \right)_{1\leq k \leq N }^{n \in \mathbb{N}},
\left( W_{t\left(k \right)}^{\ell}\right) _{1\leq k \leq N}^{\ell\in\mathbb{N}} \right)$
coincide.
Hence,
$\left( \left( \varrho_{t\left(k \right)}  \right)_{1\leq k \leq N},\left( B_{t\left(k \right)}^{\ell}\right) _{ 1\leq k \leq N }^{\ell\in\mathbb{N}} \right)$ 
and 
$\left(\left( \rho_{t\left(k \right)}  \right)_{1\leq k \leq N },
\left( W_{t\left(k \right)}^{\ell}\right) _{1\leq k \leq N}^{\ell\in\mathbb{N}} \right)$
have the same probability distributions.
Therefore,
\[
t \mapsto \sum_{ \ell =1}^{ \infty} \int_0^t  \left( \rho_s  \,  L_{\ell}\left( s \right)^{\ast}  +   L_{\ell} \left( s \right) \rho_s 
 - 2  \Re \left( tr \left(  L_{\ell} \left( s \right) \rho_s \right) \right) \rho_s \right) 
 dB_s^{\ell} 
\]
is a continuous square integrable martingale on any bounded interval,
and using (\ref{eq:pE_1}) gives  (\ref{eq:NSME-e}) (see, e.g., Section 8 of \cite{Ondrejat2004}).
\end{proof}

\section{Conclusions}

We derive a model describing the evolution of quantum systems undergoing continuous measurements
that is based on a non-linear system of stochastic evolution equations in the Hilbert space representing the state space.
We prove that this model is well-defined under general conditions,
which applied to a variety of physical situations as we illustrate in two examples.
The model under study is an alternative to the diffusive stochastic quantum master equation, 
which is also known as Belavkin equation,
and it provides a  physically meaningful solution to the Belavkin equation.   

\ack 
The authors thank the referees for their valuable comments and suggestions on the manuscript.

\section*{References}



\providecommand{\noopsort}[1]{}\providecommand{\singleletter}[1]{#1}%

\end{document}